\documentclass[journal,twoside,web]{ieeecolor}  
\pdfoutput=1
\usepackage{etoolbox}
\makeatletter
\@ifundefined{color@begingroup}%
{\let\color@begingroup\relax
	\let\color@endgroup\relax}{}%
\def\fix@ieeecolor@hbox#1{%
	\hbox{\color@begingroup#1\color@endgroup}}
\patchcmd\@makecaption{\hbox}{\fix@ieeecolor@hbox}{}{\FAILED}
\patchcmd\@makecaption{\hbox}{\fix@ieeecolor@hbox}{}{\FAILED}

\usepackage{diagbox}
\usepackage{generic}
\usepackage{cite}
\usepackage{amsmath,amssymb,amsfonts}
\usepackage{algorithmic}
\usepackage{graphicx}
\usepackage{textcomp}
\usepackage{multirow}
\usepackage{subfigure}
\usepackage{balance}
\usepackage{arydshln}
\usepackage{epstopdf}
\usepackage{bm}
\usepackage{booktabs}
\usepackage[linesnumbered,ruled,vlined]{algorithm2e}
\usepackage{soul}
\usepackage{threeparttable}
\newtheorem{theo}{Theorem}
\newtheorem{coro}{Corollary}
\newtheorem{exam}{Example}
\newtheorem{lemm}{Lemma}
\newtheorem{defi}{Definition}
\newtheorem{rema}{Remark}
\newtheorem{assu}{Assumption}
\newtheorem{prob}{Problem}

\newcommand{\Rmnum}[1]{\expandafter\@slowromancap\romannumeral #1@}

\def\BibTeX{{\rm B\kern-.05em{\sc i\kern-.025em b}\kern-.08em
		T\kern-.1667em\lower.7ex\hbox{E}\kern-.125emX}}
\begin{document}
	
	\title{ Distributed Framework Construction for Affine Formation Control
	}
	
	\author{Huiming Li$^{1}$, Hao Chen$^{1}$, Xiangke Wang$^{1}$~\IEEEmembership{Senior Member,~IEEE, } Zhongkui Li$^{2}$,~\IEEEmembership{Senior Member,~IEEE, } Lincheng Shen$^{1}$
		\thanks{*This work was supported by the National Natural Science Foundation of China under Grant 62303483, U23B2032, U2241214, and 62376280. The corresponding author is Xiangke Wang. }
		\thanks{$^{1}$College of Intelligence Science and Technology, 
			National University of Defense Technology, Changsha 410073,  China.	}
		\thanks{$^{2}$ The State Key Laboratory for Turbulence and Complex Systems, Department of Mechanics and Engineering Science, College of Engineering, Peking University, Beijing 100871, China.}
		\thanks{{\tt Email: huiminglhm@163.com, chenhao09@nudt.edu.cn, xkwang@nudt.edu.cn, zhongkli@pku.edu.cn, lcshen@nudt.edu.cn.}}
	}
	
	\maketitle
	
	\begin{abstract}
		In affine formation control problems, the construction of the framework with universal rigidity and affine localizability is a critical prerequisite, but it has not yet been well addressed, especially when additional agents join the formation or link/agent failures emerge. Motivated by this observation, we investigate the problem of constructing affine frameworks in three scenarios, including vertex addition, edge deletion and vertex deletion. Our approach starts from the original affine formation and uses geometric methods to locally adjust the structure of the weighted graph to describe the topology, so that the modified framework maintains the universal rigidity and affine localizability. Notably, the developed strategies only utilize local measurements and exhibit distributed characteristics, laying the foundation for applications in multi-agent systems. To demonstrate the compatibility with affine formation control proposals, we present a case study on affine formation tracking in a multi-UAV formation, demonstrating the effectiveness of our algorithms in constructing eligible frameworks in aforementioned scenarios. Moreover, a comparative simulation is also conducted to highlight the low time complexity of our distributed algorithm relative to the centralized optimization-based method.
	\end{abstract}
	
	
	\section{Introduction}
	\IEEEPARstart{A}{ffine} formation control problem has recently attracted an increasing attention of scholars, which allows configuration transformations including translation, rotation, scaling, shearing and their combinations \cite{TAC_Lin_2016,TAC_Zhao_2018,TITS_Wang_2023,AST_Li_2025_BearingBased}. Various studies have been conducted to investigate the collaborative control scheme for affine formations with different system dynamics, e.g., linear dynamics \cite{TAC_Zhao_2018,TAC_Lin_2022}, unicycle models \cite{FITEE_Li_2022} and Euler-Lagrange models \cite{CJA_Xu_2019}, or different operating modes, e.g., switching topologies \cite{SJCO_Yang_2021}, event-triggered conditions \cite{TCST_Zhu_2023}, etc. Undoubtedly, the implementation of the above distributed cooperative control schemes heavily relies on the affine \emph{frameworks}, which are described by assigning coordinates in Euclidean space to the topology represented by a weighted graph \cite{TAC_Lin_2016}. A determined framework fixes the direction of information flow within a multi-agent system, which directly influences the convergence of distributed collaborative control schemes. Moreover, in affine formations, the weighted graph used to describe the topology quantifies the strength of interaction among agents, which also affects the convergence speed of the control proposals.  
	Most related studies used pre-designed frameworks, which generally suffer from inevitable changes during the mission execution process~\cite{Chen2022_survivable}. To address the potential various scenarios, the construction of affine frameworks is of great value and practical significance. 
	
	From a mathematical perspective, the problem of designing frameworks for multi-agent systems is usually transformed into the problem of constructing appropriate graph-related matrices (e.g., the bearing-based Laplacian matrix \cite{TCYB_Zhang_2023} and the stress matrix\cite{RAL_Xiao_2022}), which not only reflect the underlying communication among agents but also affect the control gain of the close-loop system. Based on the theoretical analysis in \cite{TAC_Lin_2016,TAC_Zhao_2018}, there are two critical properties, \emph{universal rigidity} and \emph{affine localizability}, that need to be carefully considered to construct affine frameworks due to their direct impact on stabilizability. Up to now, some efforts (e.g. \cite{RAL_Xiao_2022}) have been put in to construct affine frameworks. 
	
	Starting with the rigidity required by an affine framework, a classic method, the Henneberg Construction (HC), is widely applied in many papers. HC was proposed in \cite{HC_1985} to grow the minimally rigid graphs based on the graph structures and the properties of related matrices. Subsequently, lots of studies followed this idea to analyze how the graph rigidity changes after HCs, and how to obtain graphs with the desired characteristics through HCs. Connelly proved the universal rigidity is maintained after HCs \cite{Geometry_Connelly_2005}, and a series of algorithms are proposed for building and reconstructing rigid graphs, not only the classical distance rigidity \cite{TMC_Andrea_2015}, but also the bearing rigidity \cite{IJC_Eren_2012,TCYB_Trinh_2019,TCYB_Zhang_2023}.  In \cite{TCYB_Yang_2019} and \cite{TAC_Yang_2019}, the issue of growing super stable tensegrity frameworks were addressed with the aid of HCs, including vertex addition, edge splitting and framework merging operations. The work offered a numerical method for determining member types and calculating stress matrices, laying a theoretical foundation for applications in tensegrity frameworks. Different from tensegrity frameworks, which rely on the balance of cables, struts, and bars to achieve mechanical stability, affine frameworks focus on maintaining geometric relationships through affine transformations, enabling scalable and flexible formation maneuvers under various scenes. Nevertheless, limited research has been conducted on reconstructing affine frameworks in general positions when they encounter other unexpected events, such as agent exits.
    
    %
	
	Beyond HCs, to meet specific constraints and address the affine framework construction issue in different situations, a normal idea is to establish an optimization problem to obtain a legal topology in which the minimum energy is spent among all the possible networks. In some existing works, the authors studied framework construction problem by using optimization methods to improve resilience \cite{TRO_Kumar_2022,CJA_Feng_2022} and achieve stable coordination \cite{CDC_Dominic_2011,ICRA_Pratik_2020}.  
	In \cite{RAL_Xiao_2022}, Xiao et al. formulated the affine framework construction problem into a centralized mixed integer semi-definite programming (MISDP) problem to obtain the proper stress matrix for affine formation maneuvers. 
    The ready-to-use solver was applied to construct suitable frameworks, which had a significant impact on the computational time of the NP-hard mixed integer programming problem. However, all the researches mentioned above rely on global information (such as global positions required by \cite{RAL_Xiao_2022}) of multi-agent systems and solve the optimization problem in a centralized way. On the one hand, limited by practical situations, global measurements and central agents are not guaranteed in some real robot systems, for example, only bearing measurements are available in \cite{Li_3D_2022}. 
    On the other hand, the centralized optimization algorithm requires powerful computing resources, which adds difficulty to the update of the framework in real-time conditions. 
    Accordingly, beyond existing research, it is a meaningful and challenging issue to construct affine frameworks using local measurements in a distributed way.

	Motivated by these observations, we investigate the distributed strategies to construct affine frameworks to provide  fundamental support for distributed formation control schemes. Different from \cite{RAL_Xiao_2022,TCYB_Yang_2019}, we design the affine framework construction strategies based on HCs by using local measurements in three specific scenarios, namely vertex addition, edge deletion and vertex deletion. 
	Theoretical analysis and simulations are provided to demonstrate that our algorithms preserve not only universal rigidity but also affine localizability to stabilize affine formations. The main contributions of this paper can be summarized as follows.
	\begin{itemize}
		\item[(1)] We propose distributed strategies to construct hierarchical affine frameworks to address various application scenarios, including vertex addition, edge deletion and vertex deletion. Our approaches do not rely on global measurements, making it well-suited for extensive application in multi-agent systems due to their distributed nature.
		
		\item[(2)] We provide comprehensive theoretical analysis on the mathematical conditions of our proposed framework construction strategies to guarantee the characteristics of affine formations, namely universal rigidity and affine localizability. The presented simulations demonstrate the low time complexity and the widespread practicality in multi-agent systems.
		
		
	\end{itemize}
	
	The remainder of this paper is organized as follows. We provide foundational definitions and lemmas in Section~\ref{sec:Preliminaries}. In Section~\ref{sec:MainResult}, we address the key problem and propose distributed algorithms on affine formation constructions. Then, we provide some discussions about our strategies in Section~\ref{sec:Discuss}. The effectiveness is validated by simulations in Section~\ref{sec:Simulation}. The conclusions are summarized in Section~\ref{sec:Conclusion}.
	
	\textbf{Notation :} Throughout the paper, the partial order $\bm{A} \succ 0 $ ($\bm{A} \succeq 0 $) means that the matrix $\bm{A}$ is positive definite (positive semi-definite, i.e., PSD).
	
	\section{Problem Formulation} \label{sec:Preliminaries}
	In this section, we introduce some basic definitions and supportive lemmas, and then establish the affine framework construction problems to be studied.
	
	\subsection{Affine Formation}
    
	We model the interaction topology of an $n$-agent system as a simple undirected graph $\mathcal{G} = \left(\mathcal{V}, \mathcal{E} \right)$, where $\mathcal{V} = \left\{v_1, \cdots ,v_n\right\}$ is the set of vertices, $\mathcal{E} \subseteq \mathcal{V} \times \mathcal{V}$ is the set of edges, and $|\mathcal{E}|=m$. An edge $e_{ij}$ exists if and only if there is a bidirectional interaction between $v_i$ and $v_j$. Hence, the communication neighbor set of vertex $v_i$ is denoted by $\mathcal{N}_i := \left\{v_j \in \mathcal{V} \mid e_{ij} \in \mathcal{E}\right\}$. Let $\bm{p}_i \in \mathbb{R}^d$ be the position of vertex $v_i$, and $\bm{p}:=\left[\bm{p}_1^T,\cdots,\bm{p}_n^T\right]^T $ describes the \emph{configuration} of the multi-agent system. We define a \emph{framework} as a graph associated with $\bm{p}$, i.e., $\left(\mathcal{G},\bm{p}\right)$. The set of points $\left\{\bm{p}_i\right\}_{i=1}^n$ are called \emph{affinely dependent} if there exist scalars $\left\{a_i\right\}_{i=1}^n$ that are not all zero such that $\sum_{i=1}^{n}a_i \bm{p}_i = 0 $ and $\sum_{i=1}^{n}a_i  = 0 $, and \emph{affinely independent} otherwise. A configuration $\bm{p}$ (or a framework $\left(\mathcal{G},\bm{p}\right)$) in $\mathbb{R}^d$ is said to be in \emph{general position} if no subset of the points $\left\{\bm{p}_i\right\}_{i=1}^n$  of cardinality $d+1$ is affinely
	dependent. For example, a set of points in the plane are in general position if no three of them lie on a straight line. Considering the limited perception distance of agents, we define the perceived neighbors of agent $i$ as $\mathcal{N}_i^p:=\{v_j \in \mathcal{V}\mid \|\bm{p}_i - \bm{p}_j\|\le d_{per}\}$, where $d_{per}$ is the perception distance. We assume that all perceptible agents can communicate, i.e., $\mathcal{N}_i^p \subseteq \mathcal{N}_i $. Thus, the perception ability of an agent determines the members of its neighbors.
	
	
	We assign a \emph{stress}, i.e., a set of scalars $\left\{\varpi_{ij} \right\}_{\left(i,j\right) \in \mathcal{E}}$, for all edges in the undirected graph. That is, a weight $\varpi_{ij} $ is assigned for the edge $ e_{ij}$, and $\varpi_{ij}=\varpi_{ji} $. 
	An \emph{equilibrium stress} is established if and only if $\sum_{v_j \in \mathcal{N}_{i}} \varpi_{i j}\left(\bm{p}_{j}-\bm{p}_{i}\right)=\bm{0}$ holds for any vertex $v_i$ in $(\mathcal{G},\bm{p})$. Its matrix form can be written as $\left(\bm{\Omega} \otimes \mathbf{I}_d\right)\bm{p} =\bm{0}$, 
	where the sign $\otimes$ represents the Kronecker product and $\bm{\Omega} \in \mathbb{R}^{n\times n}$ is termed the \emph{stress matrix}, defined as
	\begin{align}
		[\bm{\Omega}]_{i j}=\left\{\begin{array}{ll}
			-\varpi_{i j} & \text { if } i \neq j \text { and } v_j \in \mathcal{N}_{i} \\
			0 & \text { if } i \neq j \text { and } v_j \notin \mathcal{N}_{i} \\
			\sum_{k \in \mathcal{N}_{i}} \varpi_{i k} & \text { if } i=j
		\end{array}\right.\notag
	\end{align}
	The matrix $\mathbf{I}_d$ denotes the identity matrix in $\mathbb{R}^d$. The following lemmas reveal the connection between the stress matrix and universal rigidity.
    \begin{lemm}[\cite{LAA_Alfakih_2011}]\label{lemma:StressUniRigid}
		Let $\left(\mathcal{G},\bm{p}\right)$ be a framework of $n$ vertices in general position in $\mathbb{R}^d$, $d \le
		n - 1$. Then $\left(\mathcal{G},\bm{p}\right)$ is universally rigid if there exists a stress matrix $\bm{\Omega}$ of $\left(\mathcal{G},\bm{p}\right)$ such
		that $\bm{\Omega}$ is PSD and $\operatorname{rank}\left(\bm{\Omega}\right) = n - d - 1$.
		
    \end{lemm}
    \begin{lemm}[\cite{LAA_Alfakih_2011}]\label{lemma:LaterUniRigid}
		Any universally rigid framework $\left(\mathcal{G},\bm{p}\right)$ in general position admits a PSD stress matrix with rank $\left(n-d-1\right)$, if $\mathcal{G}$ contains a $\left(d+1\right)$-lateration graph as a spanning subgraph.
    \end{lemm}
    
	A graph $\mathcal{G}= \left(\mathcal{V},\mathcal{E}\right)$ of $n$ vertices is called a \emph{$\left(d+1\right)$-lateration} graph if there is a permutation $\pi$ of the vertices, $\pi(1),~\pi(2), \cdots, \pi(n), $ such that (i) the first $\left(d+1\right)$ vertices, $\pi(1), \cdots, \pi(d+1)$, form a clique $ \mathcal{C} \subseteq \mathcal{V} $, which satisfies that for all $ u, v \in \mathcal{C} $ where $ u \neq v $, the edge $ e_{uv} \in \mathcal{E} $, (ii) each remaining vertex $\pi(j)$, for $j = \left(d+2\right),\cdots, n$, is adjacent to $\left(d+1\right)$ vertices in the set $\left\{\pi(1),~\pi(2) \cdots \pi(j-1)\right\}$.

%

	Recently, the stress matrix has been widely applied to stabilize \emph{affine formations} of multiple vehicles  \cite{TAC_Lin_2016,TAC_Zhao_2018,CJA_Xu_2019,SJCO_Yang_2021,Li_LayeredAffine_2021,FITEE_Li_2022,OE_Zheng_2022,TCST_Zhu_2023,TITS_Wang_2023}. Different from rigid formations, affine formations allow the rotation, translation, scaling, shearing and combinations of them, which have a great advantage in enhancing formation maneuverability. 
	Suppose each agent is governed by a single-integrator dynamics, $\dot{z}_i = u_i$. Under the classical local interaction law $u_i = -\sum_{j \in \mathcal{N}_i} \varpi_{ij}z_{ij}$ over the undirected graph $\mathcal{G}$, the closed-loop system is described by 
	\begin{equation}\label{eq:AffineFormation_CL}
		\dot{z} = -\left(\bm{\Omega} \otimes \mathbf{I}_d\right) z.
	\end{equation}
	Considering a target configuration $\bm{p} $ in $\mathbb{R}^d$, the affine image of $\bm{p}$ is shown as 
		$\mathcal{A}\left(\bm{p}\right):=\left\{\bm{q}=(\mathbf{I}_n\otimes\bm{A})\bm{p}+\bm{1}_n \otimes \bm{b} \right\}$, 
	where the invertible matrix $\bm{A} \in \mathbb{R}^{d\times d}$ and translation vector $\bm{b} \in \mathbb{R}^{d}$ are continuous on $t$. 
	An affine formation of $\bm{p}$ is said to be stabilizable over the undirected graph $\mathcal{G}$ if there exists a symmetric matrix $\bm{\Omega}$ associated with $\mathcal{G}$ such that the state of \eqref{eq:AffineFormation_CL} converges to a point in $\mathcal{A}\left(\bm{p}\right)$. The following lemma provides a necessary and sufficient condition for stabilizability of an affine formation.
	
	\begin{lemm}\label{lemma:AffineFormation_Stabilizability}
		Suppose an undirected graph $\mathcal{G}$ has $n$ nodes with $n \ge d+2$ and $\bm{p}=\left[\bm{p}_1^T,\cdots,\bm{p}_n^T\right]^T $ is a general configuration, where $\mathcal{G}$ contains a $\left(d+1\right)$-lateration graph as a spanning subgraph. Then an affine formation of $\bm{p}$ is stabilizable over $\mathcal{G}$ if and only if $\mathcal{G}$ is universally rigid.
    \end{lemm}
    \begin{proof}
		(Sufficiency) If $\mathcal{G}$ is universally rigid, then by Lemma~\ref{lemma:LaterUniRigid}, for a general configuration $\bm{p}$, there exists a stress matrix $\bm{\Omega}$ (satisfying $\bm{\Omega}\mathbf{1}_n =\bm{0}$ and $\left(\bm{\Omega} \otimes \mathbf{I}_d\right) \bm{p} =\bm{0}$) that is of rank $\left(n-d-1\right)$ and PSD. Thus, we construct such a matrix $\bm{\Omega}$ for system \eqref{eq:AffineFormation_CL}, for which the eigenvalues of $\bm{\Omega}$ are all positive other than $d+1$ zero eigenvalues with $d+1$ linearly independent associated eigenvectors. Then it follows from the properties of Kronecker product that system \eqref{eq:AffineFormation_CL} is asymptotically stable.
		Since  $\left(\bm{\Omega} \otimes \mathbf{I}_d\right) \bm{p} =\bm{0}$) that is of rank $\left(n-d-1\right)$, it turns out that for any $A \in \mathbb{R}^{d\times d}$ and $b \in\mathbb{R}^d$,
			$\left(\bm{\Omega} \otimes \mathbf{I}_d\right) \left[\left(\mathbf{I}_n \otimes A\right)\bm{p} + \mathbf{1}_n \otimes \bm{b}\right]
			= \left(\mathbf{I}_n \otimes \bm{A}\right)\left(\bm{\Omega} \otimes \mathbf{I}_d\right)\bm{p}
			=\bm{0}$, 
		which means the affine image $\mathcal{A}\left(\bm{p}\right)$ is a subset of the equilibrium
		set. Moreover, if $\operatorname{span}\left\{
		\bm{p}_1,\cdots,\bm{p}_n\right\} = \mathbb{R}^d$, then $\mathcal{A}\left(\bm{p}\right)$ is a linear subspace of dimension $d^2 + d$, which equals to the dimension of null space of $\bm{\Omega}\otimes\mathbf{I}_d$. Therefore, it is certain that
		the equilibrium set of system \eqref{eq:AffineFormation_CL} equals to $\mathcal{A}\left(\bm{p}\right)$.
		Thus, we obtain that the system \eqref{eq:AffineFormation_CL} converges to $\mathcal{A}\left(\bm{p}\right)$.
		
		(Necessity) If an affine formation of $\bm{p}$ is stabilizable over $\mathcal{G}$,
		then there exists a symmetric matrix $\bm{\Omega}$ associated with $\mathcal{G}$ such
		that the equilibrium set of system \eqref{eq:AffineFormation_CL} equals to $\mathcal{A}\left(\bm{p}\right)$ and the
		state of system $\mathcal{A}\left(\bm{p}\right)$ converges to a point in $\mathcal{A}\left(\bm{p}\right)$. This implies that other than $d + 1$ zero eigenvalues, the eigenvalues of $\bm{\Omega}$
		are positive. Therefore, it is PSD. Thus, by
		Lemma~\ref{lemma:StressUniRigid} it follows that $\mathcal{G}$ is universally rigid.
	\end{proof}
	
	
	In the affine formation control scheme, the leader-follower structure is adopted naturally. Let $n_l$ vertices $\mathcal{V}_l =\lbrace v_1, \cdots, v_{n_l}\rbrace $ be the leaders, and the remaining $n_f = n - n_l$ vertices in the framework be the followers, denoted by $\mathcal{V}_f = \mathcal{V} \backslash \mathcal{V}_l$. Accordingly, the target configuration can be shown as $\bm{p}=\left[\bm{p}_{l}^{T}, \bm{p}_{f}^{T}\right]^{T} $, where $\bm{p}_{l}$ represents the configuration of $n_l$ leaders and $\bm{p}_{f}$ describes the configuration of the followers. 
	In this paper, we assume that $n_l=d+1$ and leaders affinely span (refer to \cite{TAC_Zhao_2018} ) $\mathbb{R}^d$, following the analysis in \cite{TAC_Zhao_2018,FITEE_Li_2022}. Moreover, the configuration of followers can be decided by leaders, that is, affine localizability.
	


	\begin{defi}[Affine localizability \cite{TAC_Zhao_2018} ]\label{defi:AffineLocalizability}
		The affine formation $\left(\mathcal{G},\bm{p}\right)$ is affinely localizable by the leaders if for any $\bm{q} = \left[\bm{q}_l^T,\bm{q}_f^T\right]^T \in \mathcal{A}(\bm{p})$, $\bm{q}_f$ is uniquely determined by $\bm{q}_l$.
	\end{defi}
	
	Consequently, based on Lemma~\ref{lemma:AffineFormation_Stabilizability} and Definition~\ref{defi:AffineLocalizability}, for an affine formation tracking problem, the two most important properties are universal rigidity and affine localizability when the leaders are fully controlled. These two properties describe the geometric properties of affine formations, and their algebraic properties are closely related to the stress matrix, as shown below. 
	
	
	\begin{lemm}[\cite{TAC_Zhao_2018}]\label{lemma:AffineLocalizable}
		If the formation $\left(\mathcal{G},\bm{p}\right)$ is universally rigid and $\left\{ \bm{p}_i \right\}_{i=1}^n$ affinely span $\mathbb{R}^d$, $\left(\mathcal{G},\bm{p}\right)$ is affinely localizable if and only if $\bm{\Omega}_{ff}$ is nonsingular, where the stress matrix $\bm{\Omega}$ associated with $\left(\mathcal{G},\bm{p}\right)$ be $\bm{\Omega} =\left[\begin{array}{cc}
			\bm{\Omega}_{l l} & \bm{\Omega}_{l f} \\
			\bm{\Omega}_{f l} & \bm{\Omega}_{f f}
		\end{array}\right]$, 
		where $\bm{\Omega}_{ll} \in\mathbb{R}^{n_l\times n_l}$, $\bm{\Omega}_{lf} \in\mathbb{R}^{n_l\times n_f},~\bm{\Omega}_{fl} \in\mathbb{R}^{n_f\times n_l}$, and $\bm{\Omega}_{ff} \in\mathbb{R}^{n_f\times n_f}$.
	\end{lemm}
	
	The null space of $\bm{\Omega}$ is shown in Lemma~\ref{lemma:NullStress}. 
	\begin{lemm}[\cite{TAC_Zhao_2018}]\label{lemma:NullStress}
		Assume that the nominal formation $\left(\mathcal{G},\bm{p}\right)$ has a PSD stress matrix satisfying $\operatorname{rank}\left(\bm{\Omega}\right)=n-d-1$. The following conditions are equivalent to each other
		\begin{itemize}
			\item[(i)] $\left\{\bm{p}_i\right\}_{i=1}^n $ affinely span $\mathbb{R}^d$; 
			\item[(ii)] $\operatorname{null}(\bm{\Omega}) = \operatorname{col}\left(\left[\begin{array}{ccc}
				\bm{p}_1 & \cdots & \bm{p}_n \\
				1 & \cdots & 1
			\end{array}\right]^T \right)$.
		\end{itemize}
	\end{lemm}
	
	Accordingly, the universal rigidity and affine localizability are two key properties to ensure the feasibility of the framework under various affine formation control schemes, which are also the focus of our subsequent theoretical analysis.
	
	\subsection{Problem Statement}
	
	In a robot swarm, due to the time-varying target tasks and complex environments, the configuration and topology are usually not fixed.  
    For example, when a group of unmanned aerial vehicles (UAVs) pass through dense buildings or woods, situations including UAV joining, crashing, and communication disruptions may occur. From the perspective of graph theory, these situations can be described by the vertex addition, edge deletion and vertex deletion (corresponding to the UAV joining, communication disruption and crashing), respectively.  The original topology of the swarm is changed, and even the inherent properties (such as connectivity, rigidity, etc.) are disrupted. 
    Considering the supportive role of topology in cooperative control schemes, it is particularly crucial to timely construct the topology in various situations, which gives rise to the demand for framework construction strategies. Furthermore, considering the lack of understanding of the global states and measurements for perception-limited agents, the study of distributed algorithms is motivated in this paper.
	
	In this paper, we investigate how to construct an affine framework to quickly adapt to different operating conditions. 
	From the perspective of graph theory, the connections between vertices and the weights of edges are reshaped to guarantee the universal rigidity and affine localizability. 
	Given an original affine framework $\left(\mathcal{G}_0,\bm{p}_0\right)$ with a stress matrix $\bm{\Omega}_0$, we aim to design distributed algorithms to guarantee the universal rigidity and affine localizability of the obtained framework after adding vertices to $\left(\mathcal{G}_0,\bm{p}_0\right)$, or removing vertices and edges from $\left(\mathcal{G}_0,\bm{p}_0\right)$. 
	We assume that $\left(\mathcal{G}_0,\bm{p}_0\right)$ meets the following condition.
	
	\begin{assu}\label{assu:OriginalFrame}
		Suppose that the framework $\left(\mathcal{G}_0,\bm{p}_0\right)$ in general position is affinely localizable and universally rigid, where $\mathcal{G}_0$ contains a $\left(d+1\right)$-lateration graph as a spanning subgraph.
	\end{assu}
	
Based on Lemma~\ref{lemma:LaterUniRigid}, it is deduced that the original stress matrix $\bm{\Omega}_0$ is PSD and $\operatorname{rank}\left(\bm{\Omega}_0\right)=n-d-1$ under  Assumption~\ref{assu:OriginalFrame}. The problems studied in this paper are established as follows. 
	
		
		

	\begin{prob}[Vertex Addition]\label{prob:VA}
		Given an affine framework $\left(\mathcal{G},\bm{p}\right)$ and a new vertex $v_u$, design a strategy to link $v_u$ with $\left(\mathcal{G},\bm{p}\right)$ and then rearrange the stress so that the obtained framework $\left(\mathcal{G}^+,\bm{p}^+\right)$ is universally rigid and affinely localizable, where $\mathcal{G}^+ = \left( \mathcal{V} \cup \{v_u\}, \mathcal{E}^+\right)$ and $\bm{p}^+ = \left[ \bm{p}^T,~\bm{p}_u^T \right]^T$.
	\end{prob}
	
	\begin{prob}[Edge Deletion]\label{prob:ED}
		Given an affine framework $\left(\mathcal{G},\bm{p}\right)$ with an edge $e_{jk}$ (i.e., $j,k\in \mathcal{V},~e_{jk} \in \mathcal{E}$), design a strategy to delete $e_{jk}$ from  $\mathcal{E}$ while maintaining the universal rigidity and affine localizability of the obtained framework $\left(\mathcal{G}_{ed},\bm{p}_{ed}\right)$, where $e_{jk} \notin \mathcal{E}_{ed}$.  
	\end{prob}
	
	\begin{prob}[Vertex Deletion]\label{prob:VD}
		Given an affine framework $\left(\mathcal{G},\bm{p}\right)$ with a specific vertex $v_u \in \mathcal{V}$, design a strategy to delete $v_u$ from the vertex set $\mathcal{V}$ while maintaining the universal rigidity and affine localizability of the obtained framework $\left(\mathcal{G}_{vd},\bm{p}_{vd}\right)$, where $v_u \notin \mathcal{V}_{vd} $.    
	\end{prob}
	
	Lemmas are introduced to lay the foundation for the theoretical analysis in the following contexts.
	\begin{lemm}[Rank-Nullity Theorem \cite{Lang_LA_1986}]\label{lemm:RN} 
		If there is a matrix $A $ with $x$ rows and $y$ columns over a field, then $
		\operatorname{rank}(A) + \operatorname{nullity}(A) = y$, 
		where $\operatorname{nullity}(A)$ means the nullity of the matrix $A$, that is, the dimension of the kernel of $A$.
	\end{lemm}

        \begin{lemm}[\cite{Zhang_Schur_2006}]\label{lemma:PD}
				For any symmetric matrix $\bm{M}$ of the form $\bm{M}=\left[\begin{array}{cc}
						\bm{A} & \bm{B}\\
						\bm{B}^T & \bm{C}
					\end{array}\right]$,
				if $\bm{C}$ is invertible then the following properties hold:
				\begin{itemize}
					\item[(1)] $\bm{M} \succ 0$ $\Longleftrightarrow$ $\bm{C} \succ 0$ and $\bm{A}-\bm{B} \bm{C}^{-1} \bm{B}^T \succ 0$.
					\item[(2)] If $\bm{C} \succ 0$, then $\bm{M} \succeq 0$ $\Longleftrightarrow$ $\bm{A}-\bm{B} \bm{C}^{-1} \bm{ B}^T \succeq 0$.
				\end{itemize}
	\end{lemm}
	

	\section{CONSTRUCTION OF AFFINE
		FRAMEWORKS IN 2D}\label{sec:MainResult}
	In this section, we aim to develop affine framework construction strategies to solve the problems proposed in Section \ref{sec:Preliminaries}. We consider $d=2$ in this section, and the methods are further extended to three-dimensional spaces in Section~\ref{sec:Discuss}.
	
	
	
	\subsection{Vertex Addition}\label{sec:VertexAddition}
	Given an affine formation $\left(\mathcal{G},\bm{p}\right)$ and a new vertex $v_u$, our objective here is to merge $v_u$ to the original framework $\left(\mathcal{G},\bm{p}\right)$ with appropriate edges and weights, resulting in an new framework $\left(\mathcal{G}^+,\bm{p}^+\right)$ with universal rigidity and affine localizability. Naturally, the leaders in $\left(\mathcal{G},\bm{p}\right)$ can be inherited by $\left(\mathcal{G}^+,\bm{p}^+\right)$, i.e, $\mathcal{V}_l^+  = \mathcal{V}_l$. The vertex addition strategy is proposed as below.
	
            \begin{theo} \label{theo:VertexAddition}
				Under Assumption~\ref{assu:OriginalFrame}, consider a new vertex $v_u$ and an affine framework $\left(\mathcal{G},\bm{p}\right)$ in $\mathbb{R}^2$. Suppose that $| \mathcal{N}_u^p| \ge 3$, and  the set $\left\{\bm{p}_u\right\} \cup \left\{\bm{p}_j:j \in \mathcal{V} \right\}$ is in general position. With a positive scaling parameter $s$, after adding three edges connecting the vertex $v_u$ and the existing vertices $v_i$, $v_j, v_k \in \mathcal{N}_u^p$ to $\left(\mathcal{G},\bm{p}\right)$, the obtained framework $\left(\mathcal{G}^+,\bm{p}^+\right)$ is universally rigid and affinely localizable.
                    %
		\end{theo}

	The requirement $\mid \mathcal{N}_u^p\mid \ge 3$ indicates that for a new vertex to be effectively integrated into the affine formation, it possesses adequate sensing and communication capabilities to establish connections with at least three vertices. 
    The proof of Theorem~\ref{theo:VertexAddition} can be translated into an algebraic problem: construct a suitable stress matrix $\bm{\Omega}^+$ for $\left(\mathcal{G}^+,\bm{p}^+\right)$ and demonstrate its algebraic properties to establish the universal rigidity and affine localizability of $\left(\mathcal{G}^+,\bm{p}^+\right)$, as shown below. 

    In $\mathcal{N}_u^p$, the three neighbors are denoted as $v_i$, $v_j$, $v_k$, and they are in a general position with respect to $v_u$. 
    To generate an eligible affine framework, we add edges between $v_u$ and all of $v_i$, $v_j$ and $v_k$, the resulting framework $\left(\mathcal{G}^+,\bm{p}^+\right)$ is shown as in Fig.~\ref{fig:VertexAdd_2}, where $\mathcal{G}^+ = \left(\mathcal{V} \cup \{v_u\}, \mathcal{E}\cup \left\{e_{iu},~e_{ju},~e_{ku}\right\} \right)$ and $\bm{p}^+ = \left[\bm{p}^T, \bm{p}_u^T\right]^T $. To maintain the universal rigidity and affine localibility of $\left(\mathcal{G}^+,\bm{p}^+\right)$, the stress in $\mathcal{G}^+$ need to be rearranged, as discussed below. 

    \begin{figure}[htbp]
		\centering
		\subfigure[]{\includegraphics[height=0.11\textwidth]{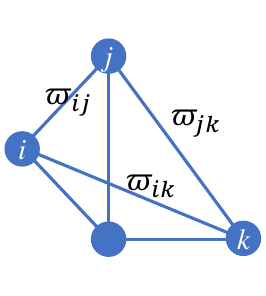} 
		}
		\subfigure[] {\includegraphics[height=0.11\textwidth]{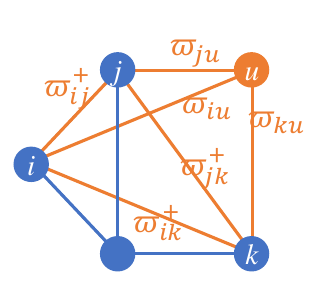}
        \label{fig:VertexAdd_2}
		}
            \subfigure[] {\includegraphics[height=0.11\textwidth]{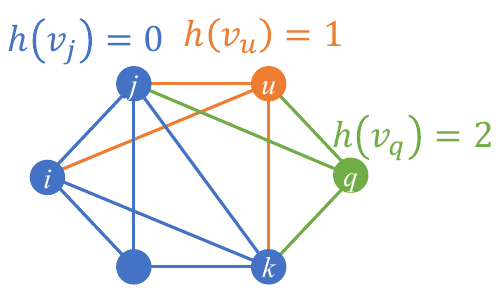}
            \label{fig:VertexAdd_3}
		}
		\caption{An example of vertex addition. (a) The original framework $\left(\mathcal{G}_0,\bm{p}_0\right)$ in general position. (b) The obtained graph $\mathcal{G}^+$, where the vertex $v_u$ (labeled by a orange circle) and three edges are added. (c) Adding the vertex $v_q$ (labeled by a green circle) and three edges are added.}
		\label{fig:VertexAdd}
    \end{figure}
	
	Reorder the elements in $\mathcal{V}$ to place the selected three vertices $v_i$, $v_j$, $v_k$ at the end, i.e., {\small$\mathcal{V} \triangleq \{\cdots,v_{i-1}, v_{i+1},\cdots,v_{j-1}, v_{j+1},\cdots,v_{k-1}, v_{k+1},\cdots,v_i,v_j,v_k\}$} if $i<j<k$. The configuration $\bm{p}$ is also reordered based on the index of $\mathcal{V}$. Thus, for the graph $\mathcal{G}$, the corresponding stress matrix $\bm{\Omega}$ can be described as below.
	\begin{equation}
    \footnotesize
		\bm{\Omega}
		=\left[\begin{array}{c:ccc}
			\bm{\Omega}^{P1} &   & \bm{\Omega}^{P2} &   \\
			\hdashline & \Omega_{i i}& \Omega_{i j}& \Omega_{i k} \\
			(\bm{\Omega}^{P2})^T & \Omega_{i j}& \Omega_{j j}& \Omega_{j k}  \\
			& \Omega_{i k} & \Omega_{k j} & \Omega_{k k}
		\end{array}\right],
	\end{equation}
	%
	where $\bm{\Omega}$ is indexed based on the elements in $\mathcal{V}$.  Inherited from Lemma~\ref{lemma:NullStress}, we have 
	\begin{equation} \label{eq:KerStress}
        \footnotesize
		\left[\begin{array}{ccccc}
			\bm{p}_1 & \cdots & \bm{p}_i & \bm{p}_j & \bm{p}_k \\
			1 & \cdots & 1 & 1 & 1
		\end{array}\right]\bm{\Omega} = \bm{0}.
	\end{equation}
    After adding a new vertex $v_u$ and three edges, $e_{iu},~e_{ju},~e_{ku}$ (as shown in Fig.~\ref{fig:VertexAdd_2}), the following equation~\eqref{eq:StressV_u} is established based on the definition of a equilibrium stress.
	\begin{equation}\label{eq:StressV_u}
        \footnotesize
		\varpi_{iu}\left(\bm{p}_i-\bm{p}_u\right) + \varpi_{ju}\left(\bm{p}_j-\bm{p}_u\right) + \varpi_{ku}\left(\bm{p}_k-\bm{p}_u\right) = \bm{0},
	\end{equation}
	where $\varpi_{iu}$, $\varpi_{ju}$ and $\varpi_{ku}$ are the weights of edges $e_{iu},~e_{ju},~e_{ku}$. 
	The resulting stress matrix $\bm{\Omega}^+$ is expressed as follows.
	\begin{equation}
		\footnotesize
		\begin{aligned}
			\bm{\Omega}^+ \triangleq & \left[\begin{array}{c:c}
				\bm{\Omega}^{P1}  & \begin{array}{cc}
					 \bm{\Omega}^{P2}  &  \bm{0}
				\end{array} \\
				\hdashline 
				\begin{array}{c}
					(\bm{\Omega}^{P2})^T\\
					\bm{0}
				\end{array} & \bm{\Omega}_{va}^+
			\end{array}\right]\\
			=&\left[\begin{array}{c:cccc}
				\bm{\Omega}^{P1}  &  & \bm{\Omega}^{P2} &  &  \bm{0} \\
				\hdashline & \Omega_{i i}^+& \Omega_{i j}^+& \Omega_{i k}^+ & - \varpi_{iu} \\
				(\bm{\Omega}^{P2})^T & \Omega_{j i}^+ & \Omega_{j j}^+ & \Omega_{j k}^+  & - \varpi_{ju} \\
				&\Omega_{k i}^+  & \Omega_{k j}^+  & \Omega_{k k}^+ & - \varpi_{ku}\\
				\bm{0} & - \varpi_{iu}  & - \varpi_{ju}  &- \varpi_{ku} &  \varpi_{uu}
			\end{array}\right],
		\end{aligned}
	\end{equation}
	where $\varpi_{uu} = \varpi_{iu} + \varpi_{ju} + \varpi_{ku}$. 
	Similarly, we have 
	\begin{equation}\label{eq:KerStressPlus}
		\footnotesize
		\left[\begin{array}{cccccc}
			\bm{p}_1 & \cdots & \bm{p}_i & \bm{p}_j & \bm{p}_k & \bm{p}_u \\
			1 & \cdots & 1 & 1 & 1 & 1
		\end{array}\right]\bm{\Omega}^+ = \bm{0}.
	\end{equation}
	
	
	Let $\bm{\Omega}_{va} = \left[\begin{array}{cccc}
		\Omega_{i i} & \Omega_{i j}& \Omega_{i k} & 0\\
		\Omega_{i j} & \Omega_{j j}& \Omega_{j k} & 0 \\
		\Omega_{i k} & \Omega_{k j} & \Omega_{k k} & 0\\
		0 & 0 & 0 & 0
	\end{array}\right]$, and $\bm{\Omega}_u = \bm{\Omega}_{va}^+ - \bm{\Omega}_{va}$. Combining eqs.~\eqref{eq:KerStress}-\eqref{eq:KerStressPlus}, we get
	\begin{equation} \label{eq:Omega_Plus}
		\footnotesize
		\begin{aligned}
			&\left[ \begin{array}{cccc}
					\bm{p}_i & \bm{p}_j & \bm{p}_k & \bm{p}_u\\
					1 & 1 & 1 & 1
				\end{array} \right] \left(\bm{\Omega}_{va}^+ - \bm{\Omega}_{va} \right) \triangleq \bm{P}_u \bm{\Omega}_u = \bm{0}.
		\end{aligned}
	\end{equation}
Since $v_u$ is not collinear with any two of $v_i,v_j,v_k$, the vectors $\left(\bm{p}_i-\bm{p}_u\right)$, $\left(\bm{p}_j-\bm{p}_u\right)$ and $\left(\bm{p}_k-\bm{p}_u\right)$ are linearly independent, which implies that $\operatorname{rank}(\bm{P}_u) = 3$. Based on Lemma~\ref{lemm:RN}, we have $\operatorname{nullity}(\bm{P}_u) = 1 $. Accordingly, there exists a nonzero vector $\bm{\phi} = \left[\phi_1, \phi_2,\phi_3,\phi_4\right]^T$ satisfying $\bm{P}_u \bm{\phi} = \bm{0}$, 
	\begin{equation}\label{eq:phi}
		\small
		\begin{aligned}
			\phi_1 +\phi_2+\phi_3+\phi_4 =0 ,~
			\phi_1 \bm{p}_i + \phi_2 \bm{p}_j + \phi_3 \bm{p}_k + \phi_4 \bm{p}_u = \bm{0} .\\
		\end{aligned}
	\end{equation}
	Thus, combined with eq.\eqref{eq:Omega_Plus}, there exists a vector $\bm{v}$ satisfying $\bm{\Omega}_u = \bm{\phi} \bm{v}^T$. Since $\bm{\Omega}_u $ is a symmetric matrix, $\bm{\Omega}_u $ is the outer product of $ \bm{\phi} $.
	Accordingly, the matrix $\bm{\Omega}_u$ is designed as below.
	\begin{equation}\label{eq:Omega_u1}
        \footnotesize
		\bm{\Omega}_u \triangleq s \bm{\phi} \bm{\phi}^T = s \left[\begin{array}{cccc}
			\phi_1^2 & \phi_1\phi_2 & \phi_1\phi_3 & \phi_1\phi_4\\
			\phi_1 \phi_2 & \phi_2^2 & \phi_2\phi_3 & \phi_2\phi_4\\
			\phi_1 \phi_3 & \phi_2\phi_3 & \phi_3^2 & \phi_3\phi_4\\
			\phi_1 \phi_4 & \phi_2\phi_4 & \phi_3\phi_4 & \phi_4^2
		\end{array}\right],
	\end{equation}
	where $s$ is a scaling parameter. Thus, we have $\varpi_{iu} = -s\phi_1\phi_4,~\varpi_{ju} = -s\phi_2\phi_4,~\varpi_{ku} = -s\phi_3\phi_4$ and $\varpi_{uu} = s \phi_4^2$.
	Then, the corresponding stress matrix of the newly obtained framework $\left(\mathcal{G}^+,\bm{p}^+\right)$ is established as follows.
	\begin{equation}\label{eq:StressMatrixPlus}
        \footnotesize
		\bm{\Omega}^+ = \underbrace{\left[\begin{array}{c:c}
				\bm{\Omega} & \bm{0}_{n\times1}\\
				\hdashline
				\bm{0}_{1 \times n} & 0
			\end{array}\right]}_{\triangleq \bm{\Omega}_a} + \underbrace{\left[\begin{array}{c:c}
				\bm{0} & \bm{0}_{\left(n-3\right)\times4}\\
				\hdashline
				\bm{0}_{4\times\left(n-3\right)}& \bm{\Omega}_u
			\end{array}\right]}_{\triangleq \bm{\Omega}_b},
	\end{equation}
	where $\bm{\Omega}_u$ is described as eq.~\eqref{eq:Omega_u1} with $s>0$ and $\bm{\Omega}$ is the original stress matrix related to $\left(\mathcal{G},\bm{p}\right)$.

		A stress matrix for $\left(\mathcal{G}^+,\bm{p}^+\right)$ is presented in eq.~\eqref{eq:StressMatrixPlus}. By analyzing the properties of $\bm{\Omega}^+$, a comprehensive theoretical analysis is provided in Appendix~\ref{sec:app} to demonstrate the effectiveness of our proposed strategy in constructing affine frameworks. Please refer to Appendix~\ref{sec:app} for details on the proof of Theorem~\ref{theo:VertexAddition}.

				According to Theorem \ref{theo:VertexAddition}, the distributed nature of our proposed strategy is fully presented. Different from \cite{RAL_Xiao_2022}, the global positions are not necessary based on the definition of a equilibrium stress \eqref{eq:StressV_u}. The weight of an edge is determined by the relative position of the corresponding two endpoints when new edges are added to the framework, e.g., $\bm{\Omega}_u$ in eq.~\eqref{eq:Omega_u1}. Accordingly, the stress matrix can be updated based on local measurements. Moreover, only several vertices are involved in the generation and reconstruction of affine frameworks, which supports the distributed execution of our methods. Consequently, the construction of affine framework is solved by local measurements and communication, leading to the great application prospects in multi-robot systems. In fact, our algorithms can be integrated with most, if not all, distributed affine formation control schemes. A simulation is carried out in Section~\ref{sec:Simu_Ctrl} to verify the compatibility, where an example of affine formation tracking control law \cite{FITEE_Li_2022} is applied. 
				 
            Based on Theorem~\ref{theo:VertexAddition}, the newly generated graph $\mathcal{G}^+$, based on our proposed strategy, includes a 3-lateration graph as a spanning subgraph. Each newly added vertex has three neighbors as the parents, and the hierarchical structure is naturally built. Denote the hierarchy of vertices in the original affine framework $\left(\mathcal{G},\bm{p}\right)$ as $0$. Define the hierarchy $h(v_i)$ of a vertex $v_i$ as the length of its longest path from $v_i$ to the vertices in $\left(\mathcal{G},\bm{p}\right)$. To add a free $v_u$ to the existing graph, three existing vertices are selected to be parents, denoted by $v_i,v_j,v_k$. The hierarchy of vertex $v_u$ is defined as $h(v_u) = \max\left(h(v_i) ,h(v_j) ,h(v_k) \right)+ 1$. An example is also shown in Fig~\ref{fig:VertexAdd_3}, where the hierarchy of $v_q$ is $2$.
			%
			 Moreover, we can repeatedly apply Theorem~\ref{theo:VertexAddition} to add $q$ vertices to $\left(\mathcal{G},\bm{p}\right)$. The algorithm is shown in Algorithm~\ref{Algo:VertexAdd}.
			With Algorithm~\ref{Algo:VertexAdd}, we can easily grow the original affine framework to a larger scale with a hierarchical structure, called as the hierarchical affine framework (HAF). In the HAF $(\mathcal{G}^+,\bm{p}^+)$, each vertex $v_u$ with $h(v_u)>0 $ has three parents, so that the graph $\mathcal{G}^+$ contains a 3-lateration graph as a spanning subgraph. The well-structured hierarchical affine framework facilitates the subsequent research on edge and vertex deletion issues. 
				
				\begin{algorithm} \label{Algo:VertexAdd}
					\small
					\caption{\textbf{ Vertex Addition Algorithm in $\mathbb{R}^2$}}
					\KwIn{An affine framework $\left(\mathcal{G},\bm{p}\right)$ with a stress matrix $\bm{\Omega}$, $q$ new vertices $v_{a1},\cdots,v_{aq}$ with $\bm{p}_{a1},\cdots,\bm{p}_{aq}$.}
					\KwOut{An augmented framework $\left(\mathcal{G}_{add},\bm{p}_{add} \right)$ with a new stress matrix $\bm{\Omega}_{add}$.}
					\SetKwFunction{MyFuns} {MyFuns}
					\SetKwProg{Fn}{Function}{:}{}
					Set the positive scaling parameter $s $;\\
					\For{$\ell=1,~\cdots,~q$}{
						$v_u \leftarrow v_{a\ell}$, $\bm{p}_u \leftarrow \bm{p}_{a\ell}$;\\
						Choose three perceived vertices $v_i$, $v_j$ and $v_k$, where $v_i,~v_j,~v_k \in \mathcal{V}$; \\
						$\mathcal{G}_{add}= \left(\mathcal{V}_{add},\mathcal{E}_{add}\right) $ where $\mathcal{V}_{add} \gets \mathcal{V}\cup\left\{ v_u\right\}$, $\mathcal{E}_{add} \gets \mathcal{E}\cup\left\{ e_{iu},~e_{ju},~e_{ku}\right\}$;\\
						$\bm{\Omega}_u \gets s \bm{\phi} \bm{\phi}^T$ based on Eq.~\eqref{eq:Omega_u1}.
						$\bm{p}_{add} \gets \left[\bm{p}^T,\bm{p}_u^T\right]^T$ ; \\ 
						$\bm{\Omega}_{add} \gets \bm{\Omega}_a + \bm{\Omega}_b$ defined in Eq.~\eqref{eq:StressMatrixPlus};\\
						$\mathcal{G} \leftarrow \mathcal{G}_{add}$, $\bm{p} \leftarrow \bm{p}_{add}$, $\bm{\Omega} \leftarrow \bm{\Omega}_{add}$;
					}
					Return $\left(\mathcal{G}_{add},\bm{p}_{add} \right)$ and $\bm{\Omega}_{add}$.
				\end{algorithm}

                It is evident that our proposed HAF growing strategy in Theorem~\ref{theo:VertexAddition}  is incremental, inherently possessing a relatively low time complexity. When adding the new vertices, we need few computational resources to reconstruct the affine framework, because only low dimensional matrix calculations are required in Algorithm~\ref{Algo:VertexAdd}. Providing that there are $q$ new vertices joining the framework, the time complexity can be calculated as $O(qn) $. 
                Therefore, in practice, our algorithm is portable enough to be applied to robots with limited computing power to deal with a variety of collaborative tasks in time.
				
			\begin{rema}
				The proposed vertex addition strategy in this paper draws inspiration from HCs, which are widely employed to construct minimally rigid graphs. However, our contribution extends beyond existing studies by not only investigating the preservation of universal rigidity during vertex addition but also introducing a stress matrix update strategy to ensure structural consistency. Moreover, in contrast to existing literature focusing on the super-stability of tensegrity frameworks \cite{TCYB_Yang_2019}, this paper constructs HAFs in general position based on $(d+1)$-lateration graphs, which inherently accommodates the dimensional scalability of affine formations. Furthermore, the hierarchical structure of the framework enables systematic exploration of edge and vertex deletion strategies in subsequent analyses.
			\end{rema}

			\subsection{Edge Deletion} \label{sec:EdgeDeletion}
			Note that removing an edge from the HAF $\left(\mathcal{G},\bm{p}\right)$ is equivalent to adjusting the related weight to zero, without damaging the universal rigidity and affine localizability. Using the parameter $s$ in eq.~\eqref{eq:Omega_u1}, a natural idea is generated to eliminate a certain edge, as detailed below.
			
			Take the removal of edge $e_{jk}$ between vertex $v_j$ and $v_k$ in $\left(\mathcal{G},\bm{p}\right)$ as an example. If we select two suitable vertices $v_i$ and $v_q$ in $\mathcal{N}_j^p \cap \mathcal{N}_k^p$, 
			the vertices $v_i,~v_j$ and $v_k$ can form a triangle, and the regions of the fourth vertex $v_q$ lies in can be labeled by $\mathfrak{a},\cdots,\mathfrak{g}$, as shown in Fig.~\ref{fig:FourVer}. 
			\begin{figure}[htbp]
				\centering
				\includegraphics[scale=0.12]{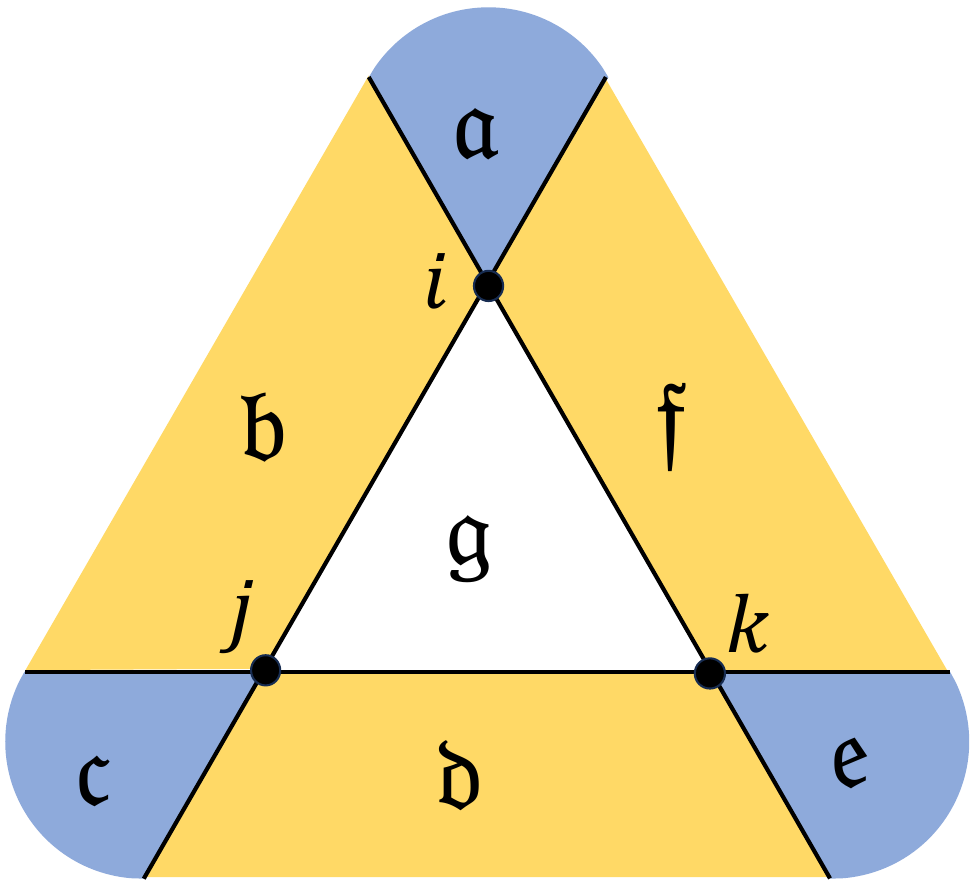}
				\caption{Possible relative configurations of four vertices in general position in $\mathbb{R}^2$.}
				\label{fig:FourVer}
			\end{figure}
			For four non-collinear vertices, we get a non-zero vector $\bm{\phi} $ satisfying $\left[ \begin{array}{cccc}
				\bm{p}_i & \bm{p}_j & \bm{p}_k & \bm{p}_q\\
				1 & 1 & 1 & 1
			\end{array} \right]\bm{\phi} =\bm{0}$, where $\bm{\phi} = \left[\phi_1,\cdots,\phi_4\right]^T$. In the original framework $\left(\mathcal{G},\bm{p}\right)$, the stress matrix $\bm{\Omega}$ is indexed based on the elements in $\mathcal{V}$, which are reordered such that $v_i,~v_j,~v_k,~v_q$ are listed at the end. The entries involved in these four vertices in the stress matrix $\bm{\Omega}$ can be described as below.
			\begin{equation}
				\footnotesize
				\bm{\Omega}_v=\left[\begin{array}{cccc}
					\Omega_{i i}& \Omega_{i j}& \Omega_{i k} & \Omega_{i q}\\
					\Omega_{i j}& \Omega_{j j}& \Omega_{j k} & \Omega_{j q} \\
					\Omega_{i k} & \Omega_{j k} & \Omega_{k k} & \Omega_{k q} \\
					\Omega_{i q} & \Omega_{j q} & \Omega_{k q} & \Omega_{q q}
				\end{array}\right].
			\end{equation}
			After eliminating $e_{jk}$, the framework is described by $\left(\mathcal{G}^-,\bm{p}\right)$, and the corresponding block matrix is established as follows.
			\begin{equation}
				\footnotesize
				\bm{\Omega}_v^-=\left[\begin{array}{cccc}
					\Omega_{i i}^-& \Omega_{i j}^-& \Omega_{i k}^- & \Omega_{i q}^-\\
					\Omega_{i j}^-& \Omega_{j j}^- &  0 & \Omega_{j q}^- \\
					\Omega_{i k}^- & 0 & \Omega_{k k}^- & \Omega_{k q}^- \\
					\Omega_{i q}^- & \Omega_{j q}^- & \Omega_{k q}^- & \Omega_{q q}^-
				\end{array}\right].
			\end{equation}
			
			Inspired by Eqs.~\eqref{eq:Omega_Plus} and \eqref{eq:Omega_u1}, we have 
			\begin{equation}\label{eq:ED_equa}
                \footnotesize
					\left[ \begin{array}{cccc}
						\bm{p}_i & \bm{p}_j & \bm{p}_k & \bm{p}_q\\
						1 & 1 & 1 & 1
					\end{array} \right] \left( \bm{\Omega}_v^- -\bm{\Omega}_v \right) =\bm{0}.
			\end{equation}
			
			Similar to eq.~\eqref{eq:Omega_u1}, let $\bm{\Omega}_u  \triangleq \bm{\Omega}_v^- -\bm{\Omega}_v = s_{ed} \bm{\phi}\bm{\phi}^T$ so that
			\begin{equation}\label{eq:omega_ed}
				\footnotesize
				\begin{aligned}
					&\bm{\Omega}_v^- = \bm{\Omega}_v + \bm{\Omega}_u \\
						=&
						\bm{\Omega}_v +\dfrac{1}{\varpi_{qq}} 
						\left[\begin{array}{cccc}
							\varpi_{i q}^2 & \varpi_{i q}\varpi_{j q} & \varpi_{i q}\varpi_{k q} & -\varpi_{i q}\varpi_{qq}\\
							\varpi_{i q}\varpi_{j q} & \varpi_{j q}^2  & \varpi_{j q}\varpi_{k q} & -\varpi_{j q}\varpi_{qq}\\
							\varpi_{i q}\varpi_{k q} & \varpi_{j q}\varpi_{k q} &  \varpi_{k q}^2 & -\varpi_{k q}\varpi_{qq}\\
							-\varpi_{i q}\varpi_{qq} & -\varpi_{j q}\varpi_{qq} & -\varpi_{k q}\varpi_{qq} & \varpi_{qq}^2
						\end{array}\right] 
					\end{aligned}
				\end{equation}
				where $\varpi_{iq} = -s_{ed} \phi_1\phi_4,~\varpi_{jq} = -s_{ed} \phi_2\phi_4,~\varpi_{kq} = -s_{ed} \phi_3\phi_4$ and $\varpi_{qq} = s_{ed} \phi_4^2$. To eliminate $e_{jk}$, the following equation is established:
				\begin{equation}\label{eq:ED_s1_2D}
                \footnotesize
					\begin{aligned}
						\Omega_{jk} + \dfrac{\varpi_{jq}\varpi_{kq}}{\varpi_{qq}} =0~~~
						\Rightarrow  s_{ed} =-\dfrac{\Omega_{jk}}{\phi_2 \phi_3}.
					\end{aligned}
				\end{equation}
				To maintain the universal rigidity and affine localizability of the obtained framework $\left(\mathcal{G}^-,\bm{p}\right)$, $s_{ed} >0$ 
				is required based on Theorem~\ref{theo:VertexAddition}, implying that
				\begin{equation}\label{eq:ED_s2}
                \footnotesize
					\Omega_{jk} \varpi_{jq}\varpi_{kq} <0.
				\end{equation}
				
				 With a specific $\Omega_{jk}$, the sign of $s_{ed}$ is determined by the signs of $\varpi_{jq}$ and $\varpi_{kq}$, which depend on the relative positions among $v_i,~v_j,~v_k$ and $v_q$. In detail, since $v_q$ is not collinear with any two of $v_i,~v_j$ and $v_k$, the vectors $\bm{p}_q-\bm{p}_i$ and $\bm{p}_q-\bm{p}_j$ can serve as a set of bases in $\mathbb{R}^2$. That is, 
				\begin{equation}
                \footnotesize
					\bm{p}_q-\bm{p}_k = k_1 \left( \bm{p}_q-\bm{p}_i\right) + k_2 \left( \bm{p}_q-\bm{p}_j\right),
				\end{equation}
				where $k_1$ and $k_2$ are two real parameters that are not simultaneously zero. Combining with eq.~\eqref{eq:StressV_u}, we have 
				\begin{equation}
                \footnotesize
					\left( \varpi_{iq} + \varpi_{kq}k_1\right) \left( \bm{p}_q-\bm{p}_i\right) + \left( \varpi_{jq} + \varpi_{kq}k_2\right) \left( \bm{p}_q-\bm{p}_j\right) = \bm{0}\notag.
				\end{equation}
				Accordingly, the following equations are deduced.
				\begin{equation}
                \footnotesize
					\varpi_{iq} = -k_1\varpi_{kq}, ~~\varpi_{jq} = -k_2\varpi_{kq},
				\end{equation}
				which mean $\varpi_{iq}\varpi_{jq} = k_1 k_2 \varpi_{kq}^2 $, $\varpi_{iq}\varpi_{kq} = -k_1 \varpi_{kq}^2 $ and $\varpi_{jq}\varpi_{kq} =-k_2  \varpi_{kq}^2 $. When $v_q$ lies in different regions, the signs of $k_1$, $k_2$ and the weights are listed in Table.~\ref{tab:Symbol_2D}, which serves as a guideline to determine whether the edge $e_{jk}$ can be deleted without damaging universal rigidity and affine localizability. For example, if the weight of $e_{jk}$ is positive, i.e., $\Omega_{jk}<0$, a vertex $v_q$ in $\mathcal{G}$ lying in Region $\mathfrak{a},~\mathfrak{d}$ or $\mathfrak{g}$ should be chosen to ensure a positive $s_{ed}$ based on \eqref{eq:ED_s2} and Table.~\ref{tab:Symbol_2D}. However, it is also possible that there are no vertices lying in these regions. Accordingly, based on whether $e_{jk}$ can be eliminated through the rearrangement of edge weights in $\mathcal{G}$, there are three possible cases presented as follows.
				
				\begin{table}[]
					\centering
					\caption{The sign of the weight of the edges connected to $v_q$.}
					\begin{tabular}{cccccccc}
						\toprule
						\diagbox{}{Regions}  & $\mathfrak{a}$ & $\mathfrak{b}$  & $\mathfrak{c}$  & $\mathfrak{d}$  & $\mathfrak{e}$  & $\mathfrak{f}$  & $\mathfrak{g}$\\
						\midrule
						$k_1$  & $+$  & $+$  & $-$  & $+$  & $+$  & $-$  & $-$ \\
						$k_2$  & $-$  & $+$  & $+$  & $-$  & $+$  & $+$  & $-$ \\
						$\varpi_{iq}$  & $+$  & $+$  & $-$  & $-$  & $-$  & $+$  & $+$ \\
						$\varpi_{jq}$  & $-$  & $+$  & $+$  & $+$  & $-$  & $-$  & $+$ \\
						$\varpi_{kq}$  & $-$  & $-$  & $-$  & $+$  & $+$  & $+$  & $+$ \\
						\bottomrule
					\end{tabular}
					\label{tab:Symbol_2D}
				\end{table}
				
				\textbf{Case (1): There are two vertices $v_i,v_q \in \mathcal{N}_j^p \cap \mathcal{N}_k^p$ to ensure that  $s_{ed} $ in \eqref{eq:ED_s1_2D} is positive.} 
				
				In this case, we can delete the edge $e_{jk}$ directly to obtain a new framework $\left(\mathcal{G}_{ed},\bm{p}_{ed}\right)$, where $\mathcal{G}_{ed} = \left(\mathcal{V},\mathcal{E}\setminus \left\{e_{jk}\right\}\right)$, $\bm{p}_{ed} =\bm{p}$.
				The stress matrix $\bm{\Omega}_{ed}$ for $\left(\mathcal{G}_{ed},\bm{p}_{ed}\right)$ is shown as below, 
				\begin{equation}\label{eq:StressMatrixMinus}
                \footnotesize
					\bm{\Omega}_{ed} = \bm{\Omega} + \left[\begin{array}{c:c}
						\bm{0}_{(n-4)\times (n-4)} & \bm{0}_{(n-4)\times 4}\\
						\hdashline
						\bm{0}_{4 \times (n-4)} & \bm{\Omega}_u
					\end{array}\right],
				\end{equation}
                where $\bm{\Omega}_u = -\dfrac{\Omega_{jk}}{\phi_2 \phi_3}\bm{\phi}\bm{\phi}^T = s_{ed} \bm{\phi}\bm{\phi}^T$.  Obviously, eq.~\eqref{eq:StressMatrixMinus} and \eqref{eq:StressMatrixPlus} have similar forms and meanings. As a direct extension of Theorem~\ref{theo:VertexAddition}, we have the following corollary.
				\begin{coro}\label{coro:ED}
				Consider a HAF $\left(\mathcal{G},\bm{p}\right)$ containing an edge $e_{jk}$ with the weight $-\Omega_{jk}$. If there are proper vertices $v_i$ and $v_q$ in $\mathcal{N}_j^p \cap \mathcal{N}_k^p$ to ensure $s_{ed}>0$ shown in eq.~\eqref{eq:ED_s1}, the obtained framework $\left(\mathcal{G}_{ed},\bm{p}_{ed}\right)$ with the stress matrix \eqref{eq:StressMatrixMinus} is universally rigid and affinely localizable.
				\end{coro}
					
				With the help of Corollary~\ref{coro:ED}, we can rearrange the stress to drive the weight of $e_{jk}$ to zero, which is equivalent to removing the edge from the original framework. 
				The key point is to find the eligible vertices in $\mathcal{N}_j^p \cap \mathcal{N}_k^p$. 
				In practice, the number of perceived neighbors is limited so that the local traversal is usually acceptable. For higher efficiency, we can also set a time upper limit on the search for suitable neighbors. If we cannot find proper neighbors within the specified time, a relay strategy is activated. That is, we can eliminate a specific edge by adding new vertices, which are called as ``temporary vertices". 
				
				\textbf{Case (2): There are no vertices in  $ \mathcal{N}_j^p \cap \mathcal{N}_k^p $ to guarantee positive $s_{ed}$ , but $ \mathcal{N}_j^p \cap \mathcal{N}_k^p \ne \emptyset$.} 
				
				In this case, there are no suitable vertices to ensure the positive scaling parameter $s_{ed} $. To complete our edge deletion algorithm, we can eliminate specific edges by adding a new temporary vertex. The following assumption describes the criteria for selecting the temporary vertices.
				
				\begin{assu}\label{assu:Relay_v1}
					Suppose that the vertices $v_i \in \mathcal{N}_j^p \cap \mathcal{N}_k^p $ and a temporary vertex $v_{r}$ meet the following conditions.
					\begin{itemize}
						\item[\labelitemi] $v_i,v_j,v_k \in \mathcal{N}_r^p$, and the set $\left\{\bm{p}_i,~\bm{p}_j,~\bm{p}_k,~\bm{p}_r\right\}$ is in general position;
						\item[\labelitemi] If the weight $\varpi_{jk} = -\Omega_{jk} > 0$, $v_r$ lies in Region $\mathfrak{a},~\mathfrak{d}$ or $\mathfrak{g}$ of a triangle formed by $v_i$, $v_j$ and $v_k$. If the weight $\varpi_{jk} = -\Omega_{jk} < 0$, $v_r$ lies in Region $\mathfrak{b},~\mathfrak{c},~\mathfrak{e}$ or $\mathfrak{f}$, as shown in Fig.~\ref{fig:FourVer}.
					\end{itemize}
				\end{assu}
					
				Under Assumption~\ref{assu:Relay_v1}, we can reconstruct a framework without edge $e_{jk}$ as follows. Remove $e_{jk}$ in the original affine framework $(\mathcal{G},\bm{p})$, and then add $v_r$ along with three edges  $e_{ir},~e_{jr}$ and $e_{kr}$ based on Theorem~\ref{theo:VertexAddition}. Accordingly, a new framework $(\mathcal{G}_{ed},\bm{p}_{ed})$ is built, where $\mathcal{G}_{ed} = \left(\mathcal{V}\cup \left\{v_r\right\},\mathcal{E}\cup \left\{e_{ir},e_{jr},e_{kr}\right\}\setminus\left\{e_{jk}\right\} \right)$, $\bm{p}_{ed} = \left[\bm{p}^T ,\bm{p}_r^T\right]^T$. The appropriate stress exist and is presented in \eqref{eq:ED_reply} to guarantee that $(\mathcal{G}_{ed},\bm{p}_{ed})$ is universally rigid and affinely localizable. 
						\begin{equation}\label{eq:ED_reply}
                            \footnotesize
							\bm{\Omega}_{ed} = \left[\begin{array}{c:c}
								\bm{\Omega} & \bm{0}_{n\times1}\\
								\hdashline
								\bm{0}_{1 \times n} & 0
							\end{array}\right] + \left[\begin{array}{c:c}
								\bm{0}_{(n-3)\times (n-3)} & \bm{0}_{(n-3)\times 4}\\
								\hdashline
								\bm{0}_{4 \times (n-3)} & \bm{\Omega}_u
							\end{array}\right].
						\end{equation}
						where ${s}_{ed} = -\dfrac{\Omega_{jk}}{\bar{\phi}_2 \bar{\phi}_3}$ and $\bm{\Omega}_u ={s}_{ed}  \bm{\bar{\phi}}\bm{\bar{\phi}}^T$. The vector $\bm{\bar{\phi}}=\left[ \bar{\phi}_1,~\bar{\phi}_2,~\bar{\phi}_3,~\bar{\phi}_4\right]^T$ satisfies  $\left[ \begin{array}{cccc}
							\bm{p}_{i} & \bm{p}_j & \bm{p}_k & \bm{p}_{r}\\
							1 & 1 & 1 & 1 
						\end{array}\right]\bm{\bar{\phi}} =\bm{0}$.

					\textbf{Case (3): $\mathcal{N}_j^p \cap \mathcal{N}_k^p = \emptyset$.} 
					
					In this case, there is no vertex in $\mathcal{N}_j^p \cap \mathcal{N}_k^p $ such that two temporary vertices, $v_{r_1}$ and $v_{r_2}$, are needed with the following assumption.
					\begin{assu} 
						Suppose that the vertices $v_{r_1}$ and $v_{r_2}$ meet the following conditions.
						\begin{itemize}
							\item[\labelitemi] Any two vertices of $v_j$, $v_k$, $v_{r_1}$ and $v_{r_2}$ can measure their relative positions to each other. The set $\left\{\bm{p}_j,~\bm{p}_k,~\bm{p}_{r_1},~\bm{p}_{r_2}\right\}$ is in general position;
							\item[\labelitemi] If the weight $\varpi_{jk} = -\Omega_{jk} > 0$, $v_{r_2}$ lies in Region $\mathfrak{a},~\mathfrak{d}$ or $\mathfrak{g}$ of a triangle formed by $v_{r_1}$, $v_j$ and $v_k$. If the weight $\varpi_{jk} = -\Omega_{jk} < 0$, $v_{r_2}$ lies in Region $\mathfrak{b},~\mathfrak{c},~\mathfrak{e}$ or $\mathfrak{f}$, as shown in Fig.~\ref{fig:FourVer}.
						\end{itemize}
					\end{assu}
					
					Choose $v_j$ and $v_k$ to be the neighbors of $v_{r_1}$, and add the vertex $v_{r_1}$ to the original framework $(\mathcal{G},\bm{p})$ by using Algorithm~\ref{Algo:VertexAdd}. Then, we merge $v_{r_2}$ to the framework with three edges, $e_{r_1 r_2},~e_{jr_2}$ and $e_{kr_2}$, based on Theorem \ref{theo:VertexAddition} again. We can set the scaling parameter as ${s}_{ed} = -\dfrac{\Omega_{jk}}{\bar{\phi}_2 \bar{\phi}_3}$, where $\left[ \begin{array}{cccc}
						\bm{p}_{r_1} & \bm{p}_j & \bm{p}_k & \bm{p}_{r_2}\\
						1 & 1 & 1 & 1 
					\end{array}\right] \left[ \bar{\phi}_1,~\bar{\phi}_2,~\bar{\phi}_3,~\bar{\phi}_4\right]^T =\bm{0}$. Consequently, the framework $(\mathcal{G}_{ed},\bm{p}_{ed})$ is constructed, where $\mathcal{G}_{ed} = \left(\mathcal{V}\cup \left\{v_{r_1},v_{r_2}\right\},\mathcal{E}\cup \left\{e_{\sim r_1},e_{jr_2},e_{kr_2},e_{r_1r_2}\right\}\setminus\left\{e_{jk}\right\} \right)$, $\bm{p}_{ed} = \left[\bm{p}^T ,\bm{p}_{r_1}^T,\bm{p}_{r_2}^T\right]^T$. In \textbf{Case (3)}, the vertex $v_{r_1}$ plays the role similar to $v_i$ in \textbf{Case (2)}. Thus, the similar steps and analysis is omitted.
				
				The edge deletion strategies proposed in different cases share the same fundamental idea presented in Corollary~\ref{coro:ED}, with the differences lying in the selection of neighbor vertices and the design of scaling parameter $s_{ed}$.
				To explain our proposed edge deletion strategy more clearly, an example is presented as below.

				\begin{exam}[Edge Deletion] To remove $e_{23}$ in Fig.~\ref{fig:EdgeDeletion_O2} and ~\ref{fig:EdgeDeletion_O1}, two strategies mentioned above are executed, respectively. The coordinates of the five vertices are $\left[0,1\right]^T,\left[1,0\right]^T,\left[0,-1\right]^T,\left[-1,0\right]^T,\left[1,-1\right]^T$. Suppose that any two vertices can be connected. In Fig.~\ref{fig:EdgeDeletion_O2}, $j=2$, $k=3$ so that $\mathcal{N}_j^p \cap \mathcal{N}_k^p = \{1,4,5\}$. To delete the edge $e_{23}$, let $i=1$ and $q=5$. We have $\bm{\phi} = \frac{1}{\sqrt{10}}\left[1,-2,-1,2\right]^T$ so that $s_{ed} = -\dfrac{\Omega_{23}}{\phi_2  \phi_3 } = 4  >0$. Thus, Corollary~\ref{coro:ED} is applied and the obtained affine  framework is shown in Fig.~\ref{fig:EdgeDeletion_D2}. Nevertheless, in Fig.~\ref{fig:EdgeDeletion_O1}, $\mathcal{N}_j^p \cap \mathcal{N}_k^p = \{1,4\} \ne \emptyset$. If $i=1$ and $q=4$, we can calculate $s_{ed}= - 4 < 0$, which implies that there is no suitable vertices to ensure a positive scaling parameter,  indicating that a temporary vertex is necessary. Thus, following \textbf{Case (2)}, choose a temporary vertex $v_5$ with $\bm{p}_5 = [1,-1]^T$, lying in Region $\mathfrak{d}$. Three edges, $e_{15},~e_{25},~e_{35}$, are added to $\left(\mathcal{G},\bm{p}\right)$ with ${s}_{ed}=2$, resulting $\varpi_{23} = 0$ after the operation. The obtained affine framework is also shown in Fig.~\ref{fig:EdgeDeletion_D2}.
					
					\begin{figure}
						\centering
						\subfigure[]{\includegraphics[scale=0.18]{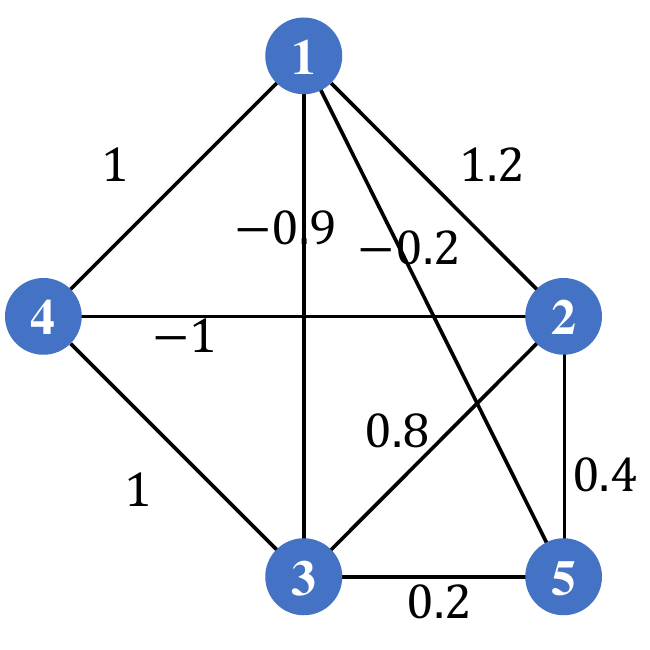} 
							\label{fig:EdgeDeletion_O2}
						}
						\subfigure[]{\includegraphics[scale=0.18]{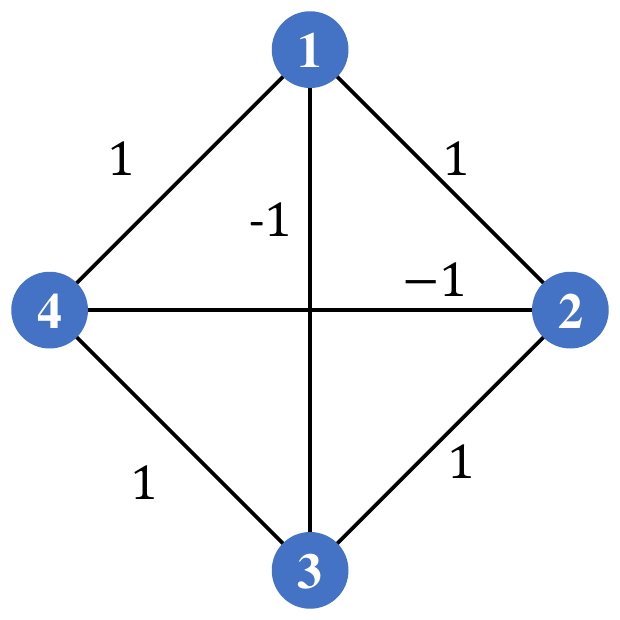} 
							\label{fig:EdgeDeletion_O1}
						}
						\subfigure[] {\includegraphics[scale=0.18]{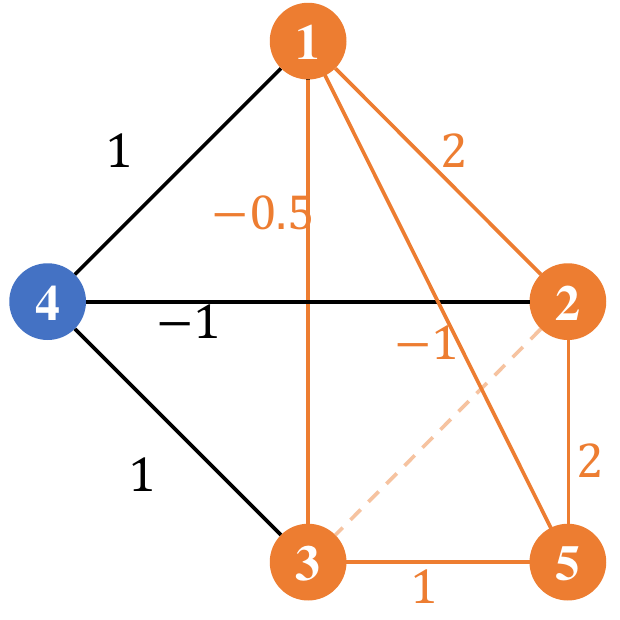}
							\label{fig:EdgeDeletion_D2}
						}
						\caption{An example of edge deletion. (a) Original affine framework $\left(\mathcal{G},\bm{p}\right)$ with five vertices. 
                        (b) Original affine framework $\left(\mathcal{G},\bm{p}\right)$ with four vertices. 
                        (c) The obtained framework $\left(\mathcal{G}_{ed},\bm{p}_{ed}\right)$ without the edge $e_{23}$.}
						\label{fig:EdgeDeletion}
					\end{figure}
					
				\end{exam}
				
				\begin{rema}
					The edge deletion strategy using temporary vertices in \textbf{Case (2) and (3)} is directly inspired by Theorem~\ref{theo:VertexAddition}. To ensure a positive scaling parameter, the positions of temporary vertices need to be carefully selected, taking Table \ref{tab:Symbol_2D} as a reference.
                    In this paper, the temporary vertex is assumed to be on standby and respond immediately. Nevertheless, in a real robot system, more technical details need to be designed to guarantee the smooth operation of temporary robots, which is beyond the scope of this paper.
				\end{rema}

			\subsection{Vertex Deletion }\label{sec:VertexDeparture}
			In this subsection, we consider an inverse operation of ``vertex addition" presented in Section~\ref{sec:VertexAddition}, the removal of a vertex from a HAF $\left(\mathcal{G},\bm{p}\right)$, constructed by Algorithm~\ref{Algo:VertexAdd}. 
            %
			According to whether a vertex in $\left(\mathcal{G},\bm{p}\right)$ has a child or not, the vertices in the graph can be classified into two categories: (a) \emph{outer node}, a vertex with no children (e.g., $v_q$ in Fig.~\ref{fig:VertexAdd}); (b) \emph{inner node}, a vertex with at least one child (e.g., $v_u$ in Fig.~\ref{fig:VertexAdd}). Next, we discuss the removal of the outer and inner nodes from $\left(\mathcal{G},\bm{p}\right)$, respectively. 
			
			\textbf{Case 1 (Deletion of an Outer Node):} 
			%
			We first consider the case with an outer node. Consider a HAF $\left(\mathcal{G},\bm{p}\right)$ with $n~\left(n \ge 5\right)$ vertices, containing an outer node $v_u$ with $h(v_u)>0$. The vertex $v_u$ to be deleted has three parents, namely $v_i,v_j,v_k$. 
            Accordingly, the stress matrix for $\left(\mathcal{G},\bm{p}\right)$ can be described as follows.
			\begin{equation}
				\footnotesize
				\begin{aligned}
					\bm{\Omega} 
						=& \left[\begin{array}{c:ccc:c}
							\bm{\Omega}_{vd}^{B1}& & \bm{\Omega}_{vd}^{B2} &  & \bm{0}_{(n-4)\times 1}\\
							\hdashline
							&\Omega_{i i} & \Omega_{ij} & \Omega_{i k} & \Omega_{i u}\\
							\left(\bm{\Omega}_{vd}^{B2}\right)^T& \Omega_{i j}& \Omega_{jj} & \Omega_{j k} & \Omega_{j u}\\
							& \Omega_{i k}& \Omega_{jk} & \Omega_{k k} & \Omega_{k u}\\
							\hdashline
							\bm{0}_{1\times (n-4)} &\Omega_{i u} & \Omega_{j u} &\Omega_{k u} & \Omega_{u u}
						\end{array}\right],\\
					\end{aligned}
				\end{equation}
				%
				Notably, $\Omega_{u u}>0$. Based on the definition of an equilibrium stress, we have the following equation for $v_u$,
				\begin{equation}\label{eq:VDequal}
					\footnotesize
					\begin{aligned}
						&\Omega_{i u}\left( \bm{p}_i - \bm{p}_u \right) + \Omega_{j u}\left( \bm{p}_j - \bm{p}_u \right) + \Omega_{k u}\left( \bm{p}_k - \bm{p}_u \right) =\bm{0}\\
						\Rightarrow & \underbrace{\left[\begin{array}{cccc}
								\bm{p}_i & \bm{p}_j & \bm{p}_k & \bm{p}_u\\
								1 & 1 & 1 & 1
							\end{array}\right]}_{\triangleq \bm{P}_u} \underbrace{\left[\begin{array}{cccc}
								\Omega_{iu} &
								\Omega_{ju} &
								\Omega_{ku} &
								\Omega_{uu} 
							\end{array}\right]^T}_{\triangleq \bm{\Omega}_\phi} = \bm{0}.
					\end{aligned}
				\end{equation}
				Since $\left(\mathcal{G},\bm{p}\right)$ is in general position, it is deduced that $\operatorname{rank}\left(\bm{P}_u\right) =3$ and $\operatorname{nullity}\left(\bm{P}_u\right) = 1$. 

            After removing vertex $v_u$ and all edges connected to $v_u$, we can obtain a framework $\left(\mathcal{G}_{vd},\bm{p}_{vd}\right)$. The stress matrix $\bm{\Omega}_{vd}$ for the framework $\left(\mathcal{G}_{vd},\bm{p}_{vd}\right)$ is defined as below. 
                \begin{equation}\label{eq:StressMatrixVertexDel_2}
                \footnotesize
					\left[\begin{array}{cc}
						\bm{\Omega}_{vd} & \bm{0} \\
						\bm{0} & 0
					\end{array}\right]   = \bm{\Omega} - \left[\begin{array}{cc}
						\bm{0} & \bm{0} \\
						\bm{0} & \bm{\Omega}_{vd}^u
					\end{array}\right] ,
				\end{equation}
                where $\bm{\Omega}_{vd}^{u} = \dfrac{1}{\Omega_{u u} }\bm{\Omega}_\phi \bm{\Omega}_\phi^T \succeq 0$. 
		%
				Thus,
				\begin{equation}\label{eq:StressMatrixVertexDel}
					\tiny
					\bm{\Omega}_{vd} = \left[\begin{array}{c:ccc}
						\bm{\Omega}_{vd}^{B1} & & \bm{\Omega}_{vd}^{B2} &\\
						\hdashline
						& \Omega_{i i} - \dfrac{\Omega_{i u}^2}{\Omega_{u u}} &\Omega_{i j} - \dfrac{\Omega_{i u} \Omega_{j u}}{\Omega_{u u}} & \Omega_{i k} -\dfrac{\Omega_{i u} \Omega_{k u}}{\Omega_{u u}} \\
						\left(\bm{\Omega}_{vd}^{B2} \right)^T& \Omega_{i j} -\dfrac{\Omega_{i u} \Omega_{j u}}{\Omega_{u u}} & \Omega_{jj} -\dfrac{ \Omega_{j u}^2}{\Omega_{u u}} & \Omega_{jk} -\dfrac{\Omega_{j u} \Omega_{k u}}{\Omega_{u u}} \\
						& \Omega_{i k} -\dfrac{\Omega_{i u} \Omega_{k u}}{\Omega_{u u}} & \Omega_{jk} -\dfrac{\Omega_{j u} \Omega_{k u}}{\Omega_{u u}} & \Omega_{kk} -\dfrac{\Omega_{k u}^2}{\Omega_{u u}} 
					\end{array}\right].
				\end{equation}
			where $\bm{\Omega}_{vd} \in \mathbb{R}^{(n-1)\times(n-1)}$. We establish the following theorem to prove $\left(\mathcal{G}_{vd},\bm{p}_{vd}\right)$ is universally rigid and affinely localizable.
			
			\begin{theo}\label{theo:VD}
                
				Consider a HAF $\left(\mathcal{G},\bm{p}\right)$ with $n~\left(n \ge 5\right)$ vertices. 
                If $v_u$ is an outer node and $h\left(v_u\right)>0$, we can obtain a universally rigid and affinely localizable framework $\left(\mathcal{G}_{vd},\bm{p}_{vd}\right)$ with a stress matrix $\bm{\Omega}_{vd}$ shown in \eqref{eq:StressMatrixVertexDel} after removing vertex $v_u$ and all edges connected to $v_u$. 
			\end{theo}
			
			\begin{proof}
                For an outer node $v_u$ in a HAF, we have $\mathcal{G}_{vd} = (\mathcal{V}\setminus \{v_u\},\mathcal{E}\setminus \{e_{iu},e_{ju},e_{ku}\})$ after removing vertex $v_u$ and all edges connected to $v_u$. Obviously, the obtained framework $\left(\mathcal{G}_{vd},\bm{p}_{vd}\right)$ is in general position. To prove  $\left(\mathcal{G}_{vd},\bm{p}_{vd}\right)$ is universally rigid and affinely localizable, we need to clarify that $\bm{\Omega}_{vd} \succeq 0$, $\operatorname{rank}\left(\bm{\Omega}_{vd}\right) = n-4$ and the block matrix describing the stress among followers $\bm{\Omega}_{ff}^{vd} \succ 0$. The following proof has a similar sketch with Theorem~\ref{theo:VertexAddition} in Appendix~\ref{sec:app}. 
				
                First, $\operatorname{rank}\left(\bm{\Omega}_{vd}\right) = n - 4$ should be explained by analyzing the null space of $\bm{\Omega}_{vd}$. Based on the fact that $\bm{P}_u\bm{\Omega}_{vd}^{u} =\bm{0}$ and $\operatorname{null}(\bm{\Omega})= \operatorname{col}\left(\left[\begin{array}{cccccc}
					\bm{p}_1 & \cdots & \bm{p}_i & \bm{p}_j & \bm{p}_k & \bm{p}_u \\
					1 & \cdots & 1 & 1 & 1 & 1 \\
				\end{array}\right]^T\right)$, we can obtain that $\operatorname{null}(\bm{\Omega}_{vd
                })= \operatorname{col}\left(\left[\begin{array}{ccccc}
					\bm{p}_1 & \cdots & \bm{p}_i & \bm{p}_j & \bm{p}_k  \\
					1 & \cdots & 1 & 1 & 1 \\
				\end{array}\right]^T\right)$ according to eq.~\eqref{eq:StressMatrixVertexDel_2}. Thus, it is deduced that $\operatorname{nullity}(\bm{\Omega}_{vd
                }) = 3$ so that $\operatorname{rank}\left(\bm{\Omega}_{vd}\right) = n - 1-3 =n-4$ based on Lemma~\ref{lemm:RN}.
                
                Then, the analysis of semi-positive definiteness of symmetric matrix $\bm{\Omega}_{vd}$ can be performed by repeatedly applying Lemma~\ref{lemma:PD}, which is similar to the operation of Appendix~\ref{sec:app}. The details are omitted to enhance brevity. The proof of the positive definiteness of $\bm{\Omega}^{vd}_{ff}$ follows similar steps.
				
                Consequently, we conclude that $\left(\mathcal{G}_{vd},\bm{p}_{vd}\right)$ is universally rigid and affinely localizable. 
			\end{proof}
			
			By directly applying Theorem~\ref{theo:VD}, we can remove outer nodes without damaging the universal rigidity and affine localizability of the remaining framework. Then, consider the deletion operation of inner nodes.
			
			
			\textbf{Case 2 (Deletion of an Inner Node):} 
			When an inner node $v_u ~(h(v_u) > 0)$ leaves the HAF, the universal rigidity and affine localizability of the framework are destroyed and need to be repaired. Hence, a strategy inspired by inheritance is derived. Denote the parents of $v_u$ as $\mathcal{P}_u =\{v_{p_i}^u,~v_{p_j}^u,~v_{p_k}^u\}$ and the children as $\mathcal{C}_u =\{v_{c_1}^u,~v_{c_2}^u\cdots v_{c_m}^u\}$, where $h\left(v_{c_1}^u\right) \le \cdots \le h\left(v_{c_m}^u \right) $.
			
			\begin{assu}\label{assu:InnerNode}
                Suppose that $v_{\mathcal{P}} \in \mathcal{N}_{v_\mathcal{C}}^p$, where $v_{\mathcal{P}}  \in \mathcal{P}_u$ and $v_{\mathcal{C}}  \in \mathcal{C}_u$, which means edges can be established between the parents and children of $v_u$. 
			\end{assu}
			
            Based on Theorem~\ref{theo:VD}, a matrix $\bm{\Omega}_{vd}^u$ in \eqref{eq:StressMatrixVertexDel_2} is constructed to update the stress matrix when deleting an outer node. Accordingly, denote the auxiliary matrix to delete $v_u,~v_{c_1}^u\cdots v_{c_m}^u$ as $\bm{\Omega}_b^u,~\bm{\Omega}_b^{uc_1}\cdots \bm{\Omega}_b^{uc_m} \in \mathbb{R}^{n \times n}$,  which can be calculated by the relative position of each vertex and its parents. The order of matrix is unified by inserting zeros. The steps of removing an inner node $v_u$ is presented as follows.
			\begin{itemize}
				\item[S1.] Remove 
				all edges connecting $v_u$, and update the stress matrix as $
					\bm{\Omega}_{cache} = \bm{\Omega} - \bm{\Omega}_b^u-\bm{\Omega}_b^{uc_1}-\cdots -\bm{\Omega}_b^{uc_m} $.
				Obviously, the elements in the $u$-th row and $u$-th column of $\bm{\Omega}_{cache}$ become zero. Delete the $u$-th row and the $u$-th column, and a simplified matrix is shown as $\tilde{\bm{\Omega}}_{cache}\in \mathbb{R}^{(n-1) \times (n-1) }$.
				\item[S2.] The vertex $v_{c_1}$ 
                is chosen to inherit the role of $v_u$ in the HAF. Choose one vertex from $\mathcal{P}_u \setminus \mathcal{P}_{c_1} $, labeled by $v_{p_1}^{c_1}$ to rebuild the parent set of $v_{c_1}$, $\bar{\mathcal{P}}_{c_1}$, with $v_{p_1}^{c_1}$ and the remaining vertices in $\mathcal{P}_{c_1}$. Using the positions of $v_{c_1}$ and its parents, construct the auxiliary matrix $\bm{\Omega}_b^{c_1} \in \mathbb{R}^{(n-1) \times (n-1) }$;
				\item[S3.] Rebuild the parent sets of other vertices in $\mathcal{C}_u \setminus \{v_{c_1}\}$. Taking the vertex $v_{c_2}$ as an example, if $v_{c_1}$ is not a parent of $v_{c_2}$ (i.e., $\mathcal{P}_{c_2} =\{v_u,~v_{p_j}^{c_2},~v_{p_k}^{c_2}\}$), the new parent set is $\bar{\mathcal{P}}_{c_2} =\{v_{c_1},~v_{p_j}^{c_2},~v_{p_k}^{c_2}\}$, where $v_u$ is replaced by $v_{c_1}$. If $v_{c_1} \in \mathcal{P}_{c_2}$, we can choose one vertex in $\mathcal{P}_u \setminus \mathcal{P}_{c_2}$ to replace $v_u$, e.g., $\bar{\mathcal{P}}_{c_2} =\{v_{p_i}^u,~v_{c_1},~v_{p_k}^{c_2}\}$. The remaining children of $v_u$ follow the strategy.
				\item[S4.] Construct auxiliary matrices for vertices in $\mathcal{C}_u \setminus \{v_{c_1}\}$ using their own and parental positions, i.e.,  $\bm{\Omega}_b^{c_1c_2},\cdots,\bm{\Omega}_b^{c_1c_m}\in \mathbb{R}^{(n-1) \times (n-1) }$ (the order is unified by inserting zero). Accordingly, a new affine framework $\left(\mathcal{G}_{vd},\bm{p}_{vd}\right)$ is reconstructed without $v_u$, and the corresponding stress matrix is 
				\begin{equation}\label{eq:SM_VD_inhe}
                \footnotesize
					\bm{\Omega}_{vd} = \tilde{\bm{\Omega}}_{cache} + \bm{\Omega}_b^{c_1c_2} +\cdots+\bm{\Omega}_b^{c_1c_m}.
				\end{equation}
			\end{itemize}
			The universal rigidity and affine localizability of $\left(\mathcal{G}_{vd},\bm{p}_{vd}\right)$ can be proved by analyzing the characteristic of $\bm{\Omega}_{vd}$ shown in Eq.~\eqref{eq:SM_VD_inhe}. Combining the proof sketch of Theorem~\ref{theo:VertexAddition} and \ref{theo:VD}, a similar process can be implemented, which is omitted here. To explain our method more clearly, an example is presented.
			
			\begin{exam}[Vertex Deletion]
				Using Algorithm~\ref{Algo:VertexAdd}, a HAF is generated as shown in Fig.~\ref{fig:Exam_VD_Ori}, where $\bm{p} = \left[8,0,0,8,-8,0,0,-8,9,-10,0,-12,11,1,14,-14,-7,-5\right]^T$, $v_7,~v_8,~v_9$ are outer nodes and $v_5,~v_6$ are inner nodes. To remove an outer node (e.g., $v_9$), Theorem~\ref{theo:VD} is directly applied, and the weights are updated as shown in Fig.~\ref{fig:Exam_VD_ON}. The eigenvalues of the stress matrix are $0,0,0,0.107,0.398,0.91,1.82,4.77$, which demonstrate the universal rigidity and affine localizability of the obtained framework without $v_9$. To verify the strategy of deleting an inner node, take the removal of $v_5$ as an example. After removing all edges connected to $v_5$, the framework is shown in Fig.~\ref{fig:Exam_VD_INC}. Obviously, the rigidity of HAF is destroyed. To repair the universal rigidity and affine localizability, the second and third step, S2 and S3, are executed. We choose $v_6$ to inherit the role of $v_5$, so that the hierarchical relationship between vertices is reconstructed, as presented in Tab.~\ref{tab:Hierarchy}. The final framework is shown in Fig.~\ref{fig:Exam_VD_IN}, and the eigenvalues of the corresponding stress matrix obtained from \eqref{eq:SM_VD_inhe} are $0,0,0,0.204,0.36,1.19,1.38,4.87$. Accordingly, the effectiveness of our proposed approached in reconstructing a HAF has been validated.
				
				\begin{figure}[htbp]
					\centering
					\subfigure[]{\includegraphics[scale=0.24]{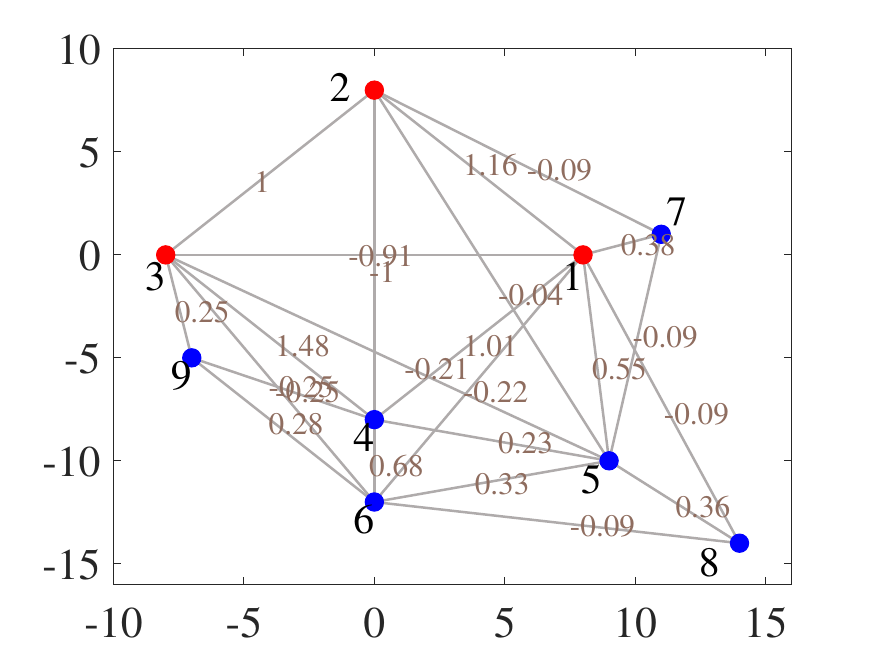} 
						\label{fig:Exam_VD_Ori}
					}
					\subfigure[]{\includegraphics[scale=0.24]{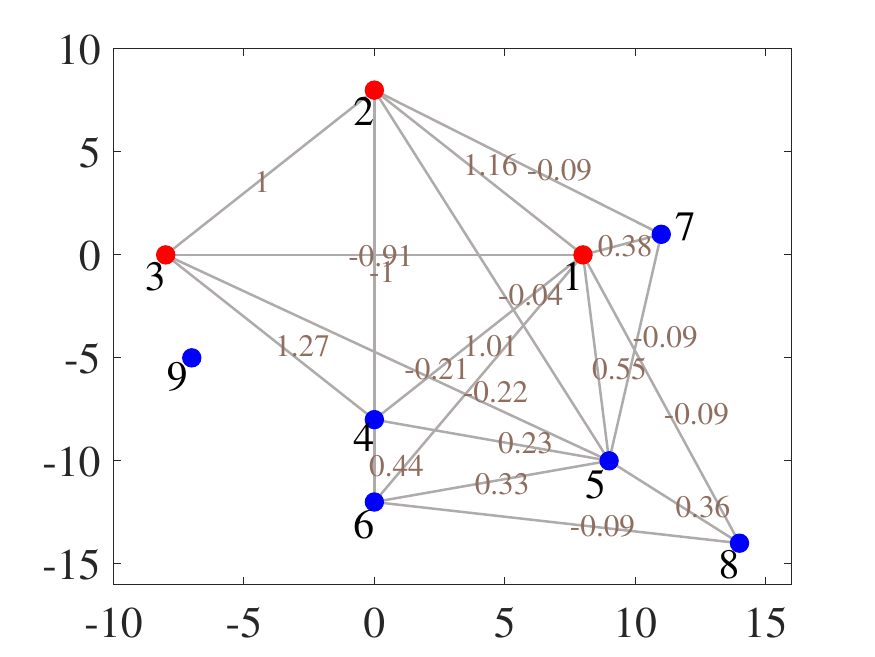} 
						\label{fig:Exam_VD_ON}
					}
					\subfigure[]{\includegraphics[scale=0.24]{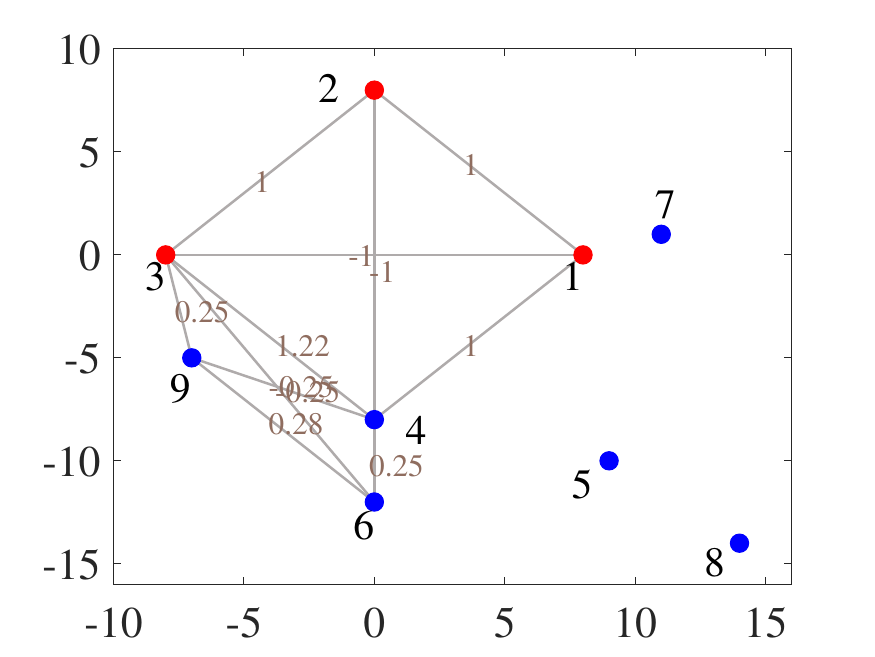} 
						\label{fig:Exam_VD_INC}
					}
					\subfigure[]{\includegraphics[scale=0.24]{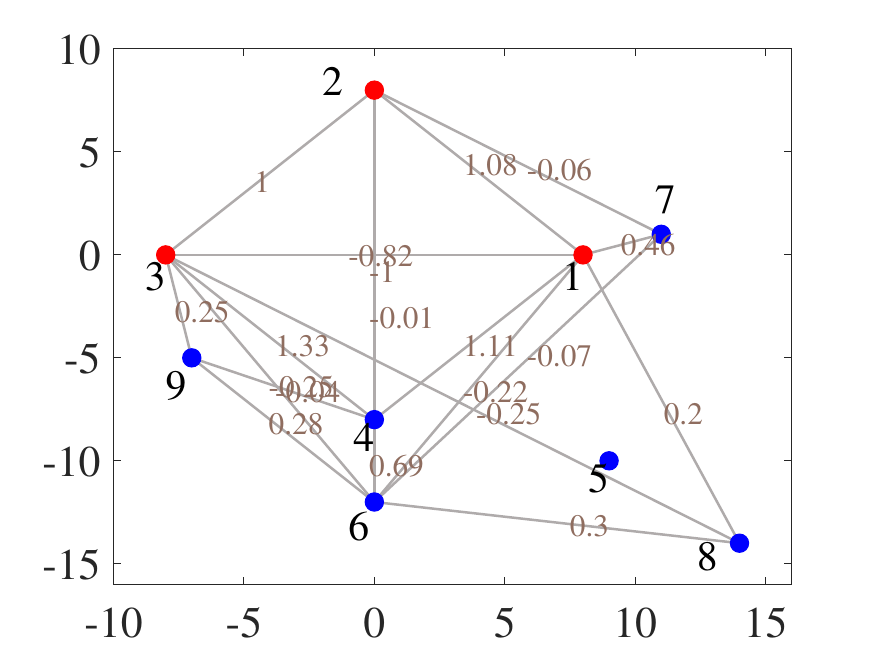} 
						\label{fig:Exam_VD_IN}
					}
					\caption{An example of the vertex deletion algorithm. (a) The original HAF, where $v_1,v_2,v_3 \in \mathcal{V}_l$. (b) The removal of an outer node $v_9$. (c) The removal of an inner node $v_5$: delete all edges connected to $v_5$. (c) The removal of an inner node $v_5$: repair the hierarchy.}
					\label{fig:Exam_VD}
				\end{figure}
				
				\begin{table*}\label{tab:Hierarchy}
					\centering
					\caption{An example of reconstructing the hierarchical structure.}
					\begin{threeparttable}
						\begin{tabular}{c|ccc|ccc}
							\toprule 
							\multirow{2}{*}{Vertex}& \multicolumn{3}{c|}{Before Operation} &\multicolumn{3}{c}{After Operation}\\
							\cmidrule{2-7}
							& Hierarchy$h(\cdot)$ & Parents & Children  & Hierarchy$h(\cdot)$ & Parents & Children\\
							\midrule
							$v_1,~v_2,~v_3,~v_4$ & $0$ & None & ---  & $0$ &  None & --- \\
							$v_5$ & $1$ & $v_1,~v_3,~v_4$ & $v_6,~v_7,~v_8$ & None &  None & None \\
							$v_6$ & $2$ & $v_1,~v_4,~v_5$ & $v_8,~v_9$ & $1$ &  $v_1,~v_3,~v_4$ & $v_7,~v_8,~v_9$ \\
							$v_7$ & $2$ & $v_1,~v_2,~v_5$ & None & $2$ &  $v_1,~v_2,~v_6$ & None\\
							$v_8$ & $3$ & $v_1,~v_5,~v_6$ & None & $2$ &  $v_1,~v_3,~v_6$ & None \\
							$v_9$ & $3$ & $v_3,~v_4,~v_6$ & None & $2$ &  $v_3,~v_4,~v_6$ & None \\
							\bottomrule
						\end{tabular}
						\label{tab:Exam_VD}
						\begin{tablenotes}    
							\footnotesize               
							\item[]``None" represents no data; ``---" represents omitting data.
						\end{tablenotes} 
					\end{threeparttable}
				\end{table*}
				
			\end{exam}

            \begin{rema}
                The affine framework generated by Algorithm~\ref{Algo:VertexAdd} has a clear hierarchical structure with specific parents and children, which is convenient for reconstructing the connection relationship after deleting vertices. For vertex $v_u$ with $h\left(v_u\right)>0$, there are only three parents. 
                As shown in eqs.~\eqref{eq:Omega_Plus}-\eqref{eq:Omega_u1}, the stress update between $v_u$ and its parents follows a specific rule. By reverse-using this rule and updating the stress matrix with negative weights, we can remove the connection edge between $v_u$ and its parents/children. The result is shown in Fig.~\ref{fig:Exam_VD_INC}, where all edges of inner node $v_5$ are disconnected. The rigid reconstruction methods in steps S3 and S4 are similar to Algorithm~\ref{Algo:VertexAdd}, which is equivalent to re-add the vertices in set $\mathcal{C}_u$ to the HAF.
            \end{rema}

                The proposed vertex deletion strategy is also distributed, as it only utilizes local measurements to change the local structure of the affine framework.	Taking the deletion of a inner node shown in Fig.~\ref{fig:Exam_VD_INC} and \ref{fig:Exam_VD_IN} as an example, the connections among the inner node's parents and children are reconstructed, and the remaining part in the framework is preserved. 	

				\section{Extension of affine framework construction in 3D}\label{sec:Discuss}
				The strategies proposed in Section~\ref{sec:MainResult} concentrate on the construction of HAFs in $\mathbb{R}^2$. Actually, our method can be extended to three-dimensional scenarios following the same idea.  The following text briefly analyzes three affine framework reconstruction strategies one by one.

                \textbf{(1) Vertex Addition in $\mathbb{R}^3$.} 
                Given an original affine framework meeting Assumption~\ref{assu:OriginalFrame}, a new vertex $v_u$ is intended to be added to $\left(\mathcal{G},\bm{p}\right)$ in three-dimensional space. Inspired by Theorem~\ref{theo:VertexAddition}, four edges are introduced to the original framework $\left(\mathcal{G},\bm{p}\right)$, and the corresponding stress is determined by the position of the new vertex $v_u$ with respect to the four affinely independent vertices in $\mathcal{N}^p_u$, denoted by $v_i, ~v_j,~v_k$ and $v_g$. Similar to Section~\ref{sec:VertexAddition}, suppose that the set $\left\{\cdots,\bm{p}_i,~\bm{p}_j,~\bm{p}_k,~\bm{p}_g,~\bm{p}_u\right\}$ is in general position.
				
				To grow the affine framework with a new vertex $v_u$, four edges $e_{iu},~e_{ju},~e_{ku},~e_{gu}$ are linked to $\left(\mathcal{G},\bm{p}\right)$ to obtain a extended framework $\left(\mathcal{G}^+,\bm{p}^+\right)$, where $\bm{p}^+ = \left[\bm{p}^T,~\bm{p}_u^T\right]^T$ and $\mathcal{G}^+ = \left(\mathcal{V}\cup \left\{v_u \right\},\mathcal{E}\cup \left\{e_{iu},~e_{ju},~e_{ku},~e_{gu}\right\}\right)$. 
				For the new vertex $v_u$, we have the following equation due to the definition of the equilibrium stress,
				\begin{equation}\label{eq:ES_3D}
                \footnotesize
					\begin{aligned}
						&\varpi_{iu}\left(\bm{p}_u-\bm{p}_i\right) + \varpi_{ju}\left(\bm{p}_u-\bm{p}_j\right) \\
						+& \varpi_{ku}\left(\bm{p}_u-\bm{p}_k\right) + \varpi_{gu}\left(\bm{p}_u-\bm{p}_g\right) = \bm{0},
					\end{aligned}
				\end{equation}
				where $\varpi_{iu},\varpi_{ju},\varpi_{ku},\varpi_{gu}$ are the weights of edges $e_{iu},e_{ju},e_{ku},e_{gu}$. Define $ \bm{P}_u \triangleq \left[ \begin{array}{ccccc}
						\bm{p}_i & \bm{p}_j & \bm{p}_k & \bm{p}_g &  \bm{p}_u\\
						1 & 1 & 1 & 1 & 1
					\end{array} \right]$. 
				Since $\operatorname{rank}(\bm{P}_u) = 4$, there is a nonzero vector $\bm{\phi} = \left[\phi_1,\phi_2,\phi_3,\phi_4,\phi_5\right]^T$ satisfying $\bm{P}_u\bm{\phi} = \bm{0}$. Similarly, define $\bm{\Omega}_u = s \bm{\phi}\bm{\phi}^T$ with $s>0$. Accordingly, the weights can be represented by $\varpi_{iu} = -s\phi_1\phi_5,~\varpi_{ju} = -s\phi_2\phi_5,~\varpi_{ku} = -s\phi_3\phi_5$ and $\varpi_{gu} = -s\phi_4\phi_5$. The augmented stress matrix $\bm{\Omega}^+$ can be rewritten as
				\begin{equation}\label{eq:VA_three}
					\footnotesize
					\bm{\Omega}^+ = \left[\begin{array}{c:c}
						\bm{\Omega} & \bm{0}_{n\times1}\\
						\hdashline
						\bm{0}_{1 \times n} & 0
					\end{array}\right] + \left[\begin{array}{c:c}
						\bm{0}_{(n-4)\times (n-4)} & \bm{0}_{(n-4)\times 5}\\
						\hdashline
						\bm{0}_{5 \times (n-4)} & \bm{\Omega}_u
					\end{array}\right].
				\end{equation}
				Correspondingly, we have the following result on vertex addition in $\mathbb{R}^3$. 
				\begin{coro}\label{coro:VA_three}
%
                    Under Assumption~\ref{assu:OriginalFrame}, consider a new vertex $v_u$ and an affine framework $\left(\mathcal{G},\bm{p}\right)$ in $\mathbb{R}^3$. Suppose that $\mid \mathcal{N}_u^p\mid \ge 4$, and  the set $\left\{\bm{p}_u\right\} \cup \left\{\bm{p}_j:j \in \mathcal{V} \right\}$ is in general position. With a positive $s$, after adding three edges connecting the vertex $v_u$ and the existing vertices $v_i$, $v_j, v_k, v_g \in \mathcal{N}_u^p$ to $\left(\mathcal{G},\bm{p}\right)$, the obtained framework $\left(\mathcal{G}^+,\bm{p}^+\right)$ is universally rigid and affinely localizable, with the corresponding stress matrix determined by \eqref{eq:VA_three}.
				\end{coro}
				
				The similar scheme employed in the proof of Theorem~\ref{theo:VertexAddition} can be applied to explain Corollary~\ref{coro:VA_three}, which is omitted here for simplicity. Obviously, the affine framework generated in $\mathbb{R}^3$ according to Corollary~\ref{coro:VA_three} also has a clear hierarchical structure. For node $v_i$ with $h(v_i) > 0 $, there are four and only four parents. Thus, the graph $ \mathcal{G}^+ $ has a 4-lateration graph as the spanning subgraph. Leveraging the well-defined structure of the hierarchical affine framework, the strategies for edge and vertex deletion in $ \mathbb{R}^3 $ can be extended along the lines of their counterparts in two-dimensional space.
                
				
			\textbf{(2) Edge Deletion in $ \mathbb{R}^3 $.} 
                Assuming that the four neighbors of node $v_u $ are $v_i,v_j,v_k,v_g $, edge $e_ {jk} $ in a HAF $(\mathcal{G},\bm{p})$ is going to be deleted. The entries involved in these five vertices in the stress matrix $\bm{\Omega}$ can be described as below.
			\begin{equation}
				\footnotesize
				\bm{\Omega}_v=\left[\begin{array}{ccccc}
					\Omega_{i i}& \Omega_{i j}& \Omega_{i k} & \Omega_{i g} & \Omega_{i u}\\
					\Omega_{i j}& \Omega_{j j}& \Omega_{j k} & \Omega_{j g} & \Omega_{j u}\\
					\Omega_{i k} & \Omega_{j k} & \Omega_{k k} & \Omega_{k g} & \Omega_{j u}\\
					\Omega_{i g} & \Omega_{j g} & \Omega_{k g} & \Omega_{g g}& \Omega_{g u}\\
					\Omega_{i u} & \Omega_{j u} & \Omega_{k u} & \Omega_{g u}& \Omega_{u u}
				\end{array}\right].
			\end{equation}
			After eliminating $e_{jk}$, the framework is described by $\left(\mathcal{G}^-,\bm{p}\right)$, and the corresponding block matrix is established as follows.
			\begin{equation}
				\footnotesize
				\bm{\Omega}_v^-=\left[\begin{array}{ccccc}
					\Omega_{i i}^-& \Omega_{i j}^-& \Omega_{i k}^- & \Omega_{i g}^- & \Omega_{i u}^-\\
					\Omega_{i j}^-& \Omega_{j j}^-& 0 & \Omega_{j g}^- & \Omega_{j u}^-\\
					\Omega_{i k}^- & 0 & \Omega_{k k}^- & \Omega_{k g}^- & \Omega_{j u}^-\\
					\Omega_{i g}^- & \Omega_{j g}^- & \Omega_{k g}^- & \Omega_{g g}^- & \Omega_{g u}^-\\
					\Omega_{i u}^- & \Omega_{j u}^- & \Omega_{k u}^- & \Omega_{g u}^- & \Omega_{u u}^-
				\end{array}\right].
			\end{equation}
			
			According to the definition of equilibrium stress, we have 
			\begin{equation}\label{eq:ED_equa}
                \footnotesize
				\underbrace{
                \left[ \begin{array}{ccccc}
					\bm{p}_i & \bm{p}_j & \bm{p}_k & \bm{p}_g & \bm{p}_u\\
					1 & 1 & 1 & 1 & 1
				\end{array} \right]}_{\triangleq\bm{P}_{u}^{vd}} \left( \bm{\Omega}_v^- -\bm{\Omega}_v \right) =\bm{0}.
			\end{equation}
		Due to $\operatorname{rank}(\bm{P}_{u}^{vd})=4$, we have $\operatorname{nullity}(\bm{P}_{u}^{vd}) = 1$, so there exists a non-zero vector  $\bm{\phi} = \left[\phi_1,\cdots,\phi_5\right]^T$ that satisfies $\bm{P}_{u}^{vd} \bm{\phi} = \bm{0}$. Let $\bm{\Omega}_u  \triangleq \bm{\Omega}_v^- -\bm{\Omega}_v = s_{ed} \bm{\phi}\bm{\phi}^T$. Thus, $\bm{\Omega}_v^- = \bm{\Omega}_v + \bm{\Omega}_u$, where 
			\begin{equation}\label{eq:omega_3D}
            \footnotesize
				\begin{aligned}
					&\bm{\Omega}_u =\\
							&
							\left[\begin{array}{ccccc}
								\dfrac{\varpi_{i u}^2}{\varpi_{uu}} &  \dfrac{\varpi_{i u}\varpi_{j u}}{\varpi_{uu}} &  \dfrac{\varpi_{i u}\varpi_{k u}}{\varpi_{uu}}&  \dfrac{\varpi_{i u}\varpi_{gu}}{\varpi_{uu}}& -\varpi_{i u}\\
								\dfrac{\varpi_{i u}\varpi_{j u}}{\varpi_{uu}}&   \dfrac{\varpi_{j u}^2}{\varpi_{uu}}&  \dfrac{\varpi_{j u}\varpi_{k u}}{\varpi_{uu}}& \dfrac{\varpi_{j u}\varpi_{gu}}{\varpi_{uu}}& -\varpi_{j u}\\
								\dfrac{\varpi_{i u}\varpi_{k u}}{\varpi_{uu}}&  \dfrac{\varpi_{j u}\varpi_{k u}}{\varpi_{uu}}&   \dfrac{\varpi_{k u}^2}{\varpi_{uu}}& \dfrac{\varpi_{k u}\varpi_{gu}}{\varpi_{uu}}& -\varpi_{k u}\\
								\dfrac{\varpi_{i u}\varpi_{gu}}{\varpi_{uu}}&  \dfrac{\varpi_{j u}\varpi_{gu}}{\varpi_{uu}}&  \dfrac{\varpi_{k u}\varpi_{gu}}{\varpi_{uu}}& \dfrac{\varpi_{gu}^2}{\varpi_{uu}}& -\varpi_{g u}\\
								-\varpi_{i u}& -\varpi_{j u}& -\varpi_{k u}& -\varpi_{g u}& \varpi_{uu}
							\end{array}\right] 
						\end{aligned}
					\end{equation}
					where $\varpi_{iu} = -s_{ed} \phi_1\phi_5,~\varpi_{ju} = -s_{ed} \phi_2\phi_5,~\varpi_{ku} = -s_{ed} \phi_3\phi_5,~\varpi_{gu} = -s_{ed} \phi_4\phi_5$ and $\varpi_{uu} = s_{ed} \phi_5^2$. To eliminate $e_{jk}$, the equation is established as follows.
					\begin{equation}\label{eq:ED_s1}
						\begin{aligned}
							\Omega_{jk} + \dfrac{\varpi_{ju}\varpi_{ku}}{\varpi_{uu}} =0~~~
							\Rightarrow  s_{ed} =-\dfrac{\Omega_{jk}}{\phi_2 \phi_3}.
						\end{aligned}
					\end{equation}
					To maintain the universal rigidity and affine localizability of the obtained framework $\left(\mathcal{G}^-,\bm{p}\right)$, $s_{ed} >0$ 
					is required based on Corollary~\ref{coro:VA_three}, implying that
					\begin{equation}\label{eq:ED_3D}
						\Omega_{jk} \varpi_{ju}\varpi_{ku} <0 ~~\Rightarrow \Omega_{jk} \phi_2 \phi_3 <0.
					\end{equation}
				With a specific $\Omega_{jk}$, the sign of $s_{ed}$ is determined by the signs of $\varpi_{ju}$ and $\varpi_{ku}$, which depend on the relative positions among $v_i,~v_j,~v_k,~v_g$ and $v_u$. 
%
                Similar to Fig.~\ref{fig:FourVer}, four vertices in general position can form a tetrahedron, which divides the space into 15 parts in $\mathbb{R}^3$, as shown in Fig.~\ref{fig:SpaceDivision_3D}. When the fifth node is located in these 15 distinct regions, the sign of the stress varies.
                    \begin{figure}[htbp]
						\centering
                            \subfigure[]{\includegraphics[scale=0.27]{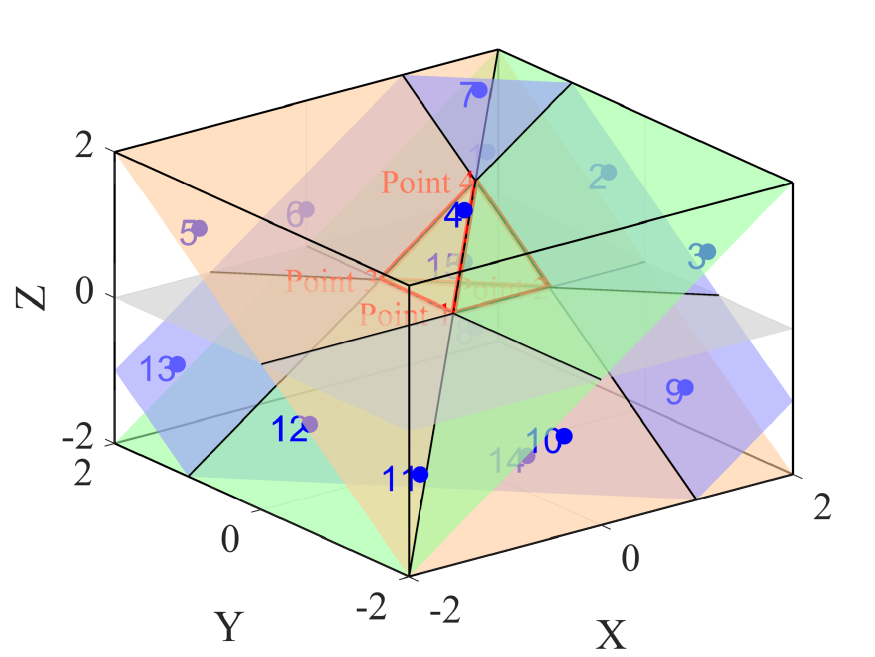} }
						\subfigure[]{\includegraphics[scale=0.27]{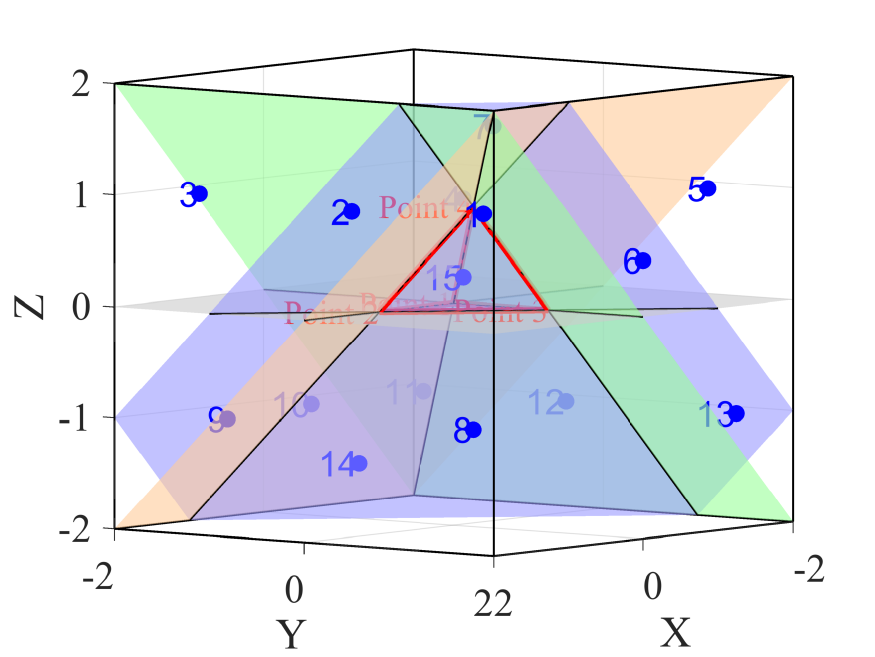} }
						\caption{Possible relative configurations of five vertices in general position in $\mathbb{R}^3$. The four faces of a tetrahedron divide $\mathbb{R}^3$ into 15 regions, with 15 blue dots scattered across these regions. (a) front view. (b) back view.}
						\label{fig:SpaceDivision_3D}
				\end{figure}
                    A similar sign analysis to that in $\mathbb{R}^2$ is presented as follows.
                    The vectors $\bm{p}_u-\bm{p}_i$ and $\bm{p}_u-\bm{p}_j$, $\bm{p}_u-\bm{p}_k$ can serve as a set of bases in $\mathbb{R}^3$. That is, 
					\begin{equation}
                    \footnotesize
						\bm{p}_u-\bm{p}_g = k_1 \left( \bm{p}_u-\bm{p}_i\right) + k_2 \left( \bm{p}_u-\bm{p}_j\right) + k_3 \left( \bm{p}_u-\bm{p}_k\right),
					\end{equation}
				where $k_1$, $k_2$ and $k_3$ are real parameters that are not simultaneously zero. Combining with eq.~\eqref{eq:ES_3D}, we have 
				\begin{equation}
                \begin{aligned}
                    \left( \varpi_{iu} + \varpi_{gu}k_1\right) &\left( \bm{p}_u-\bm{p}_i\right) + \left( \varpi_{ju} + \varpi_{gu}k_2\right) \left( \bm{p}_u-\bm{p}_j\right) \\
				+ &\left( \varpi_{ku} + \varpi_{gu}k_3\right) \left( \bm{p}_u-\bm{p}_k\right) = \bm{0}\notag.
                \end{aligned}
				\end{equation}
					Accordingly, the following equations are deduced.
					\begin{equation}
						\varpi_{iu} = -k_1\varpi_{gu}, \varpi_{ju} = -k_2\varpi_{gu}, \varpi_{ku} = -k_3\varpi_{gu},
					\end{equation}
					which mean $\varpi_{iu}\varpi_{ju} = k_1 k_2 \varpi_{gu}^2 $, $\varpi_{iu}\varpi_{ku} = k_1 k_3  \varpi_{gu}^2 $ and $\varpi_{ju}\varpi_{ku} = k_2 k_3  \varpi_{gu}^2 $. When $v_u$ lies in different regions, the signs of $k_1$, $k_2$, $k_3$ and the weights are listed in Table.~\ref{tab:Symbol}.

					\begin{table*}[htbp]
						\centering
						\caption{The sign of the weight of the edges connected to $v_u$ in $\mathbb{R}^3$.}
						\begin{tabular}{cccccccccccccccc}
							\toprule
							 & R1 & R2  & R3  & R4  & R5  & R6  & R7 & R8 & R9 & R10 & R11 & R12 & R13 & R14 & R15\\
							\midrule
							$k_1$  & $+$  & $+$  & $-$  & $-$  & $-$  & $+$  & $+$ & $-$  & $-$  & $+$  & $+$    & $+$  & $-$  & $+$  & $-$ \\
							$k_2$  & $-$  & $-$  & $-$  & $+$  & $+$  & $+$  & $+$ & $+$  & $+$  & $+$  & $-$    & $-$  & $-$  & $+$  & $-$ \\
							$k_3$  & $-$  & $+$  & $+$  & $+$  & $-$  & $-$  & $+$ & $+$  & $-$  & $-$  & $-$    & $+$  & $+$  & $+$  & $-$ \\
							$\varpi_{iu}\varpi_{ju}$  & $-$  & $-$  & $+$  & $-$  & $-$  & $+$  & $+$ & $-$  & $-$  & $+$  & $-$  & $-$  & $+$  & $+$  & $+$ \\
							$\varpi_{iu}\varpi_{ku}$  & $-$  & $+$  & $-$  & $-$  & $+$  & $-$  & $+$ & $-$  & $+$  & $-$  & $-$  & $+$  & $-$  & $+$  & $+$ \\
							$\varpi_{iu}\varpi_{gu}$  & $-$  & $-$  & $+$  & $+$  & $+$  & $-$  & $-$ & $+$  & $+$  & $-$  & $-$  & $-$  & $+$  & $-$  & $+$ \\
							$\varpi_{ju}\varpi_{ku}$  & $+$  & $-$  & $-$  & $+$  & $-$  & $-$  & $+$ & $+$  & $-$  & $-$  & $+$  & $-$  & $-$  & $+$  & $+$ \\
							$\varpi_{ju}\varpi_{gu}$  & $+$  & $+$  & $+$  & $-$  & $-$  & $-$  & $-$ & $-$  & $-$  & $-$  & $+$  & $+$  & $+$  & $-$  & $+$ \\
							$\varpi_{ku}\varpi_{gu}$  & $+$  & $-$  & $-$  & $-$  & $+$  & $+$  & $-$ & $-$  & $+$  & $+$  & $+$  & $-$  & $-$  & $-$  & $+$ \\
							\bottomrule
						\end{tabular}
						\label{tab:Symbol}
					\end{table*}
                    
                     Similar to the cases analyzed in Section~\Rmnum{3}-B, edge deletion strategies can be divided into two main types. The first achieves the elimination of edge $e_{jk}$ by rearranging stress within the framework to readjust $\Omega_{jk}$ to zero. The second involves introducing temporary nodes for remocing specific edges. However, unlike the two-dimensional case, the position of a temporary node can have up to 15 possibilities, and the number of temporary nodes can reach a maximum of three, increasing the complexity of the analysis. Example~\ref{exam:AffineConstruct_3D} is provided to illustrate the process of deleting an edge from a HAF in $\mathbb{R}^3$.

                \textbf{(3) Vertex Deletion in $\mathbb{R}^3$.} 
                Similar to Section~\Rmnum{3}-C, vertices in the HAF in $\mathbb{R}^3$ can also be divided into two types: inner nodes and outer nodes. On one hand, for an outer node $v_i$ with $h(v_i) > 0 $, the following corollary can be derived from Theorem~\ref{theo:VD}.
                \begin{coro}
                    Consider a HAF $\left(\mathcal{G},\bm{p}\right)$ with $n~\left(n \ge 6\right)$ vertices. If $v_u$ is an outer node and $h\left(v_u\right)>0$, we can obtain a universally rigid and affinely localizable framework $\left(\mathcal{G}_{vd},\bm{p}_{vd}\right)$ after removing vertex $v_u$ and all edges connected to $v_u$. 
                \end{coro}

              On the other hand, for an inner node $v_i$ with $h(v_i) > 0 $, by adopting the inheritance approach and following the four-step method proposed in Section~\Rmnum{3}-C, the deletion of the inner node can be achieved. An example is also shown in Example~\ref{exam:AffineConstruct_3D}.
                
            \begin{exam}\label{exam:AffineConstruct_3D}
                \emph{(Construction of HAF in $\mathbb{R}^3$):} An original affine framework $(\mathcal{G}_0,\bm{p}_0)$ is given in Fig.~\ref{fig:AffineConstruct_3D_Original}, where $\bm{p} = \left[ 0,0,0,8,0,0,0,8,0,0,0,8,4,8,8\right]^T$. The stresses of the edges in $\mathcal{G}_0$ are presented in Fig.~\ref{fig:AffineConstruct_3D_Original}. The operations of vertex addition, edge deletion, and vertex deletion for $(\mathcal{G}_0,\bm{p}_0)$ are shown as follows. 

                \begin{figure}[htbp]
				\centering
				\subfigure[]{
					\includegraphics[scale=0.26]{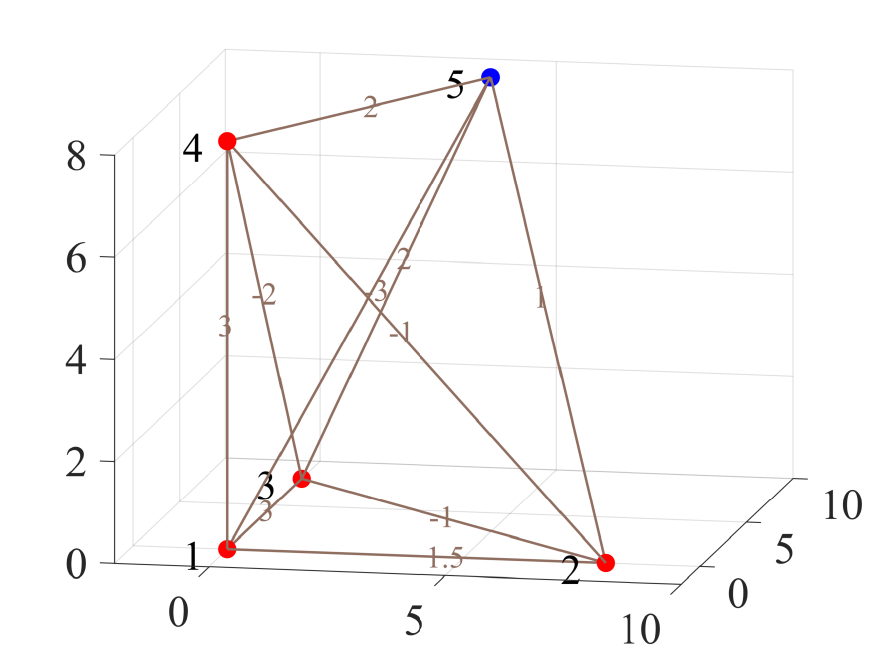}
                    \label{fig:AffineConstruct_3D_Original}
				}
				\subfigure[]{
					\includegraphics[scale=0.26]{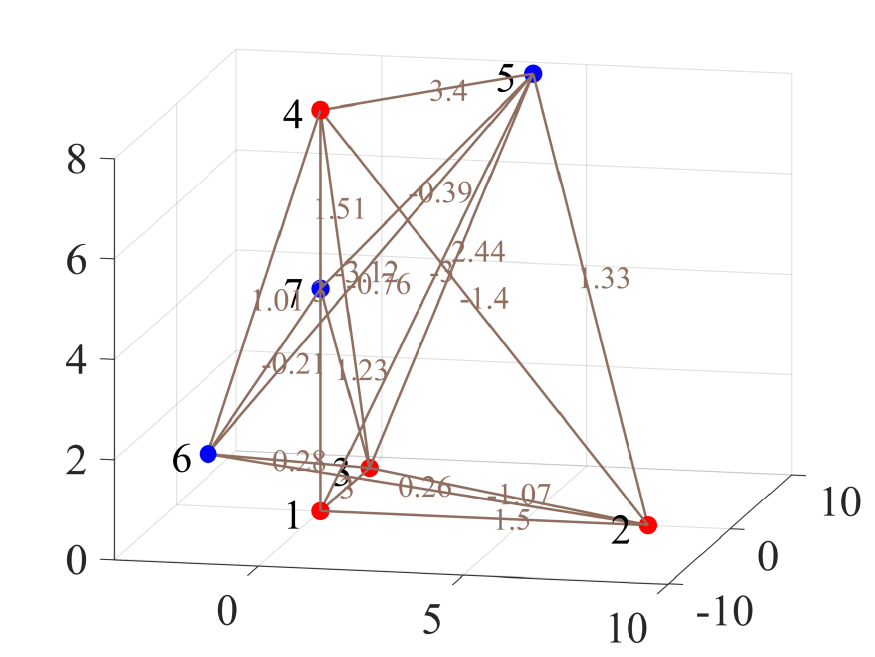}
                    \label{fig:AffineConstruct_3D_VA}
				}
				\subfigure[]{
					\includegraphics[scale=0.26]{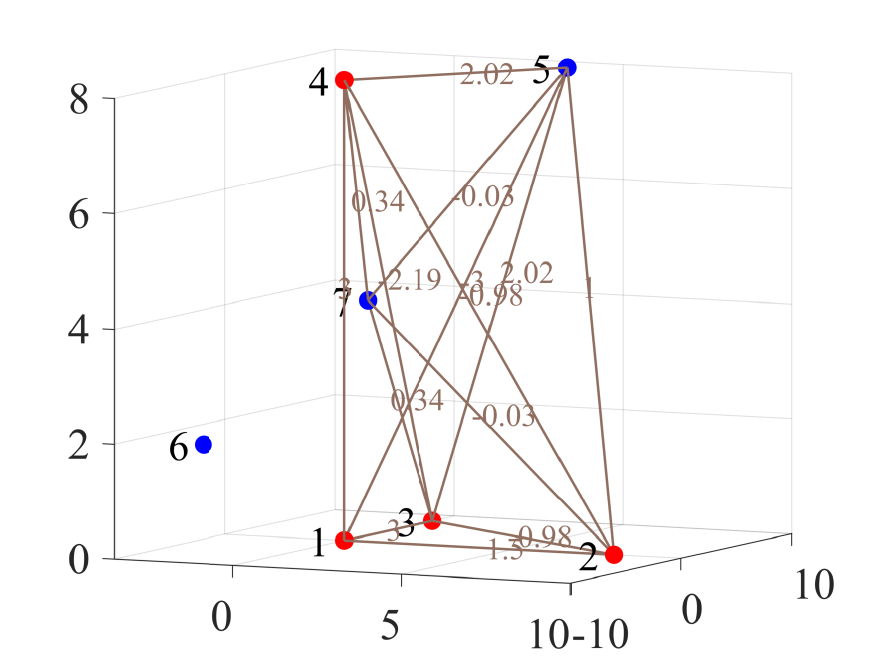}
                    \label{fig:AffineConstruct_3D_VD}
				}
				\subfigure[]{
					\includegraphics[scale=0.26]{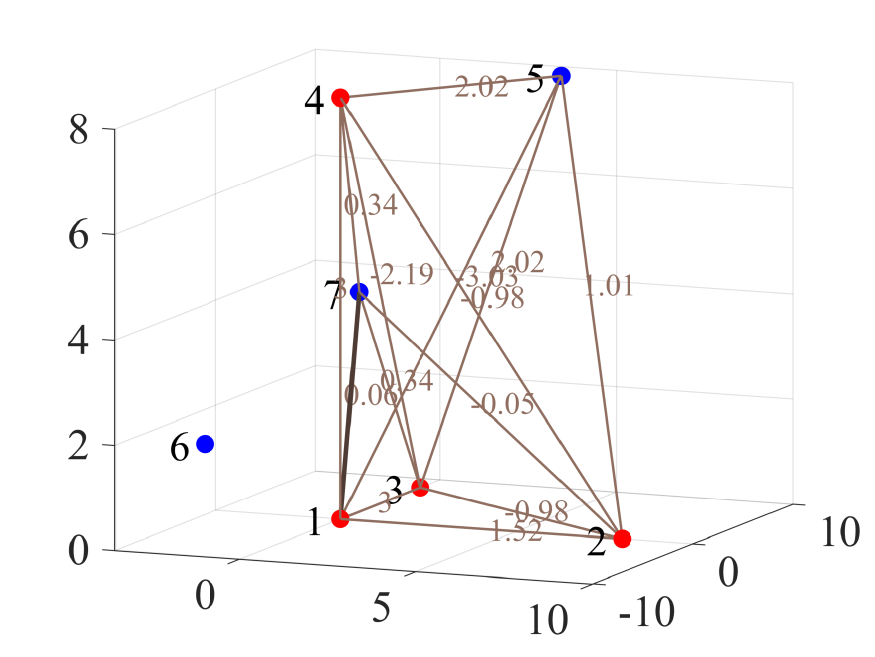}
                    \label{fig:AffineConstruct_3D_ED}
				}
				\caption{An example of edge deletion, vertex addition and deletion in $\mathbb{R}^3$. (a) The original framework. (b) The obtained HAF after adding vertices $v_6$ and $v_7$. (c) The framework after deleting vertex $v_6$. (d) The framework after removing edge $e_{57}$, while adding $e_{17}$.}
				\label{fig:AffineConstruct_3D}
			\end{figure}

            \textbf{(1) Vertex Addition.} 
            In three-dimensional space, randomly select two vertices $ v_6$ and $ v_7$, where $\bm{p}_6 = \left[ -1.542, -8.115  , 1.971 \right]^T$ and $\bm{p}_7 = \left[-0.582,3.919,3.998\right]^T$. Select four parents for $ v_6 $: $ v_2 ,v_3 , v_4 , v_5$, and then connect $ v_6 $ to its parents so that $h(v_6)=1$. Then, update the stress matrix according to eq.~\eqref{eq:VA_three}. Next, add $ v_7 $ to the new framework and update the stress matrix in the same way. The parents of $ v_7 $ are $ v_3 ,v_4,v_5,v_6$ so that $h(v_7)=2$. After adding the two vertices, the stresses of the edges in the obtained HAF are shown in Fig.~\ref{fig:AffineConstruct_3D_VA}.

            \textbf{(2) Vertex Deletion.} 
            According to the definition, it is obvious that $ v_6 $ is an inner node with four parents and one child. To delete $ v_6 $, following the four steps in Section~\Rmnum{3}-C: first, disconnect all edges linked to $ v_6 $. Then, select new parents for the child $ v_7 $, which are $ v_2, v_3, v_4, v_5 $. Update the stress matrix using eq.\eqref{eq:SM_VD_inhe}. The new stresses of the edges are shown in Fig.~\ref{fig:AffineConstruct_3D_VD}.

            \textbf{(3) Edge Deletion.} 
            Taking the deletion of edge $ e_{57} $ in Fig.~\ref{fig:AffineConstruct_3D_VD} as an example, the corresponding element in the stress matrix is $\Omega_{57} = 0.0344$. To ensure $s_{ed} > 0 $, select vertices $v_1,v_2,v_3$ as neighbors. Then, we have
            \begin{equation}
            \footnotesize
            \notag
            \begin{aligned}
                \bm{\phi}_u^{ed} = &\operatorname{null}(\left[\begin{array}{ccccc}
                    \bm{p}_1 & \bm{p}_2 &\bm{p}_3 &\bm{p}_5 &\bm{p}_7\\
                     1 & 1 &1 &1 &1
                \end{array}
                \right]) \\
                =&\left[\begin{array}{ccccc}
                    -0.582 & 0.225 & 0.007 & -0.349 & 0.699
                \end{array}
                \right]^T
            \end{aligned}
            \end{equation}
            Thus, we have $\Omega_{57} \phi_4 \phi_5 = 0.0344 \times (-0.349) \times 0.699<0$, meeting eq.~\eqref{eq:ED_3D}. Therefore, let $s_{ed} = 0.141>0$ and update the stress matrix based on eq.~\eqref{eq:omega_3D}. The obtained affine framework is shown in Fig.~\ref{fig:AffineConstruct_3D_ED}, where edge $e_{57}$ is removed.   
                \end{exam}
				
\section{Simulations}\label{sec:Simulation}
	In this section, two simulations and a comparison are carried out to validate our proposed algorithms. 
				\subsection{Simulation 1: Growing Framework for Affine Formations}
				By using Algorithm~\ref{Algo:VertexAdd}, we present simulations to grow an eligible HAFs with different number of agents, as shown in Fig.~\ref{fig:SimuCompare}. An original framework with four vertices is given in Fig.~\ref{fig:SimuCompare_4}, then a series of vertices with random positions are considered to be added to the framework. As described in Fig.~\ref{fig:SimuCompare_54}-\ref{fig:SimuCompare_204}, frameworks with different swarm size are generated while maintaining universal rigidity and affine localizability, which implying the effectiveness of the proposed distributed vertex addition strategy. We further demonstrate the rigidity of constructing affine frameworks by comparative simulations with \cite{RAL_Xiao_2022}, where a MISDP problem is established and an optimization-based topology design scheme is designed. 
				We use two methods to grow affine frameworks with the same number of agents. The programs run on a laptop with AMD Ryzen 7 and 16 GB memory. Each simulation is repeated ten times and the average runtime is recorded in Table.~\ref{tab:SimuCompare}. It is obvious that Xiao's method takes much more runtime than ours, due to the fact that the number of variables that need to be optimized in the MISDP algorithm increases rapidly with the number of vertices, resulting in a sharp increase in runtime. As a comparison, our method only focuses on several neighbors near the newly added vertex, bringing about a significant reduction in computational cost, which reveals the potential of our method in onboard applications in robot swarms.
				
				\begin{figure}[htbp]
					\centering
					\subfigure[]{
						\includegraphics[scale=0.25]{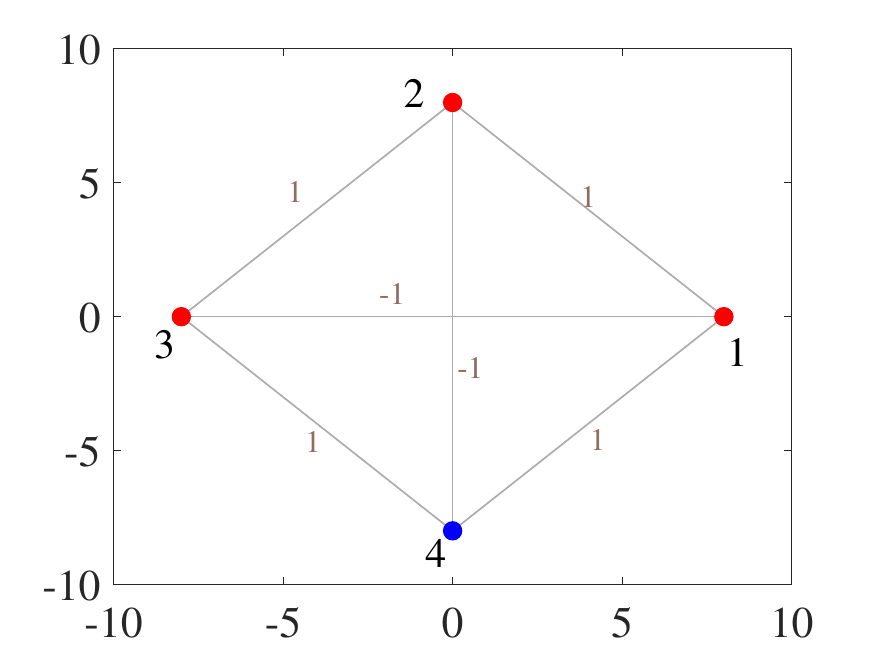}
						\label{fig:SimuCompare_4}
					}
					\subfigure[]{
						\includegraphics[scale=0.25]{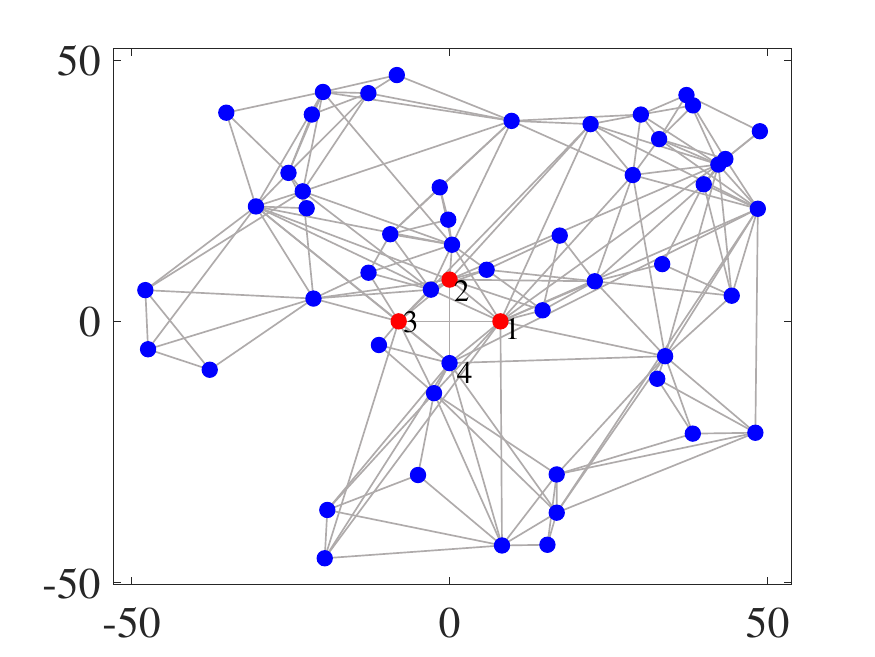}
						\label{fig:SimuCompare_54}
					}
					\subfigure[]{
						\includegraphics[scale=0.25]{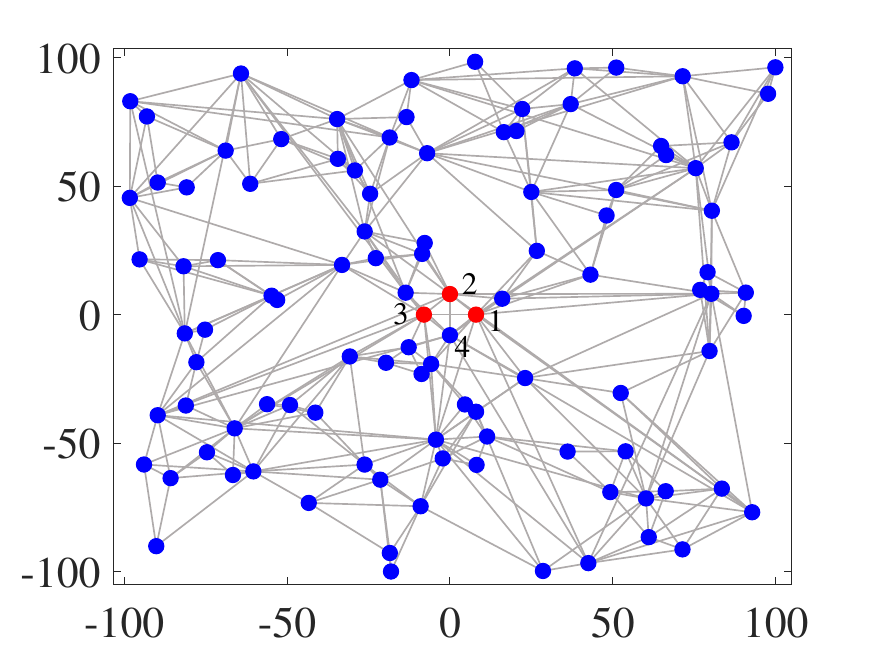}
						\label{fig:SimuCompare_104}
					}
					\subfigure[]{
						\includegraphics[scale=0.25]{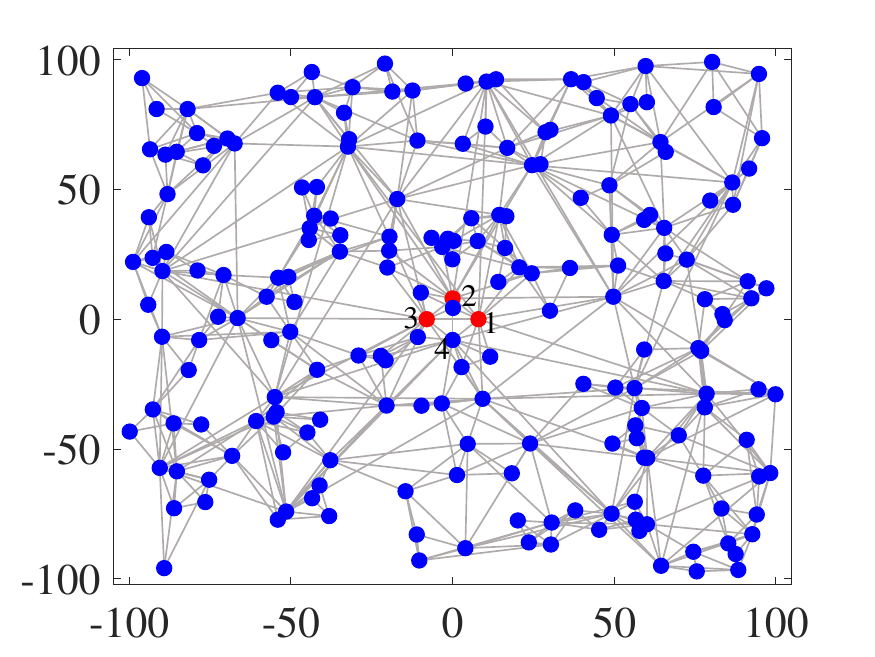}
						\label{fig:SimuCompare_204}
					}
					\caption{An example to grow affine frameworks using Algorithm~\ref{Algo:VertexAdd}. (a) The original framework with 4 agents. (b) The obtained framework after adding 50 agents. (c) The framework after adding after adding 100 agents. (d) The framework after adding after adding 200 agents.}
					\label{fig:SimuCompare}
				\end{figure}
				
				\begin{table}[h]
					\centering
					\caption{Average Runtime with network size when different framework growing strategy is imposed.}
					\begin{tabular}{ccc}
						\toprule 
						Number of Agents & 	Xiao's \cite{RAL_Xiao_2022}$ 
                        \left(s\right) $ & \textbf{Ours} $ \left(s\right) $\\
						\midrule
						5 Agents & 0.8886 & 0.0010822 \\
						6 Agents & 27.9123 & 0.0013588 \\
						7 Agents & 755.8002 & 0.001694 \\
						8 Agents & NA & 0.00217 \\
						20 Agents & NA & 0.01111 \\
						50 Agents & NA & 0.06963 \\
						100 Agents & NA & 0.27694 \\
						200 Agents & NA & 1.30814 \\
						\bottomrule
					\end{tabular}
					\label{tab:SimuCompare}
				\end{table}
				
				\subsection{Simulation 2: Edge Deletion \& Vertex Addition and Deletion}\label{sec:Simu_Ctrl}
				In this subsection, a scenario is designed to validate the effectiveness of our proposed affine framework construction algorithms, namely the proposed vertex addition, edge deletion and vertex deletion strategies in section \ref{sec:MainResult}. We couple the affine framework construction strategies with the formation tracking control law proposed in \cite{FITEE_Li_2022}, with the former providing the nominal topology for the latter. In the scenario, we consider a group of fixed-wing unmanned aerial vehicles (UAVs) modeled by unicycles moving in two-dimensional space while tracking moving leaders and achieving affine transformations.  
				The convergence of affine formation tracking errors can serve as a powerful evidence to prove the effectiveness of our framework construction algorithms. Moreover, the adopted control scheme \cite{FITEE_Li_2022} is just an example and can be replaced by other similar affine formation control laws.
				

			We consider a group of six agents with an original framework shown in Fig.~\ref{fig:SimuScen1_Topo1}, where the first three red vertices are regarded as leaders. Followers are driven to track leaders while achieving affine transformations. During the maneuver, agents adjust the topological connection between each other, as shown in Fig.~\ref{fig:SimuScen1_Topo2}. Then, two additional agents join the formation (Fig.~\ref{fig:SimuScen1_Topo3})  following Algorithm~\ref{Algo:VertexAdd}. After a period of formation flight, three agents leave the formation to perform other tasks, and the affine framework is reshaped as shown in Fig.~\ref{fig:SimuScen1_Topo4}. Fig.~\ref{fig:SimuScen1_Traje} depicts the trajectory of UAVs and Fig.~\ref{fig:SimuScen1_TrackError} shows the evolution of formation tracking errors $\| \bm{p}_i - \bm{p}_i^*\|$, where $ \bm{p}_i$ represents the position of the $i$-th UAV and $\bm{p}_i^*$ represents the desired one. It can be seen that the target formations can be achieved while time-varying affine transformations are performed to pass through various environments,  implying that the obtained framework is universally rigid and affine localizable.
			\begin{figure}
				\centering
				\subfigure[]{
					\includegraphics[scale=0.22]{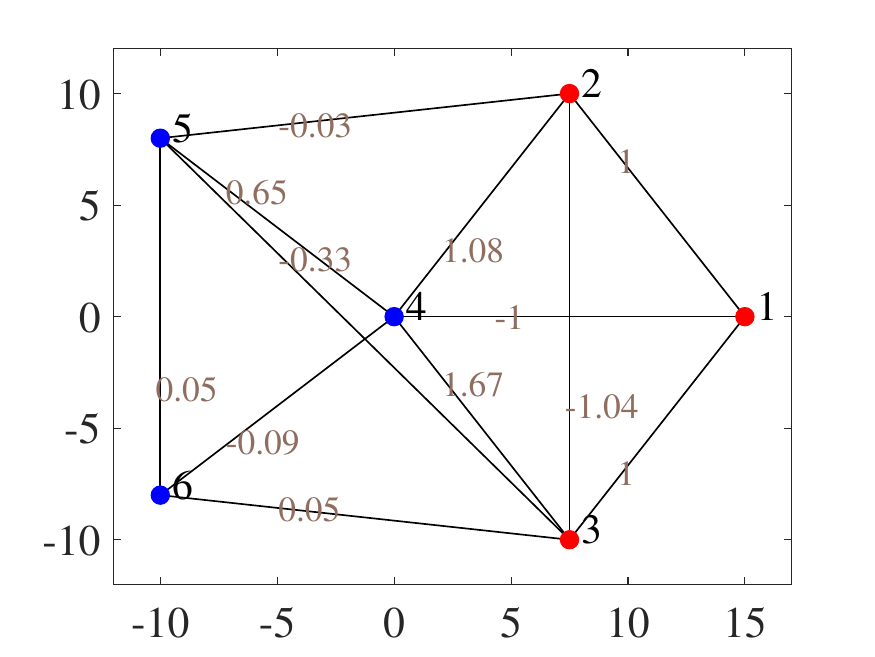}
					\label{fig:SimuScen1_Topo1}
				}
				\subfigure[]{
					\includegraphics[scale=0.22]{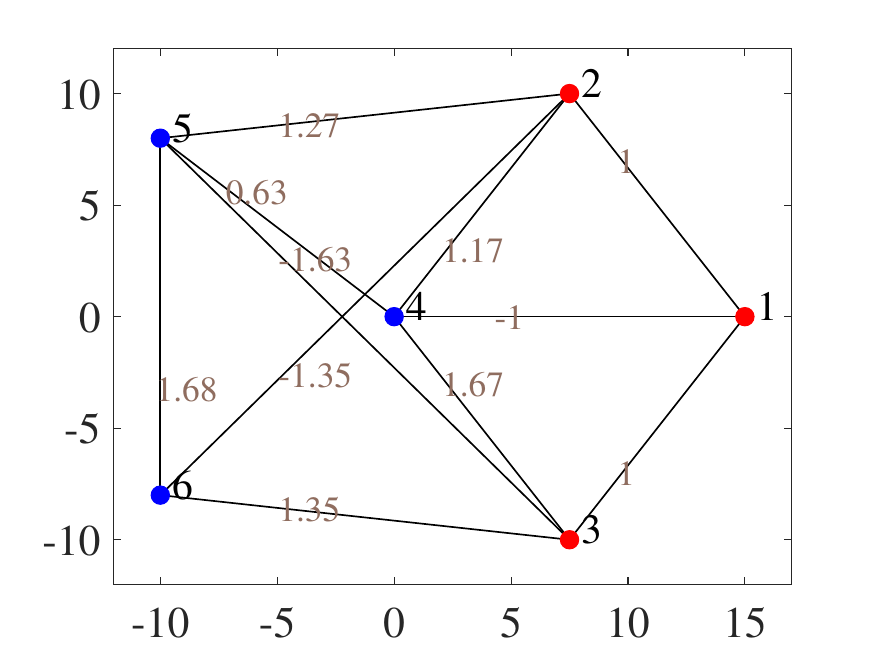}
					\label{fig:SimuScen1_Topo2}
				}
				\subfigure[]{
					\includegraphics[scale=0.22]{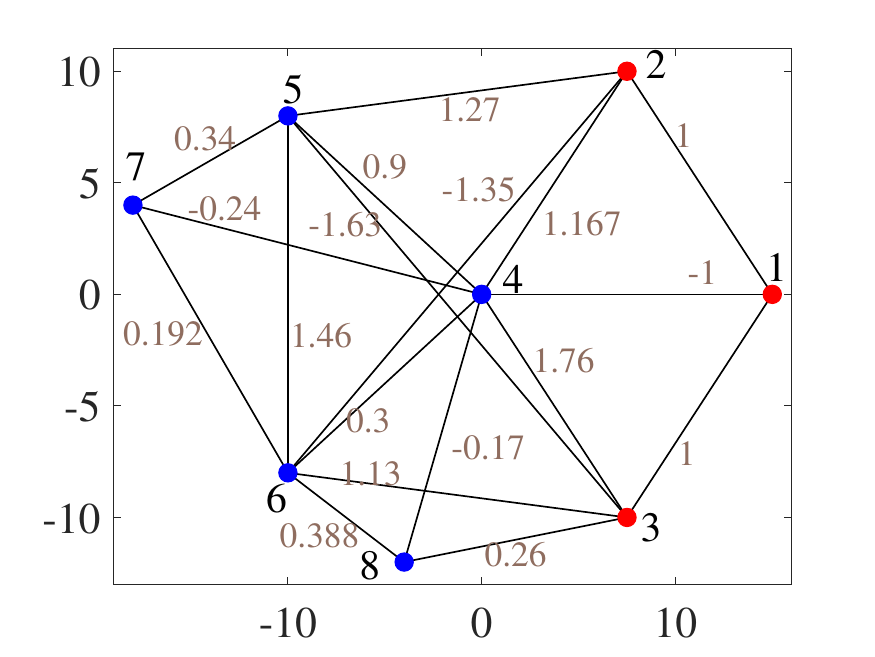}
					\label{fig:SimuScen1_Topo3}
				}
				\subfigure[]{
					\includegraphics[scale=0.22]{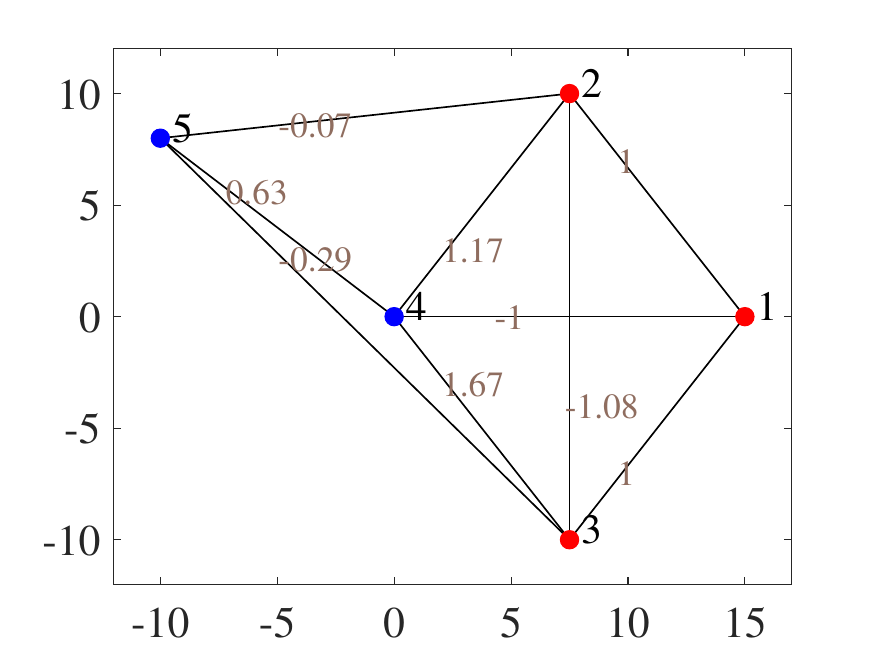}
					\label{fig:SimuScen1_Topo4}
				}
				\caption{An example of edge deletion, vertex addition and deletion. (a) The original framework. (b) The obtained framework after deleting $e_{23}$ and $e_{46}$. (c) The framework after adding two vertices, denoted by $v_7$ and $v_8$. (d) The framework after three vertices deletion.}
				\label{fig:SimuScen1_Topo}
			\end{figure}
			
			\begin{figure}
				\centering
				\includegraphics[scale=0.25]{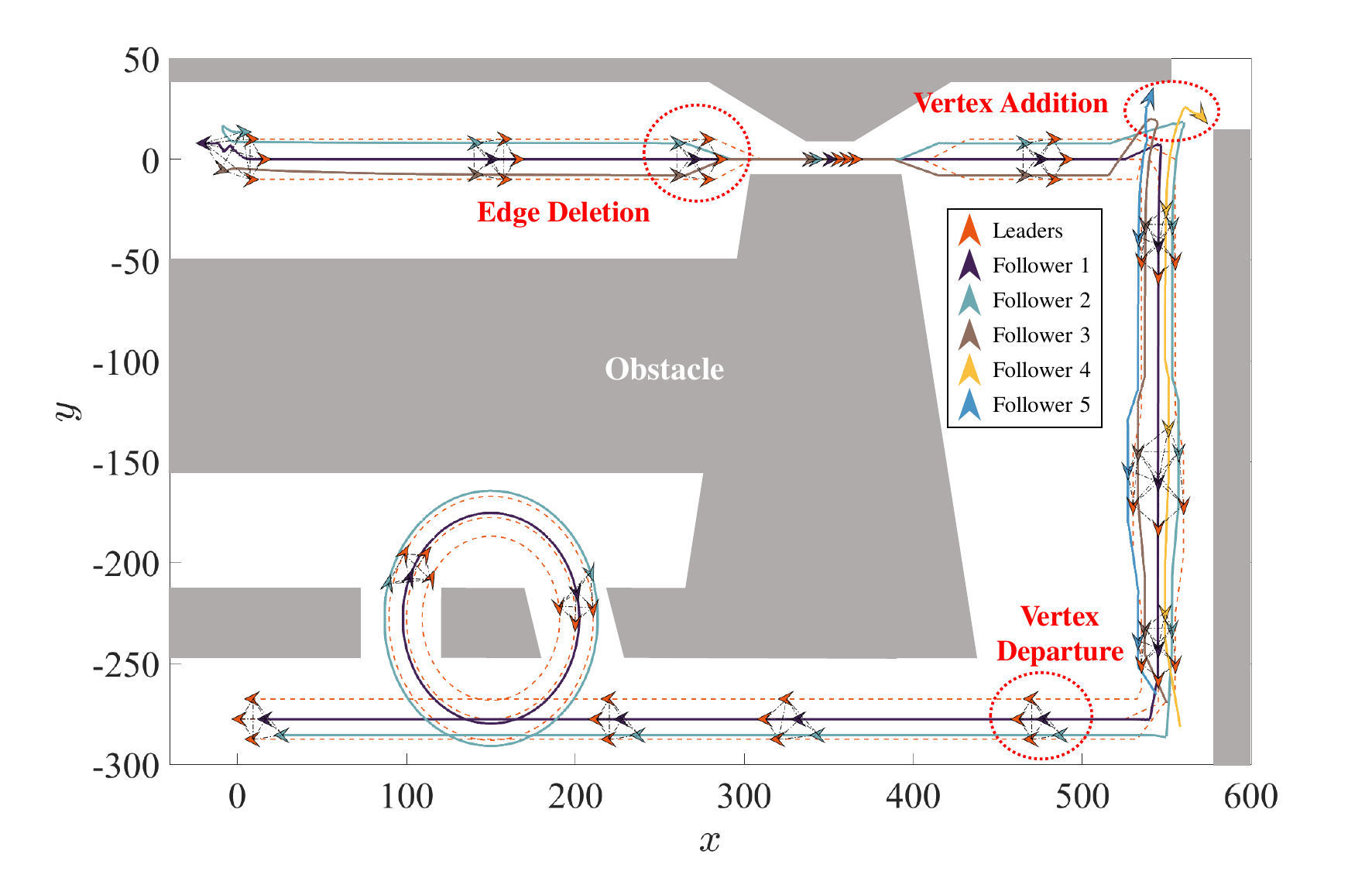}
				\caption{The trajectory of fixed-wing UAVs  under unexpected events.}
				\label{fig:SimuScen1_Traje}
			\end{figure}
			
			\begin{figure}
				\centering
				\includegraphics[width=0.3\textwidth]{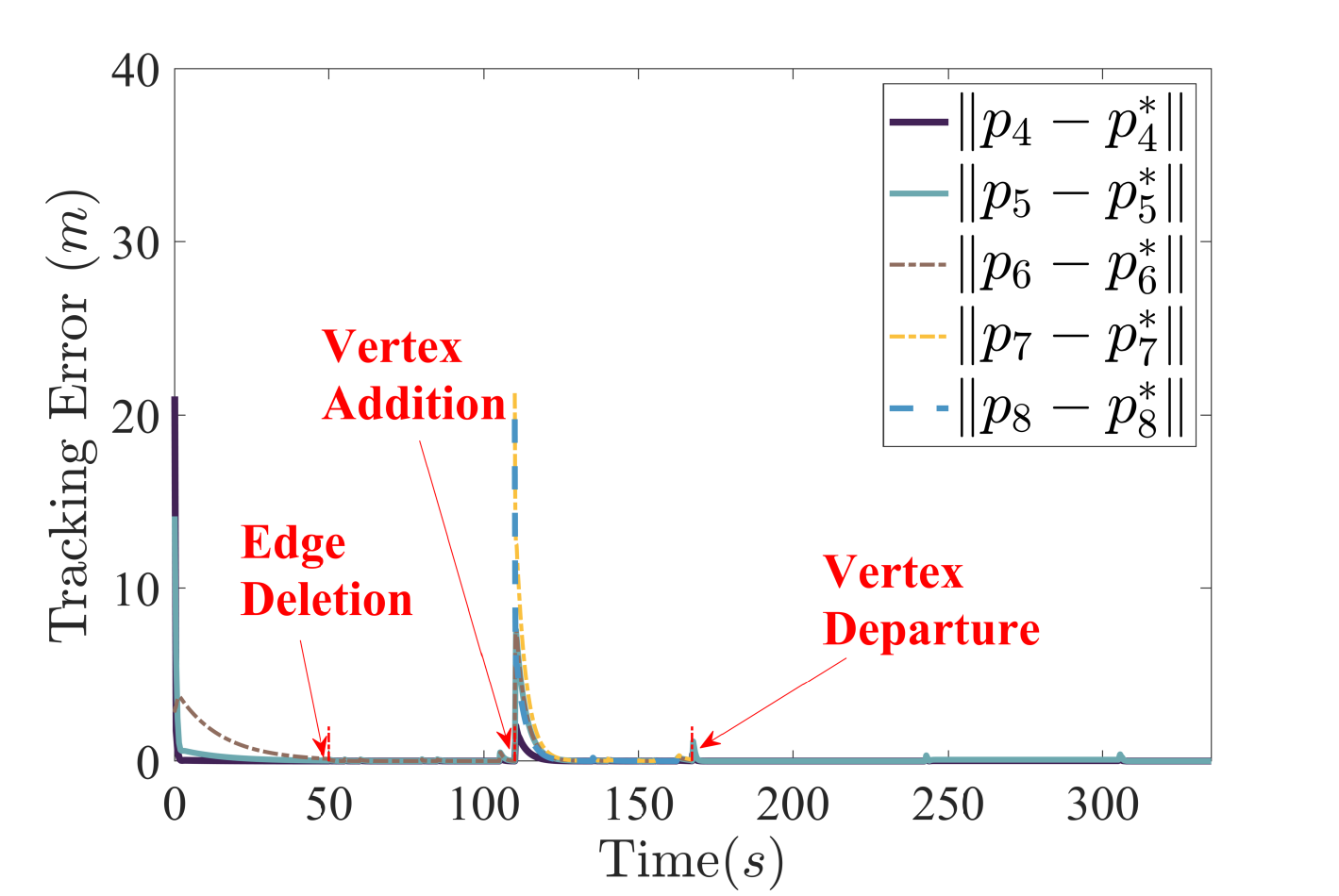}
				\caption{The evolution of formation tracking errors for fixed-wing UAVs in Scenario 1.}
				\label{fig:SimuScen1_TrackError}
			\end{figure}
			
			Based on the simulation results, it is evident that our proposed distributed framework construction algorithms are highly effective to generate and reconstruct affine frameworks, demonstrating the potential for extensive applications in robot swarms.
			
			\section{Conclusion}\label{sec:Conclusion}
			This paper have addressed the issue of constructing affine frameworks in a distributed manner in three scenarios, i.e., vertex addition, edge deletion and vertex deletion. We have designed strategies to construct frameworks with universal rigidity and affine localizability from the geometric perspective based on the structure of weighted graphs used to describe topology of affine formations. Naturally, comprehensive theoretical analysis has been provided to demonstrate the effectiveness, and simulations have been presented to verify the compatibility with distributed affine formation control proposals. Moreover, the comparative simulation has demonstrated the rapidity in constructing large-scale affine formations. 
			Our approach relaxes the requirements for global information and computing resources, making it more easy to implement for robot swarms to perform various tasks. One future research direction is to extend distributed framework construction algorithms to directed graphs. Furthermore, how to integrate the dynamic control scheme design of multi-agent systems with the proposed time-varying framework construction strategy in this paper for comprehensive theoretical analysis to obtain better performance of affine formation control, such as faster convergence speed and lower energy consumption, is also a promising research problem.
			
			\appendices
			\section{The proof of Theorem~\ref{theo:VertexAddition}}
			\label{sec:app}
			
			To complete the proof of Theorem~\ref{theo:VertexAddition}, the following lemmas are introduced for the analysis presented in the sequel.

			\begin{lemm}\label{lemma:ker}\cite{TCYB_Yang_2019}
				Given PSD matrices $\bm{X} \in \mathbb{R}^{n \times n}$ and $\bm{Y} \in \mathbb{R}^{n \times n}$, let $\bm{Z}=\bm{X}+\bm{Y}$. Then for any nonzero vector $\bm{\varepsilon} \in \mathbb{R}^{n }$, $\bm{\varepsilon} \in \operatorname{null}(\bm{Z})$ if and only if $\bm{\varepsilon} \in \operatorname{null}(\bm{X})$ and $\varepsilon \in \operatorname{null}(\bm{Y})$.
			\end{lemm}
			
			
			\begin{proof}
				In the newly obtained framework $\left(\mathcal{G}^+,\bm{p}^+\right)$, we have $\mathcal{V}_l^+ = \mathcal{V}_l$, $\mathcal{V}_f^+ := \mathcal{V}_f \cup \left\{v_u\right\}$ and $\bm{\Omega}^+ :=\left[\begin{array}{cc}
					\bm{\Omega}^+_{l l} & \bm{\Omega}^+_{l f} \\
					\bm{\Omega}^+_{f l} & \bm{\Omega}^+_{f f}
				\end{array}\right]$.
				Combining Lemmas~\ref{lemma:StressUniRigid} and \ref{lemma:AffineLocalizable}, to demonstrate $\left(\mathcal{G}^+,\bm{p}^+\right)$ is universally rigid and affinely localizable, it is equivalent to prove the following conditions: (1) $\bm{\Omega}^+$ is PSD; (2)  $\operatorname{rank}\left(\bm{\Omega}^+\right) = n-2$; (3) the block matrix $\bm{\Omega}_{ff}^+$ is positive definite. These are proven one by one as follows.
				
				\textbf{The proof of (1):} Based on Assumption~\ref{assu:OriginalFrame} and eq.~\eqref{eq:Omega_u1}, we have $\bm{\Omega}$ and $\bm{\Omega}_u$ are both PSD. Accordingly, it is obvious that the matrix $\bm{\Omega}^+$ in \eqref{eq:StressMatrixPlus} is PSD, because $\bm{\Omega}_a$ and $\bm{\Omega}_b$ are both PSD. 
					
				\textbf{The proof of (2):} Then, we analyze the rank of $\bm{\Omega}^+$. According to Lemma~\ref{lemm:RN}, Let us start with the kernel of $\bm{\Omega}^+$.
					
				Based on Lemma~\ref{lemma:NullStress}, we obtain the null space of $\bm{\Omega}_a$ in eq.~\eqref{eq:StressMatrixPlus} is
				\begin{equation}
					\footnotesize
					\begin{aligned}
						\mathbb{S}_a =& \operatorname{span}\left(\left[\begin{array}{c}
							\bm{p}_{1, 1}\\
							\vdots \\
							\bm{p}_{i, 1}\\
							\bm{p}_{j, 1}\\
							\bm{p}_{k, 1}\\
							\bm{q}_{u, 1}
						\end{array}\right],~\left[\begin{array}{c}
							\bm{p}_{1, 2}\\
							\vdots \\
							\bm{p}_{i, 2}\\
							\bm{p}_{j, 2}\\
							\bm{p}_{k, 2}\\
							\bm{q}_{u, 2}
						\end{array}\right],~\left[\begin{array}{c}
							1\\
							\vdots \\
							1\\
							1\\
							1\\
							c_u
						\end{array}\right]\right) \\
						\triangleq & \operatorname{span}\left(\bm{n}_1^a,~\bm{n}_2^a,~\bm{n}_3^a\right),
					\end{aligned}
				\end{equation}
			where $\bm{p}_{i, 1}$ is the first component of $\bm{p}_i \left( i \in \left\{1,\cdots,n\right\}\right)$, and $\bm{p}_{i, 2}$ is defined analogously. $\bm{q}_{u,\sim }$ and $c_u$ are any arbitrary scalars. Similarly, the null space of $\bm{\Omega}_b$ is
					\begin{equation}
						\footnotesize
						\begin{aligned}
							\mathbb{S}_b = & \operatorname{span}\left(\left[\begin{array}{c}
								\bm{q}_{1, 1}\\
								\vdots \\
								\bm{q}_{\left(i-1\right)\cdot 1}\\
								\bm{p}_{i, 1}\\
								\bm{p}_{j, 1}\\
								\bm{p}_{k, 1}\\
								\bm{p}_{u, 1}
							\end{array}\right],~\left[\begin{array}{c}
								\bm{q}_{1, 2}\\
								\vdots \\
								\bm{q}_{\left(i-1\right)\cdot 2}\\
								\bm{p}_{i, 2}\\
								\bm{p}_{j, 2}\\
								\bm{p}_{k, 2}\\
								\bm{p}_{u, 2}
							\end{array}\right],~\left[\begin{array}{c}
								c_1\\
								\vdots \\
								c_i\\
								1\\
								1\\
								1\\
								1
							\end{array}\right]\right) \\
							\triangleq & \operatorname{span}\left(\bm{n}_1^b,~\bm{n}_2^b,~\bm{n}_3^b\right).
						\end{aligned}
					\end{equation}
					
			According to Lemma~\ref{lemma:ker} and eq.~\eqref{eq:StressMatrixPlus}, the null space of $\bm{\Omega}^+$ can be presented as $\operatorname{null}\left(\bm{\Omega}^+\right) = \mathbb{S}_a \cap \mathbb{S}_b$, whose nominal form is shown as
					\begin{equation}
						\alpha_1 \bm{n}_1^a + \alpha_2 \bm{n}_2^a + \alpha_3 \bm{n}_3^a = \beta_1 \bm{n}_1^b + \beta_2 \bm{n}_2^b + \beta_3 \bm{n}_3^b,
					\end{equation}
			where $\alpha_i$ and $\beta_i$ $\left(i =1,2,3\right)$ are scalars that are not all zero. Because $\mathbb{S}_a$ and $\mathbb{S}_b$ share three  identical elements, the following equation is deduced,
					\begin{equation}
						\footnotesize
						\left(\alpha_1-\beta_1\right)\left[\begin{array}{c}
							\bm{p}_{i, 1}\\
							\bm{p}_{j, 1}\\
							\bm{p}_{k, 1}
						\end{array}\right] + \left(\alpha_2-\beta_2\right)\left[\begin{array}{c}
							\bm{p}_{i, 2}\\
							\bm{p}_{j, 2}\\
							\bm{p}_{k, 2}
						\end{array}\right]  + \left(\alpha_3-\beta_3\right)\left[\begin{array}{c}
							1\\
							1\\
							1
						\end{array}\right] = \bm{0},
					\end{equation}
			which means
					\begin{equation}
						\footnotesize
						\underbrace{\left[\begin{array}{ccc}
								\bm{p}_i & \bm{p}_j & \bm{p}_k\\
								1 & 1 & 1
							\end{array}\right]}_{\triangleq \bm{P}_k} \left[\begin{array}{c}
							\alpha_1-\beta_1\\
							\alpha_2-\beta_2\\
							\alpha_3-\beta_3
						\end{array}\right] = \bm{0}.
					\end{equation}
			Since the configuration $\bm{p} $ is in general position, vertices $v_i$, $v_j$ and $v_k$ are not affinely dependent so that $\operatorname{rank}\left(\bm{P}_k\right) = 3$ holds. Accordingly, we have $\alpha_i = \beta_i, ~\forall i \in \left\{1,~2,~3\right\}$. Assume there is a vector $\bm{v} \in \operatorname{null}\left(\bm{\Omega}^+\right)$, then we will show that $\bm{v}$ can be represented by $\bm{p}^+_{i, 1}$, $\bm{p}^+_{i, 2}~\left(i \in \left\{1, \cdots,n+1\right\}\right)$ and $\bm{1}_{n+1}$. According to Lemma~\ref{lemma:ker}, $\bm{v} \in \operatorname{null}\left(\bm{\Omega}^+\right)$ implies $\bm{v} \in \mathbb{S}_a$ and $\bm{v} \in \mathbb{S}_b$. That is, 
					\begin{equation}
							\bm{v} = \alpha_1 \bm{n}_1^a + \alpha_2 \bm{n}_2^a + \alpha_3 \bm{n}_3^a =\alpha_1 \bm{n}_1^b + \alpha_2 \bm{n}_2^b + \alpha_3 \bm{n}_3^b. \notag
					\end{equation}
			Thus, let the last entry in $\bm{n}_i^a$ equal the last entry in $\bm{n}_i^b$, and the first $n$ entries in $\bm{n}_i^b$ equal the first $n$ entries in $\bm{n}_i^a$. It follows that 
					\begin{equation}
						\small
						\bm{v} = \alpha_1 \left[\begin{array}{c}
							\bm{p}_{1, 1}\\
							\vdots \\
							\bm{p}_{i, 1}\\
							\bm{p}_{j, 1}\\
							\bm{p}_{k, 1}\\
							\bm{p}_{u, 1}
						\end{array}\right] + \alpha_2 ~\left[\begin{array}{c}
							\bm{p}_{1, 2}\\
							\vdots \\
							\bm{p}_{i, 2}\\
							\bm{p}_{j, 2}\\
							\bm{p}_{k, 2}\\
							\bm{p}_{u, 2}
						\end{array}\right] + \alpha_3\left[\begin{array}{c}
							1\\
							\vdots \\
							1\\
							1\\
							1\\
							1
						\end{array}\right],
					\end{equation}
			which means $${\footnotesize \operatorname{null}\left(\bm{\Omega}^+\right) = \operatorname{col}(\underbrace{\left[\begin{array}{cccccc}
							\bm{p}_1 & \cdots & \bm{p}_i & \bm{p}_j & \bm{p}_k & \bm{p}_u \\
							1 & \cdots & 1 & 1 & 1 & 1
						\end{array}\right]^T}_{\triangleq \bar{P}\left(\bm{p}^+\right)})}.$$
			Since $\operatorname{rank}\left(\bar{P}\left(\bm{p}^+\right)\right) = 3$, we have $\operatorname{nullity}\left(\bm{\Omega}^+\right) = 3$ so that $\operatorname{rank}\left(\bm{\Omega}^+\right) = n+1-3= n-2$. According to Lemma~\ref{lemma:StressUniRigid}, the framework $\left(\mathcal{G}^+, \bm{p}^+ \right)$ is universally rigid.
					
			\textbf{The proof of (3):} Now consider the block matrix $\bm{\Omega}_{ff}^+$. The added vertex $v_u$ plays the role of a follower in $\left(\mathcal{G}^+,\bm{p}^+\right)$. 
			We first prove that the diagonal elements of $\bm{\Omega}_{ff}^+$, denoted as $\varpi_{uu} = s \phi_4^2 >0 $. With a positive $s$, it is equivalent to prove $\phi_4 \ne 0$. 
			Suppose $\phi_4 = 0$, we get the following equations based on eq.~\eqref{eq:phi},
			\begin{equation}
				\begin{aligned}
					\phi_1 +\phi_2+\phi_3 =0 ,
					~ \phi_1 \left(\bm{p}_i -\bm{ p}_k \right) + \phi_2 \left(\bm{p}_j - \bm{p}_k \right)=\bm{0} .
				\end{aligned}
			\end{equation}
			Since $\left(\mathcal{G},\bm{p}\right)$ is in general position, three vertices $v_i,v_j,v_k$ are non-collinear. The vectors $\left(\bm{p}_i -\bm{ p}_k \right)$ and $\left(\bm{p}_j - \bm{p}_k \right)$ are linearly independent, implying $\phi_1 = \phi_2 = 0$. Then, $\phi_3 = 0$ holds. Consequently, $\bm{\phi} = \bm{0}$ is obtained, which contradicts with the fact that $\bm{\phi}$ is a nonzero vector. Therefore, we have $\phi_4 \ne 0$ so that $\varpi_{uu}$ is positive when $s>0$.

			Depending on whether the selected vertices $v_i$, $v_j$ and $v_k$ represent leaders or not, there are four cases in total as below. 
			\begin{itemize}
				\item[\labelitemi] \textbf{All the vertices $v_i,v_j,v_k\in\mathcal{V}_l$.}
					
				Based on Assumption~\ref{assu:OriginalFrame}, we get the block stress matrix $\bm{\Omega}_{ff}$ associated with $\bm{\Omega}$ is positive definite.
				If $v_i,v_j,v_k\in\mathcal{V}_l$, then $\bm{\Omega}_{ff}^+$ can be described by 
					$\bm{\Omega}_{ff}^+ = \left[\begin{array}{c:c}
							\bm{\Omega}_{ff} & \bm{0}_{n_f \times 1}\\
							\hdashline
							\bm{0}_{1 \times n_f} & \varpi_{uu}
						\end{array}\right]$. 
				Since $\varpi_{uu}>0$, $\bm{\Omega}_{ff}^+  \succ 0$ holds.
                
				\item[\labelitemi] \textbf{Two of the vertices $v_i,v_j,v_k$ belong to $\mathcal{V}_l$, and one belongs to $\mathcal{V}_f$.}
					
				Without losing generality, assume $v_i \in \mathcal{V}_f$ and $v_j,~v_k \in \mathcal{V}_l$ so that  $\bm{\Omega}_{ff}$ can be divided into {$\footnotesize
						\bm{\Omega}_{ff} = \left[\begin{array}{c:c}
							\bm{\Omega}_{ff}^{P1} & \bm{\Omega}_{ff}^{P2}\\
							\hdashline
							\left(\bm{\Omega}_{ff}^{P2}\right)^T & \Omega_{ii}
						\end{array}\right] \succ 0$}. 
				Based on Lemma \ref{lemma:PD}, we have $\Omega_{ii} > 0$ and $\bm{\Omega}_{ff}^{P1}  - \dfrac{1}{\Omega_{ii}} \bm{\Omega}_{ff}^{P2} \left(\bm{\Omega}_{ff}^{P2}\right)^T \succ 0 $. After adding a new vertex $v_u$, based on eq.~\eqref{eq:StressMatrixPlus}, the block stress matrix $\bm{\Omega}_{ff}^+ $ associated with the augmented framework $\left(\mathcal{G}^+,\bm{p}^+\right)$ is described as follows.
					\begin{equation}\label{eq:OffPlus1}
						\footnotesize
						\bm{\Omega}_{ff}^+ = \left[\begin{array}{c:cc}
							\bm{\Omega}_{ff}^{P1} & \bm{\Omega}_{ff}^{P2} & \bm{0}_{\left(n_f-1\right) \times 1}\\
							\hdashline
							\left(\bm{\Omega}_{ff}^{P2}\right)^T & \Omega_{ii} + \dfrac{\varpi_{iu}^2}{\varpi_{uu}}& -\varpi_{iu}\\
							\bm{0}_{1 \times \left(n_f-1\right) } &   -\varpi_{iu} & \varpi_{uu}\\
						\end{array}\right].
					\end{equation}
				Since $\varpi_{uu} > 0$ and $\Omega_{ii} + \dfrac{\varpi_{iu}^2}{\varpi_{uu}} - \dfrac{\varpi_{iu}^2}{\varpi_{uu}} = \Omega_{ii} >0$, we get ${\footnotesize \left[\begin{array}{cc}
							\Omega_{ii} + \dfrac{\varpi_{iu}^2}{\varpi_{uu}}& -\varpi_{iu}\\
							-\varpi_{iu} & \varpi_{uu}
						\end{array}\right] \succ 0}$ based on Lemma~\ref{lemma:PD}.
				Hence, the following equation is deduced,
					\begin{equation}
						\footnotesize
						\begin{aligned}
							&\bm{\Omega}_{ff}^{P1} -\left[\begin{array}{cc}
								\bm{\Omega}_{ff}^{P2} & \bm{0}_{\left(n_f-1\right) \times 1}
							\end{array}\right] \cdot \\
							&\left[\begin{array}{cc}
								\Omega_{ii} + \dfrac{\varpi_{iu}^2}{\varpi_{uu}}& -\varpi_{iu}\\
								-\varpi_{iu} & \varpi_{uu}
							\end{array}\right]^{-1}\left[\begin{array}{c}
								\left(\bm{\Omega}_{ff}^{P2}\right)^T \\
								\bm{0 }_{1 \times \left(n_f-1\right) }
							\end{array}\right]\\
								=&\bm{\Omega}_{ff}^{P1}  - \dfrac{1}{\Omega_{ii}} \bm{\Omega}_{ff}^{P2} \left(\bm{\Omega}_{ff}^{P2}\right)^T
								\succ 0.
							\end{aligned}
						\end{equation}
			Thus, $\bm{\Omega}_{ff}^+$ in \eqref{eq:OffPlus1} is positive definite.
						
			\item[\labelitemi]\textbf{Two of the vertices $v_i,v_j,v_k$ belong to $\mathcal{V}_f$, and one belongs to $\mathcal{V}_l$.}
						
			In this case, we apply Lemma \ref{lemma:PD} repeatedly in a similar way to prove that the symmetric matrix $\bm{\Omega}_{ff}^+$ is positive definite. Without losing its generality, assume $v_i,~v_j \in \mathcal{V}_f$ and $v_k \in \mathcal{V}_l$, and $\bm{\Omega}_{ff} $ can be presented as 
						${\footnotesize \bm{\Omega}_{ff} = \left[\begin{array}{c:c}
								\bm{\Omega}_{ff}^{P1} &  \bm{\Omega}_{ff}^{P2} \\
								\hdashline
								\left(\bm{\Omega}_{ff}^{P2}\right)^T & \begin{array}{cc}
									\Omega_{ii} & \Omega_{ij} \\
									\Omega_{ij}  & \Omega_{jj} 
								\end{array}
							\end{array}\right]}$.
			Obviously, due to $\bm{\Omega}_{ff} \succ 0$, we have $\left[\begin{array}{cc}
							\Omega_{ii}   & \Omega_{ij}\\
							\Omega_{ij}   & \Omega_{jj}
						\end{array}\right] \succ 0$ and 
						\begin{equation}\label{eq:ker2}
                        \footnotesize
							\bm{\Omega}_{ff}^{P1}  - \bm{\Omega}_{ff}^{P2} \left[\begin{array}{cc}
								\Omega_{ii}   & \Omega_{ij}\\
								\Omega_{ij}   & \Omega_{jj}
							\end{array}\right]^{-1} \left(\bm{\Omega}_{ff}^{P2}\right)^T  \succ 0.
							\end{equation}
							
			Based on eq.~\eqref{eq:StressMatrixPlus}, the block stress matrix $\bm{\Omega}_{ff}^+$ is established as follows. 
					\begin{equation}\label{eq:ker3}
						\footnotesize
						\bm{\Omega}_{ff}^+ = \left[\begin{array}{c:c}
							\bm{\Omega}_{ff}^{P1} & \begin{array}{cc}
								\bm{\Omega}_{ff}^{P2} &  \bm{0}_{\left(n_f-2\right) \times 1}
							\end{array}\\
							\hdashline
							\begin{array}{c}
								\left(\bm{\Omega}_{ff}^{P2}\right)^T  \\
								\bm{0}_{1 \times \left(n_f-2\right)}
							\end{array} & \bm{\Omega}_{ff}^{P3+}  
						\end{array}\right],
					\end{equation}
			where 
					\begin{equation}
						\footnotesize
						\bm{\Omega}_{ff}^{P3+} \triangleq \left[\begin{array}{ccc}
							\Omega_{ii} + \dfrac{\varpi_{iu}^2}{\varpi_{uu}}& \Omega_{ij} + \dfrac{\varpi_{iu}\varpi_{ju}}{\varpi_{uu}}& -\varpi_{iu}\\
							\Omega_{ij} + \dfrac{\varpi_{iu}\varpi_{ju}}{\varpi_{uu}}& \Omega_{jj} + \dfrac{\varpi_{ju}^2}{\varpi_{uu}}& -\varpi_{ju}\\
							-\varpi_{iu} & -\varpi_{ju} & \varpi_{uu}
						\end{array}\right].\notag
					\end{equation}
							
			Since $\varpi_{uu} >0$ and
					\begin{equation}
						\footnotesize
						\begin{aligned}
							&\left[\begin{array}{cc}
								\Omega_{ii} + \dfrac{\varpi_{iu}^2}{\varpi_{uu}}& \Omega_{ij} + \dfrac{\varpi_{iu}\varpi_{ju}}{\varpi_{uu}} \\
								\Omega_{ij} + \dfrac{\varpi_{iu}\varpi_{ju}}{\varpi_{uu}}& \Omega_{jj} + \dfrac{\varpi_{ju}^2}{\varpi_{uu}}
							\end{array}\right] - \dfrac{1}{\varpi_{uu}}\\
							&\left[\begin{array}{c}
								-\varpi_{iu}\\
								-\varpi_{ju}
							\end{array}\right] \left[\begin{array}{cc}
								-\varpi_{iu} &-\varpi_{ju}
							\end{array}\right] = \left[\begin{array}{cc}
								\Omega_{ii} & \Omega_{ij}\\
								\Omega_{ij} & \Omega_{jj} 
							\end{array}\right] \succ 0,\notag
						\end{aligned}
					\end{equation}
			we have $\bm{\Omega}_{ff}^{P3+} \succ 0$.
			It follows from eq.~\eqref{eq:ker3} and eq.~\eqref{eq:ker2} that 
							\begin{equation}
								\footnotesize
								\begin{aligned}
									\bm{\Omega}_{ff}^{P1} - \left[\begin{array}{cc}
										\bm{\Omega}_{ff}^{P2} & \bm{0}
									\end{array}\right]
									\left(\bm{\Omega}_{ff}^{P3+} \right)^{-1}
									\left[\begin{array}{c}
										\left(\bm{\Omega}_{ff}^{P2}\right)^T \\
										\bm{0}_{1 \times \left(n_f-2\right) }
									\end{array}\right] \succ  0.
								\end{aligned}
							\end{equation}
			By applying Lemma \ref{lemma:PD} again, we have $\bm{\Omega}_{ff}^+\succ0$.
							
			\item[\labelitemi] \textbf{All the vertices $v_i,~v_j,~v_k \in\mathcal{V}_f$.} 
							
			The analysis shares the same idea as in the above two cases. The block stress matrix $\bm{\Omega}_{ff} $ is positive definite and can be divided into
							\begin{equation}
                            \footnotesize
								\notag
								\bm{\Omega}_{ff} = \left[\begin{array}{c:c}
									\bm{\Omega}_{ff}^{P1} &  \bm{\Omega}_{ff}^{P2}\\
									\hdashline
									\left(\bm{\Omega}_{ff}^{P2}\right)^T & \underbrace{\begin{array}{ccc}
										\Omega_{ii}   & \Omega_{ij} & \Omega_{ik}\\
										\Omega_{ij}   & \Omega_{jj} & \Omega_{jk}\\
										\Omega_{ik}   & \Omega_{jk} & \Omega_{kk}
									\end{array}}_{=\bm{\Omega}_{ff}^{P3} }
								\end{array}\right] \succ 0.
							\end{equation}
			Thus, we have 
							\begin{equation} \label{eq:NI4_Omegaff}
                            \footnotesize
								\begin{aligned}
									&\bm{\Omega}_{ff}^{P1} - \bm{\Omega}_{ff}^{P2} \left(\bm{\Omega}_{ff}^{P3}\right)^{-1} \left(\bm{\Omega}_{ff}^{P2}\right)^T \succ 0.
								\end{aligned}
							\end{equation}
							
			After adding a new vertex $v_u$ as a follower, we have 
							\begin{equation}
								\notag
								\footnotesize
								\begin{aligned}
									\bm{\Omega}_{ff}^+ &= \left[\begin{array}{c:c}
										\bm{\Omega}_{ff}^{P1} &  \begin{array}{cc}
											\bm{\Omega}_{ff}^{P2} &   \bm{0}_{\left(n_f-3\right) \times 1}
										\end{array}\\
										\hdashline
										\begin{array}{c}
											\left(\bm{\Omega}_{ff}^{P2}\right)^T  \\
											\bm{0}_{1 \times \left(n_f-3\right)}
										\end{array} & \bm{\Omega}_{ff}^{P3+} 
									\end{array}\right], \\
								\end{aligned}
							\end{equation}
			where 
				$ \bm{\Omega}_{ff}^{P3+} = \left[\begin{array}{cc}
					\bm{\Omega}_{ff}^{P3}   & \bm{0}\\
					\bm{0}  & 0
				\end{array}\right] + \bm{\Omega}_{u}$, and $ \bm{\Omega}_{u}$ is shown as below. 
				\begin{equation}\label{eq:4_OmegaffPlus}
                \footnotesize
					\bm{\Omega}_{u}= \left[\begin{array}{cccc}
						\dfrac{\varpi_{iu}^2}{\varpi_{uu}} &  \dfrac{\varpi_{iu}\varpi_{ju}}{\varpi_{uu}} &  \dfrac{\varpi_{iu}\varpi_{ku}}{\varpi_{uu}} & -\varpi_{iu}\\
						\dfrac{\varpi_{iu}\varpi_{ju}}{\varpi_{uu}}&  \dfrac{\varpi_{ju}^2}{\varpi_{uu}} &  \dfrac{\varpi_{ju} \varpi_{ku}}{\varpi_{uu}} & -\varpi_{ku}\\
						\dfrac{\varpi_{iu}\varpi_{ku}}{\Omega_{uu}}&  \dfrac{\varpi_{ju}\varpi_{ku}}{\varpi_{uu}} &  \dfrac{\varpi_{ku}^2}{\varpi_{uu}} & -\varpi_{ju}\\
						-\varpi_{iu} & -\varpi_{ju} & -\varpi_{ku} & \varpi_{uu}
					\end{array}\right] 
				\end{equation}
								
			Since $\varpi_{uu}>0$ and 
				\begin{equation}
					\footnotesize
					\begin{aligned}
						&\left[\begin{array}{cccc}
							\Omega_{ii} + \dfrac{\varpi_{iu}^2}{\varpi_{uu}}& \Omega_{ij} + \dfrac{\varpi_{iu}\varpi_{ju}}{\varpi_{uu}} & \Omega_{ik} + \dfrac{\varpi_{iu}\varpi_{ku}}{\varpi_{uu}} \\
							\Omega_{ij} + \dfrac{\varpi_{iu}\varpi_{ju}}{\varpi_{uu}}& \Omega_{jj} + \dfrac{\varpi_{ju}^2}{\varpi_{uu}} & \Omega_{jk} + \dfrac{\varpi_{ju} \varpi_{ku}}{\varpi_{uu}} \\
							\Omega_{ik} + \dfrac{\varpi_{iu}\varpi_{ku}}{\varpi_{uu}}& \Omega_{jk} + \dfrac{\varpi_{ju}\varpi_{ku}}{\varpi_{uu}} & \Omega_{kk} + \dfrac{\varpi_{ku}^2}{\varpi_{uu}} 
						\end{array}\right]\\
						& - \left[\begin{array}{c}
							-\varpi_{iu}\\
							-\varpi_{ju}\\
							-\varpi_{ku}
						\end{array}\right]\dfrac{1}{\varpi_{uu}} \left[\begin{array}{ccc}
							-\varpi_{iu} &-\varpi_{ju} & -\varpi_{ku}
						\end{array}\right] \\
						=& \bm{\Omega}_{ff}^{P3} \succ {0},
					\end{aligned}
				\end{equation}
			we get $\bm{\Omega}_{ff}^{P3+}\succ {0} $. By applying Lemma \ref{lemma:PD} again, we can prove the matrix $\bm{\Omega}_{ff}^+$ is positive definite, as shown below.
								
				\begin{equation}
                \footnotesize
					\begin{aligned}
						&\bm{\Omega}_{ff}^{P1} - \left[\begin{array}{cc}
							\bm{\Omega}_{ff}^{P2} & \bm{0}
						\end{array}\right] \left(\bm{\Omega}_{ff}^{P3+}\right)^{-1}
						\left[\begin{array}{c}
							\left(\bm{\Omega}_{ff}^{B2}\right)^T \\
							\bm{0}
						\end{array}\right] \\
						=&\bm{\Omega}_{ff}^{P1} - \bm{\Omega}_{ff}^{P2}
						\left(\bm{\Omega}_{ff}^{P3}\right)^{-1} \left(\bm{\Omega}_{ff}^{P2}\right)^T  \overset{\eqref{eq:NI4_Omegaff}}{\succ} 0.
					\end{aligned}
				\end{equation}
			According to Lemma \ref{lemma:PD}, it is proved that $\bm{\Omega}_{ff}^+$ in \eqref{eq:4_OmegaffPlus} is also positive definite.
			\end{itemize}

			According to Lemma \ref{lemma:NullStress}, the stress matrix $\bm{\Omega}^+$ associated with the augmented framework $\left(\mathcal{G}^+,\bm{p}^+\right)$ follows that
				\begin{equation}
                \footnotesize
					\left(\bm{\Omega}^+ \otimes \mathbf{I}_d\right) \bm{p}^+ = \left[\begin{array}{cc}
						\bar{\bm{\Omega}}_{ll}^+ & \bar{\bm{\Omega}}_{lf}^+ \\
						\bar{\bm{\Omega}}_{fl}^+ & \bar{\bm{\Omega}}_{ff}^+ 
					\end{array}\right]\left[\begin{array}{c}
						\bm{p}_l^+ \\
						\bm{p}_f^+ 
					\end{array}\right] = \bm{0}.
				\end{equation}
			Accordingly, we have $\bar{\bm{\Omega}}_{fl}^+ \bm{p}_l^+ + \bar{\bm{\Omega}}_{ff}^+ \bm{p}_f^+ =\bm{0} $. We clarify that the augmented framework $\left(\mathcal{G}^+,\bm{p}^+\right)$ is affinely localized by the selected leaders since $\bm{\Omega}_{ff}^+ \succ 0$, which leads to $ \bm{p}_f^+ = - \left( \bar{\bm{\Omega}}_{ff}^+ \right)^{-1} \bar{\bm{\Omega}}_{fl}^+ \bm{p}_l^+ $.
							
			With all the discussions above, the augmented framework $\left(\mathcal{G}^+,\bm{p}^+\right)$ obtained from adding a new $v_u$ and three weighted edges meets the requirements for universal rigidity and affine localizability.
			\end{proof}
			
			\bibliographystyle{IEEEtran}
			\bibliography{MyRef}

			\begin{IEEEbiography}[{\includegraphics[width=1in,height=1.25in,clip,keepaspectratio]{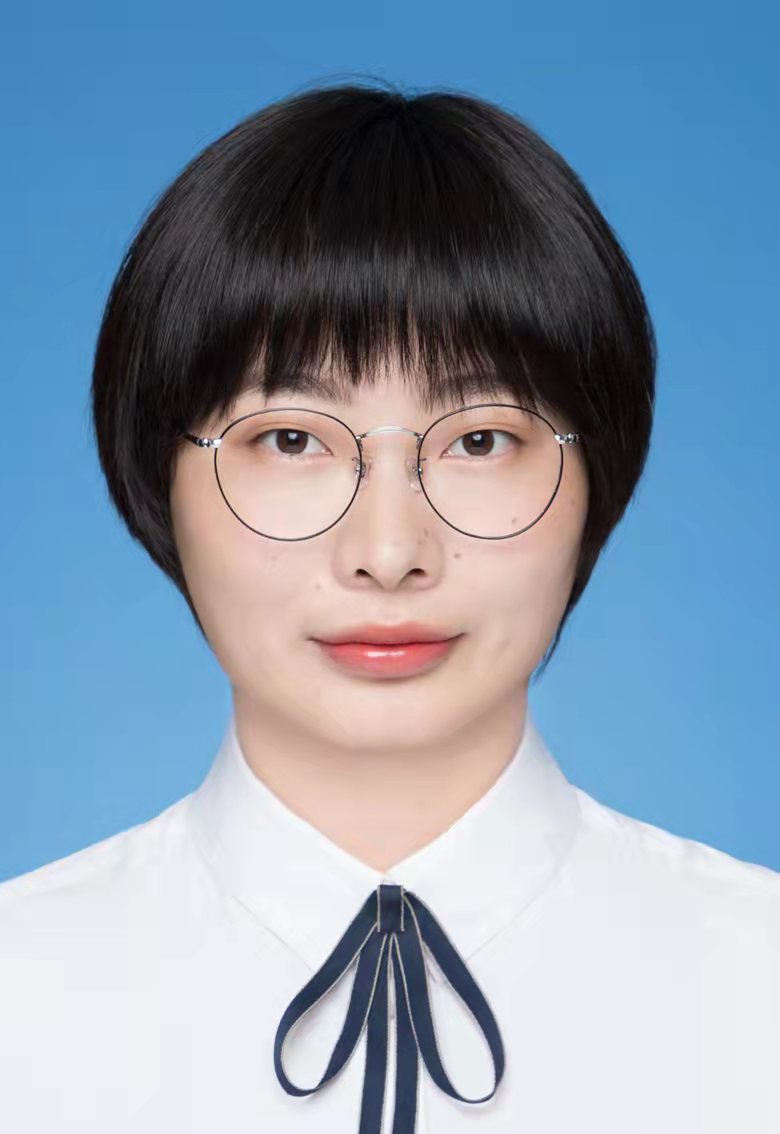}}]{Huiming Li} received her B.S. degree from Northeastern University, Shenyang, China, in 2018, and M.S. degree from National University of Defense Technology (NUDT), Changsha, China, in 2020, where she is currently pursuing her Ph.D. degree. Her research interests include coordinated control and unmanned aerial vehicles.
			\end{IEEEbiography}

			\begin{IEEEbiography}[{\includegraphics[width=1in,height=1.25in,clip,keepaspectratio]{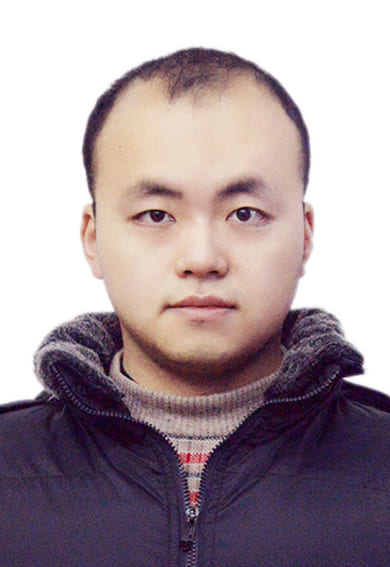}}]{Hao Chen} received the Ph.D degrees in control science and engineering from National University of Defense Technology (NUDT), Changsha, China, in 2020. 
				
				He is currently an Associate Professor in College of Intelligence Science and Technology, NUDT. He was a visiting student with the Technion-Israel Institute of Technology, supported by the China Scholarship Council from 2017 to 2018. His research interests include coordinated control and graph theory.
			\end{IEEEbiography}
			
			\begin{IEEEbiography}[{\includegraphics[width=1in,height=1.25in,clip,keepaspectratio]{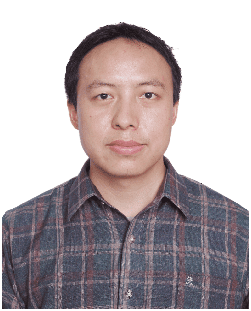}}]{Xiangke Wang} (Senior Member, IEEE) received the B.S., M.S., and Ph.D. degrees in Control Science and Engineering from National University of Defense Technology, China, in 2004, 2006, and 2012, respectively. Since 2012, he has been with the College of Intelligence Science and Technology, National University of Defense Technology, where he is currently a full professor. He was a visiting student at the Research School of Engineering, Australian National University from 2009 to 2011. He is a senior member of IEEE and is supported by the Hunan Outstanding Youth Award Program.  His current research interests include the control of multi-agent systems and its applications on unmanned aerial vehicles. He has authored or coauthored 5 books and more than 200 papers in refereed journals or international conferences, including IEEE Transactions/Letters, CDC, IFAC, ICRA. etc.
			\end{IEEEbiography}
			
			\begin{IEEEbiography}[{\includegraphics[width=1in,height=1.25in,clip,keepaspectratio]{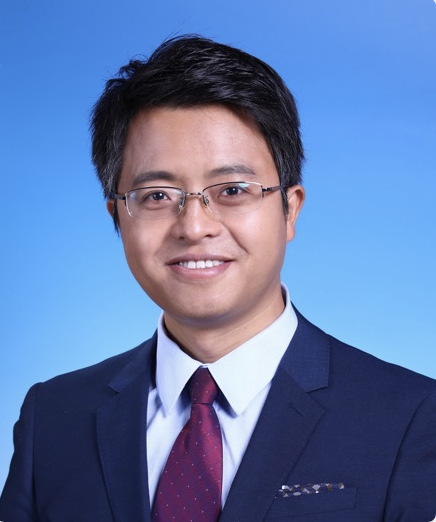}}]{Zhongkui Li} (M’11, SM’21) received the B.S. degree in space engineering from the National University of Defense Technology, China, in 2005, and his Ph.D. degree in dynamics and control from Peking University, China, in 2010. Since 2013, Dr. Li has been with the Department of Mechanics and Engineering Science, College of Engineering, Peking University, China, where he is currently a tenured Associate Professor. His current research interests include cooperative control and planning of multi-agent systems. 
				
			Dr. Li was the recipient of the State Natural Science Award of China in 2015, the Natural Science Award of the Ministry of Education of China in 2022 and 2011, and the National Excellent Doctoral Thesis Award of China in 2012. His coauthored papers received the IET Control Theory \& Applications Premium Award in 2013 and the Best Paper Award of Journal of Systems Science \& Complexity in 2012. He serves as an Associate Editor of IEEE Transactions on Automatic Control, and several other journals.
			\end{IEEEbiography}
			
			\begin{IEEEbiography}[{\includegraphics[width=1in,height=1.25in,clip,keepaspectratio]{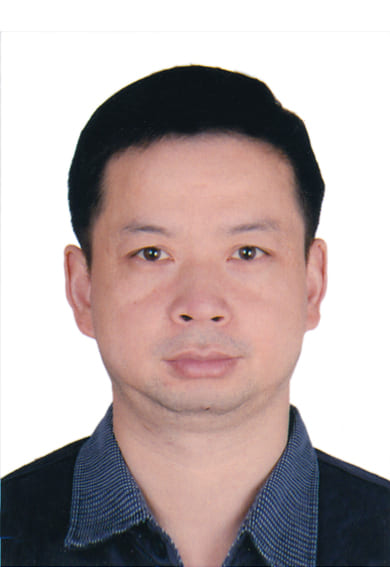}}]{Lincheng Shen} received the Ph.D. degrees in automatic control from the National University of Defense Technology, Changsha, China, in 1994.
				
			In 1989, he joined the Department of Automatic Control, NUDT, where he is currently a Full Professor and serves as the Dean of the Graduate School. His research interests include unmanned aerial vehicles, swarm robotics, and artificial intelligence. Dr. Shen has been serving as an Editorial Board Member of the Journal of Bionic Engineering since 2007.
			\end{IEEEbiography}
			
			\clearpage


		\end{document}